\documentclass[journal, romanappendices, final]{IEEEtran}
\usepackage{latexsym,amsfonts,amssymb,amsthm, setspace, verbatim,cite,mathrsfs,graphicx,bm,stfloats, hyperref, tikz,xcolor, bm, algorithmic, algorithm, enumerate, paralist}
\usepackage[cmex10]{amsmath}
\newcommand{\be}{\begin{equation}}
\newcommand{\ee}{\end{equation}}
\newcommand{\ben}{\begin{equation*}}
\newcommand{\een}{\end{equation*}}
\newcommand{\mc}{\mathcal}

\newcommand{\e}{\epsilon}
\newcommand{\mcb}{\mathcal{B}_{M,L}}

\newcommand{\abs}[1]{\lvert#1\rvert}
\newcommand{\norm}[1]{\lVert#1\rVert}
\newcommand{\expec}{\mathbb{E}}
\newcommand{\bks}{\backslash}
\newcommand{\secl}{\text{sec}(i)}
\newcommand{\mscrs}{\mathscr{S}}

\newcommand{\snr}{\textsf{snr}}
\newtheorem{lem}{Lemma}

\newtheorem{defi}{Definition}
\newtheorem{thm}{Theorem}
\newtheorem{fact}{Fact}

\newtheorem{prop}{Proposition}

\begin{document}
\title{Capacity-achieving Sparse Superposition Codes via \\ Approximate Message Passing Decoding}
\author{Cynthia Rush,~\IEEEmembership{Member,~IEEE,}
Adam Greig,~\IEEEmembership{Student Member,~IEEE,}
and Ramji Venkataramanan,~\IEEEmembership{Member,~IEEE}
\thanks{This paper was presented in part at the 2015 IEEE International Symposium on Information Theory.}%
\thanks{C.~Rush was with the Department of Statistics, Yale University. She is now with the Department of Statistics, New York, NY 10027, Columbia University, USA (e-mail: cynthia.rush@columbia.edu).}%
\thanks{A.~Greig and R.~Venkataramanan are with Department of Engineering, University of Cambridge, Cambridge CB2 1PZ, UK (e-mail: \{ag611, rv285\}@cam.ac.uk).}
%
}
\maketitle

\begin{abstract}
Sparse superposition codes were recently introduced by Barron and Joseph for reliable communication over the AWGN channel at rates approaching the channel capacity. The codebook is defined in terms of a Gaussian design matrix, and codewords are sparse linear combinations of columns of the matrix. In this paper, we propose an approximate message passing decoder for sparse superposition codes, whose decoding complexity scales linearly with the size of the design matrix. The performance of the decoder is rigorously analyzed and it is shown to asymptotically achieve the AWGN capacity with an appropriate power allocation. Simulation results are provided to demonstrate the performance of the decoder at finite blocklengths. We introduce a power allocation scheme to improve the empirical performance, and demonstrate how the decoding complexity can be significantly reduced by using Hadamard design matrices.
\end{abstract}


\section{Introduction}
\label{sec:intro}
\IEEEPARstart{T}his paper considers the problem of constructing low-complexity,  capacity-achieving codes for the memoryless additive white Gaussian noise (AWGN) channel.  The channel generates output $y$ from input $x$ according to
\be
y = x + w,
\label{eq:awgn}
\ee
where the noise $w$ is a  Gaussian  random variable with zero mean and variance $\sigma^2$. There is an average power constraint $P$ on the input $x$: if $x_1, \ldots, x_n$ are transmitted over $n$ uses of the channel, then we require that $\tfrac{1}{n}\sum_{i=1}^n x_i^2 \leq P$. The signal-to-noise ratio $\tfrac{P}{\sigma^2}$  is denoted by \snr. The goal is to construct  codes with computationally efficient encoding and decoding, whose rates approach the channel capacity given by
\be \mc{C}: =\tfrac{1}{2} \log (1 + \snr). \ee

 Sparse superposition codes, also called Sparse Regression Codes (SPARCs),  were  recently introduced by Barron and Joseph \cite{AntonyML,AntonyFast} for communication over the channel in \eqref{eq:awgn}. They proposed an efficient decoding algorithm called `adaptive successive decoding', and showed that for any fixed rate $R < \mc{C}$, the probability of decoding error decays to zero exponentially in $\tfrac{n}{\log n}$, where $n$ is the block length of the code. 
 Despite the strong theoretical performance guarantees,  the rates achieved by this decoder for practical block lengths are significantly less than $\mc{C}$. Subsequently, a soft-decision iterative decoder was proposed by Cho and Barron \cite{BarronC12,choThesis}, with theoretical guarantees  similar to the earlier decoder in \cite{AntonyFast} but improved empirical performance for finite block lengths.

In this paper, we propose an approximate message passing (AMP) decoder for SPARCs. We analyze its performance and prove  that  the probability of decoding error goes to zero with growing block length for all fixed rates $R < \mc{C}$. The decoding complexity is proportional to the size of the design matrix defining the code, which is a low order polynomial in $n$.  

\subsection{Approximate Message Passing (AMP)}
``Approximate message passing" refers to a class of algorithms \cite{DonMalMont09,DonMalMontITW, BayMont11,MontChap11, bayMontLASSO, krz12,Rangan11,DonSpatialC13}
that are Gaussian or quadratic approximations of  loopy belief propagation algorithms (e.g., min-sum, sum-product) on dense factor graphs. AMP has proved  particularly effective for the problem of reconstructing sparse signals from a small number of noisy linear measurements. This problem, commonly referred to as compressed sensing \cite{CSspissue08}, is described by the measurement model
\be
y = A\beta + w.
\label{eq:cs_model}
\ee
 Here ${A}$ is an $n \times N$ measurement matrix with $n <N$, $\beta \in \mathbb{R}^{N}$ is a sparse vector to be estimated from the observed vector $y \in \mathbb{R}^n$, and $w \in \mathbb{R}^n$ is the measurement noise. One popular class of algorithms to reconstruct $\beta$ is $\ell_1$-norm based convex optimization, e.g. \cite{CandesTaoLP,DonohoCS,TroppRelax}. Though these algorithms have strong theoretical guarantees and excellent empirical performance, the computational cost  makes it challenging to implement the convex optimization procedures for problems where $N$ is large.  A fast AMP reconstruction algorithm for the model in \eqref{eq:cs_model} was proposed in \cite{DonMalMont09}. Its empirical performance  (for a large class of measurement matrices)  was found to be  similar to convex optimization based methods at significantly lower computational cost. 

The factor graph corresponding to the model in \eqref{eq:cs_model} is dense,  hence it is infeasible to implement message passing algorithms in which the messages are complicated real-valued functions. AMP circumvents this difficulty by passing only scalar parameters corresponding to these functions. For example, the scalars could be the mean and the variance if the functions are posterior distributions.   The references\cite{DonMalMontITW,MontChap11,krz12,Rangan11}  describe how various flavors of AMP for the model  in \eqref{eq:cs_model} can be obtained by approximating the standard message passing equations. These approximations reduce the message passing equations to a set of simple rules for computing successive estimates of $\beta$. 

In \cite{DonMalMont09}, it was demonstrated via numerical experiments that the mean-squared reconstruction error of these estimates of $\beta$ could be tracked by a simple scalar iteration called \emph{state evolution}.   In \cite{BayMont11}, it was rigorously proved that the state evolution is accurate in the large system limit\footnote{The large system limit considered in \cite{BayMont11} lets $n,N \to \infty$ with $n/N$  held constant.} for measurement matrices $A$ with i.i.d.\ Gaussian entries.

In addition to  compressed sensing,  AMP has also been applied to a variety of related problems, e.g. \cite{KamilovRFU14, schAMP1,schAMP2}.  We will not attempt a complete survey of the growing literature on AMP; the reader is referred to \cite{Rangan11,DonSpatialC13} for comprehensive lists of related work.

\subsection{Contributions of the Paper}

\begin{compactitem}
\item We propose an AMP decoder for sparse regression codes, which is derived via a first-order approximation of a min-sum-like message passing algorithm.

\item The main result of the paper is Theorem \ref{thm:main_amp_perf}, in which we rigorously show that  the probability of decoding error goes to zero  as the  block length tends to infinity, for all rates $R < \mc{C}$. 

\item The performance of the decoder for finite block lengths  is demonstrated via simulation results. We introduce a power allocation scheme that significantly improves the empirical performance for rates not close to $\mc{C}$. We also show how the decoding complexity can be reduced by using Hadamard-based design matrices.
\end{compactitem}

To prove our main result, we use the framework of Bayati and Montanari \cite{BayMont11,  bayMontLASSO}, who in turn built on techniques introduced by Bolthausen \cite{Bolt12}. However, we remark that the analysis of the proposed algorithm does not follow directly from the results in \cite{BayMont11, JavMonState13}. The main reason for this is that the \emph{undersampling ratio} $n/N$ in our setting goes to zero in the large system limit, whereas previous rigorous analyses of AMP consider the case where the undersampling ratio is a constant. This point, as well as  other differences from the analysis in \cite{BayMont11, bayMontLASSO},  is discussed further in Section \ref{subsec:lem_comments}.

\subsection{Related work on communication with SPARCs}
The adaptive successive decoder of  Joseph-Barron \cite{AntonyFast} and the iterative soft-decision decoder of Cho-Barron \cite{BarronC12,choThesis}  both have probability of error that decays as $n/\log n$ for any fixed rate $R < \mc{C}$, but the latter has better empirical performance. Theorem  \ref{thm:main_amp_perf} shows that the probability of error for the AMP decoder goes to zero for all $R  < \mc{C}$, but does not give a rate of decay; hence we cannot theoretically compare its  performance with the Cho-Barron decoder in \cite{choThesis}. We can, however, compare the two decoders qualitatively.

Both the AMP and the Cho-Barron decoder generate a succession of estimates $\beta^1, \beta^2, \ldots$ for the message vector $\beta$ based on test statistics $s^0, s^1, \ldots$, respectively.  At step $t$, the Barron-Cho decoder generates statistic $s^t$ based on an orthonormalization of the observed vector $y$ and the previous `fits' $A \beta^1, \ldots, A\beta^{t}$. In contrast, the test statistic in the AMP decoder is based on a modified version of the residual $(y-A\beta^t)$.  Despite being generated in very different  ways, the test statistics of the AMP and Cho-Barron decoders have a similar structure: they are asymptotically equivalent to an observation of $\beta$ corrupted by additive Gaussian noise whose variance decreases with $t$. However,  the AMP statistic is faster to compute in each step, which makes it feasible to implement the decoder for larger block lengths. 

An approximate message passing decoder for sparse superposition codes was recently proposed by Barbier and Krzakala in \cite{barbKrzISIT14}. This decoder has different update rules from the AMP proposed here. A replica-based analysis of the decoder in \cite{barbKrzISIT14} suggested it could not achieve rates beyond a threshold which was strictly smaller than $\mc{C}$.  Subsequently, Barbier et al \cite{BarbSchKrz15} reported empirical results which show that the performance of the decoder in \cite{barbKrzISIT14} can be improved by using spatially coupled Hadamard matrices to define the code. 

Finally, we mention that bit-interleaved coded modulation \cite{bicmAlbert} is a technique widely used for communication over AWGN channels. Some alternative approaches to designing high-rate codes for the AWGN channel are 
  low-density lattice codes \cite{LDLC08} and the recently proposed polar lattices \cite{YanLLW14}. 
  
\subsection{Paper outline  and Notation}
The paper is organized as follows. The SPARC construction is described in Section \ref{sec:sparc}. We describe the AMP channel decoder in Section \ref{sec:amp_channel_decoder},   and provide some intuition about its iterations. We also show how the decoder can be derived as a first-order approximation to a min-sum-like message passing algorithm. Section \ref{sec:AMP_perf} contains the main result,  which  characterizes the performance of the AMP decoder for any rate $R < \mc{C}$ in the large system limit.  In Section \ref{sec:ex_results}, we present simulation results to demonstrate the performance of the decoder at finite block lengths. Section \ref{sec:amp_proof} contains the proof of the main result, and the proof of a key technical  lemma is given in Section \ref{sec:lem1_proof}.

\emph{Notation}:  The $\ell_2$-norm of vector $x$  is denoted by $\norm{x}$. The transpose of a matrix $B$ is denoted by $B^*$.  The Gaussian distribution with mean $\mu$ and variance $\sigma^2$ is denoted by $\mc{N}(\mu,\sigma^2)$. For any positive integer $m$, $[m]$ denotes the set
 $\{1, \dots, m \}$. The indicator function of an event $\mc{A}$ is denoted by $\mathbf{1}(\mc{A})$. $f(x)=o(g(x))$ means $\lim_{x \to \infty} f(x)/g(x) =0$; $f(x)=\Theta(g(x))$ means $f(x)/g(x)$ asymptotically lies in an interval $[\kappa_1,\kappa_2]$ for some constants $\kappa_1,\kappa_2>0$.  $\log$  and $\ln$ are used to denote logarithms with base $2$ and base $e$, respectively. Rate is measured in bits.
 
 \section{The Sparse Regression Codebook} \label{sec:sparc}
\begin{figure}[t]
\centering
\includegraphics[width=3.5in]{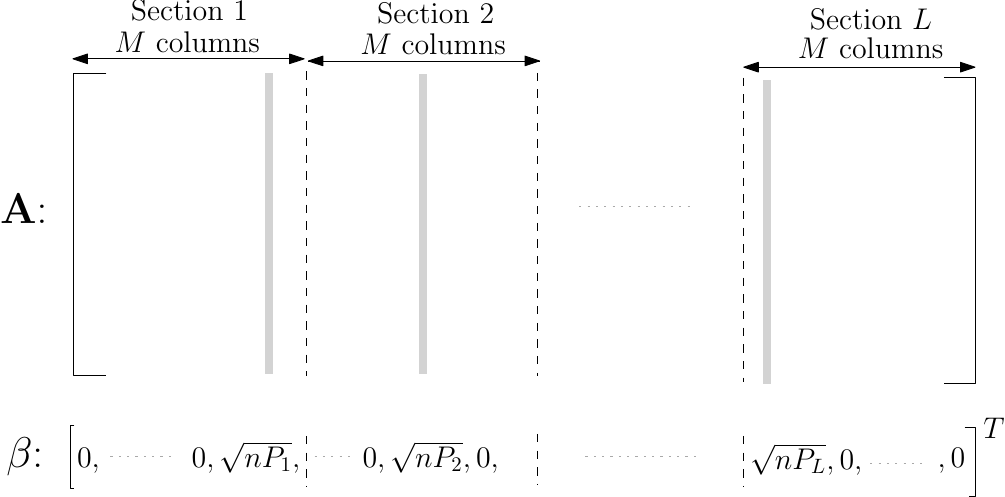}
\caption{\small{$A$ is an $n \times ML$ matrix and $\beta$ is a $ML \times 1$ vector. The positions of the non-zeros in $\beta$  correspond to the gray columns of $A$ which combine to form the codeword $A\beta$.}}
\vspace{-7pt}
\label{fig:sparserd}
\end{figure}

A sparse regression code  is defined in terms of a dictionary or design matrix $A$ of dimension $n \times ML$,  whose entries are i.i.d.\ $\mathcal{N}(0,\tfrac{1}{n})$. Here $n$ is the block length, and $M, L$ are integers whose values are specified below in terms of $n$ and the rate $R$.  As shown in Fig. \ref{fig:sparserd}, one can think of the matrix $A$ being composed of $L$ sections with $M$ columns each. Each codeword is a linear combination of $L$ columns, with one column from each section.
Formally, a codeword can be expressed as  $A \beta$, where $\beta$ is  an $ML \times 1$ vector $(\beta_1, \ldots, \beta_{ML})$ with the following property:  there is exactly one non-zero $\beta_j$ for  $1 \leq j \leq M$, one non-zero $\beta_j$ for $M+1 \leq j \leq 2M$, and so forth.  The non-zero value of $\beta$ in section $\ell \in [L]$  is set to $\sqrt{n P_\ell}$, where the positive constants $P_\ell$ satisfy $\sum_{\ell=1}^L P_\ell = P$.
Denote the set of all $\beta$'s that satisfy this property by $\mcb(P_1, \ldots,P_L)$.

Since each of the $L$ sections contains $M$ columns, the total number of codewords is $M^L$. To obtain a communication rate of $R$ bits/sample, we need
\be
M^L = 2^{nR} \quad \text{ or } \quad L \log M = nR.
\label{eq:ml_nR}
\ee
 There are several choices for the pair $(M,L)$ which satisfy \eqref{eq:ml_nR}. For example, $L=1$ and $M=2^{nR}$ recovers the Shannon-style random codebook in which the number of columns in $A$ is $2^{nR}$. For our constructions, we will choose $M$ equal to $L^{\textsf{a}}$, for some constant $\textsf{a} >0$. In this case, \eqref{eq:ml_nR} becomes
 \be
 \textsf{a} L \log L = nR.
 \label{eq:llogl_nR}
 \ee
  Thus $L= \Theta(\tfrac{n}{\log n})$, and the size of the design matrix $A$ (given by $n \times ML = n \times L^{\textsf{a}+1}$) now grows as $n^{2+\mathsf{a}}/(\log n)^{\mathsf{a}+1}$.

\emph{Encoding}: The encoder splits its stream of input bits into segments of $\log M$ bits each. A length $ML$  message  vector $\beta_0$ is indexed by $L$ such segments---the decimal equivalent of segment $\ell$ determines the position of the non-zero coefficient in section $\ell$ of $\beta_0$. The input codeword is then computed as $x=A\beta_0$; note that computing $x$ simply involves adding  $L$ columns of $A$, weighted by the appropriate coefficients.

  \emph{Power Allocation}: The power allocation $\{P_\ell \}_{\ell =1}^L$,  plays an important role in determining the performance of the decoder. We will consider allocations where $P_\ell=\Theta(\tfrac{1}{L})$. Two examples are:
  \begin{compactitem}
 \item Flat power allocation across sections: $P_\ell=\tfrac{P}{L}$, $\ell \in [L]$.
  \item Exponentially decaying power allocation:  Fix parameter  $ \kappa  >0$. Then
  $ P_\ell  \propto  2^{-\kappa \ell/L}, \  \ell \in [L]$.
  \end{compactitem}
  We use the exponentially decaying allocation with $\kappa=2\mc{C}$ for Theorem \ref{thm:main_amp_perf}. In Section \ref{sec:ex_results}, we discuss other power allocations,  and find that an appropriate combination of exponential and flat allocations yields good decoding performance at finite block lengths.   

Both the design matrix $A$ and the power allocation $\{P_\ell\}$  are known to the encoder and the decoder before communication begins.

\emph{Some more notation}:  In the analysis, we will treat the message as a random vector $\beta$, which is uniformly distributed over 
$\mcb(P_1,\ldots,P_L)$, the set of length $ML$ vectors that have a single non-zero entry $\sqrt{nP_\ell}$ in section $\ell$, for $\ell \in [L]$. We will denote the true message vector by $\beta_0$; $\beta_0$ should be understood as a \emph{realization} of the random vector $\beta$. 

We will use indices $i, j$ to denote specific entries of $\beta$, while the index $\ell$ will be used to denote the entire section $\ell$ of $\beta$. Thus $\beta_i, \beta_j$ are scalars, while $\beta_\ell$ is a length $M$ vector. We also set $N = ML$.  

The performance of the SPARC decoder will be characterized in the limit as the dictionary size  goes to $\infty$. We write $\lim x$ to denote the limit of the quantity $x$ as the SPARC parameters $n, L, M \to \infty$ simultaneously, according to
$M=L^{\textsf{a}} \quad  \text{and }  \ \textsf{a} L \log L = nR$.


\section{The AMP Channel Decoder} \label{sec:amp_channel_decoder}

Given the received vector $y= A \beta_0 + w$, the AMP decoder generates successive estimates of the message vector, denoted by $\{ \beta^t \}$, where $\beta^t \in \mathbb{R}^N$ for $t=1,2,\ldots$. Set $\beta^0=0$, the all-zeros vector. For $t=0,1,\ldots$, compute
\begin{align}
z^t & =y - A\beta^t + \frac{z^{t-1}}{\tau^2_{t-1}}\left( P - \frac{\norm{\beta^t}^2}{n} \right), \label{eq:amp1}\\
\beta^{t+1}_i & = \eta^t_{i}( \beta^t + A^*z^t), \quad  \text{ for } i=1,\ldots,N=ML, \label{eq:amp2}
\end{align}
where  quantities with negative indices are set equal to zero. The constants $\{ \tau_{t} \}$, and the estimation functions $\eta^t_i(\cdot)$ are defined as follows for $t=0,1,\ldots$.

Define
\begin{align}
\tau^2_{0} & =\sigma^2+P,  \qquad
\tau^2_{t+1}   =  \sigma^2 + P(1-x_{t+1}), \quad t \geq 0,
\label{eq:taut_def}
\end{align}
where
\be
x_{t+1} = \sum_{\ell=1}^{L} \frac{P_\ell}{P} \, \expec \left[
\frac{e^{\frac{\sqrt{n P_\ell}}{\tau_t} \, (U^{\ell}_1  + \frac{\sqrt{n P_\ell}}{\tau_t})} }
{e^{\frac{\sqrt{n P_\ell}}{\tau_t} \, (U^{\ell}_1  + \frac{\sqrt{n P_\ell}}{\tau_t})} + \sum_{j=2}^M e^{\frac{\sqrt{n P_\ell}}{\tau_t}U^{\ell}_j} } \right].
\label{eq:xt_def}
\ee
In \eqref{eq:xt_def}, $\{ U^\ell_j\}$ are i.i.d.\ $\mc{N}(0,1)$ random variables for $j\in [M], \ \ell \in [L]$.

The notation $j \in \text{sec}(\ell)$ will be used as shorthand for ``index $j$ in section $\ell$", i.e., $j \in \{(\ell-1)M +1, \ldots, \ell M\}$ where $\ell \in [L]$. 
For $i \in [N]$ such that $i \in \text{sec}(\ell)$,  define
\be
\eta^t_i(s) = \sqrt{nP_\ell}\,  \frac{e^{ {s_i \sqrt{n P_\ell}}/ {\tau^2_t}}}
{\sum_{j \in \text{sec}(\ell)} \, e^{ {s_j \sqrt{n P_\ell}}/ {\tau^2_t}} }. 
\label{eq:eta_def}
\ee
 Notice that $\eta^t_i(s)$ depends on all the components of $s$ in the \emph{section} containing $i$.  For brevity, the argument of $\eta^t_i$ in \eqref{eq:amp2} is written as $A^* z^t +   \beta^t$, with the understanding that only the components in the section containing $i$ play a role in computing $\eta^t_i$.

 Before running the AMP decoder,  the constants $\{ \tau_t \}$ must be iteratively computed using \eqref{eq:taut_def} and \eqref{eq:xt_def}. This is an offline computation: for  given values of $M,L,n$, the expectations in \eqref{eq:xt_def} can be computed via Monte Carlo simulation.  The relation \eqref{eq:taut_def}, which describes how $\tau_{t+1}$ is obtained from  $\tau_{t}$,  is called \emph{state evolution}, following the terminology in \cite{DonMalMont09,BayMont11}.  In Section \ref{sec:AMP_perf} (Lemmas \ref{lem:conv_expec} and \ref{lem:lim_xt_taut}), we derive closed form expressions for  $x_{t}$ and $\tau^2_t$ as $n \to \infty$ for each $t>0$, which we denote by $\bar{x}^{t}$ and $\bar{\tau}^2_t$.   In Section \ref{sec:AMP_perf}, it is shown that for an appropriately chosen power allocation, $\bar{x}^{t}$ strictly increases with $t$ until it reaches $1$ in a finite number of steps $T^*$ for any fixed $R < \mc{C}$. (For the exponentially decaying allocation used in Theorem \ref{thm:main_amp_perf},  $T^*= \left\lceil \frac{2 \mc{C}}{\log(\mc{C}/R)} \right\rceil$, as given in \eqref{eq:Tstar_def}.)  
  
The AMP decoder is run for $T^*$ steps, and  iteratively computes codeword estimates $\beta^1, \ldots, \beta^{T^*}$ using \eqref{eq:amp1} and \eqref{eq:amp2}.  Finally,  in each section $\ell$ of $\beta^{T^*}$, set the maximum value to $\sqrt{nP_\ell}$ and remaining entries to $0$ to obtain the decoded message $\hat{\beta}$.  Our main theoretical result  (Theorem \ref{thm:main_amp_perf}) characterizes the  performance of the AMP decoder run for $T^*$ steps with the asymptotic values $\{ \bar{\tau}^2_t \}_{t=0, \ldots, T^*}$

\subsection{The Test Statistics $\beta^t + A^*z^t$} \label{subsec:test_stat}

To understand the decoder let us first focus on \eqref{eq:amp2}, in which $\beta^{t+1}$ is generated from the test statistic
\be
s^t:= \beta^t + A^*z^t.
\ee
 The AMP update step \eqref{eq:amp2}  is underpinned by the following key property of  the {test statistic}:
 \emph{$s^t$ is asymptotically (as $n \to \infty$) distributed as  $\beta + \bar{\tau_t} Z$, where $\bar{\tau}_t$ is the limit of $\tau_t$, and $Z$ is an i.i.d.\ $\mc{N}(0,1)$ random vector independent of the  message vector
 $\beta$.} This property, which is proved in Section \ref{sec:amp_proof},  is due to the presence of the ``Onsager" term
\[ \frac{z^{t-1}}{\tau^2_{t-1}}\left( P - \frac{\norm{\beta^t}^2}{n} \right) \]
in the residual update step  \eqref{eq:amp1}. The reader is referred to \cite[Section I-C]{BayMont11} for intuition about role of the Onsager term in the standard AMP algorithm.

In light of the above property, a natural way to generate $\beta^{t+1}$ from $s^t=s$ is 
\be \beta^{t+1}(s) =  \expec[ \beta \, | \, \beta + \tau_t Z =s],   \label{eq:bayes_opt_est} \ee
i.e., $\beta^{t+1}$ is  the Bayes optimal estimate of $\beta$ given the observation $s^t = \beta + \tau_t Z$. For $i \in \text{sec}(\ell)$, $\ell \in [L]$, we have
\be
\label{eq:cond_exp_beta}
\begin{split}
& \beta^{t+1}_i(s)  = \expec[\beta_i \mid  \beta + \tau_t Z = s ]  \\
& = \expec[\beta_i \mid  \{ \beta_j + \tau_t Z_j = s_j \}_{j \in \text{sec}(\ell)} ] \\
& = \sqrt{n P_\ell} \ P(\beta_i = \sqrt{n P_\ell} \mid   \{ \beta_j + \tau_t Z_j = s_j \}_{j \in \text{sec}(\ell)} )\\
&=   \frac{ \sqrt{n P_\ell} \, f( \{s_j \}_{j \in \text{sec}(\ell)} \mid \beta_i = \sqrt{n P_\ell}) \, P(\beta_i = \sqrt{n P_\ell})}
{\sum_{k \in \text{sec}(\ell)}  f( \{s_j \}_{j \in \text{sec}(\ell)} \mid \beta_k = \sqrt{n P_\ell}) \, P(\beta_k = \sqrt{n P_\ell})}
\end{split}
\ee
where we have used Bayes' theorem with  $f(\cdot | \beta_k=\sqrt{nP_\ell})$ denoting the joint density of $\{ \beta_j + \tau_t Z_j \}_{j \in \text{sec}(\ell)}$ conditioned on $\beta_k$ being the non-zero entry in section $\ell$.  Since  $\beta$ and $Z$ are independent with $Z$ having i.i.d.\ $\mc{N}(0,1)$ entries,  for each $k \in \text{sec}(\ell)$ we have
\be
\begin{split}
& f( \{ \beta_j + \tau_t Z_j =s_j \}_{j \in \text{sec}(\ell)} \mid \beta_k = \sqrt{n P_\ell}) \\
  & \propto e^{-(s_k - \sqrt{nP_\ell})^2/2 \tau_t^2} \prod_{j \in \text{sec}(\ell), j \neq k } e^{-s_j^2/2 \tau_t^2} \\
& = e^{s_k \sqrt{n P_\ell}/ \tau^2_t } \, e^{-nP_\ell/2\tau^2_t} \prod_{j \in \text{sec}(\ell)} e^{-s_j^2/2 \tau_t^2}.
\end{split}
\label{eq:cond_denf}
\ee
Using  \eqref{eq:cond_denf} in \eqref{eq:cond_exp_beta},  together with the fact that $P(\beta_k = \sqrt{n P_\ell}) = \frac{1}{M}$ for each $k \in \text{sec}(\ell)$, we obtain
\be
\beta^{t+1}_i(s) = \expec[\beta_i \, | \,  \beta + \tau_t Z = s] = \sqrt{nP_\ell}  \frac{e^{s_i \sqrt{n P_\ell}/\tau^2_t}}
{\sum_{j \in \text{sec}(\ell)} \, e^{s_j \sqrt{n P_\ell} / \tau^2_t}}
\label{eq:bit1_def}
\ee
which is the expression in \eqref{eq:eta_def}.

Thus, under the distributional assumption that $s^t$ equals $\beta + \tau_t Z$, $\beta^{t+1}$  is the estimate of the message vector  $\beta$ (based on $s^t$) that minimizes the expected squared estimation error.    Further, for $i \in \text{sec}(\ell)$, $\beta^{t+1}_i/\sqrt{n P_\ell}$ is  the posterior probability  of $\beta_i$ being the non-zero entry in section $\ell$, conditioned on the observation $s^t = \beta + \tau_t Z$.  Fig. \ref{fig:eta_prog} shows the progression of  $\beta^{t}_{i_\ell}/\sqrt{n P_\ell}$ with $t$ for various sections $\ell$, where $i_\ell$ denotes the index of the true non-zero entry in section $\ell$. We see that the later sections (which are allocated less power) require a larger number of iterations for the posterior probability of the correct term in the section to transition to a value close to one. The iteration at which this transition occurs is determined by the state evolution equations \eqref{eq:taut_def} and \eqref{eq:xt_def}, as discussed below.

 \begin{figure}[t]
\centering
\includegraphics[width=3.5in]{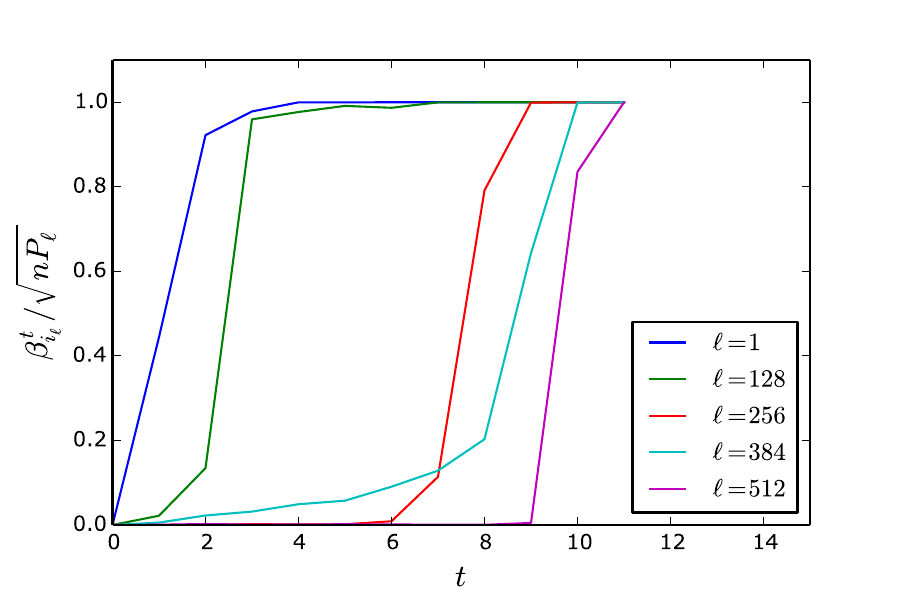}
\caption{ \small{Progression of  $\beta^{t}_{i_\ell}/\sqrt{n P_\ell}$ with $t$ for various sections $\ell$, where $i_\ell$ is the correct term in section $\ell$. The SPARC parameters are $L=512, M=1024, \snr =15, R=0.7 \mc{C}, \, P_\ell \propto 2^{-2R{\ell}/{L}}$. The figure shows the progression for a `typical' simulation run of the AMP decoder, where there were no section errors after decoding. In 100 runs with the above SPARC parameters, a majority of runs resulted in no section errors, and over 95\% of the runs had fewer than five section errors.}}
\vspace{-7pt}
\label{fig:eta_prog}
\end{figure}


\subsection{State Evolution and its Consequences} \label{subsec:sevol}

We now discuss the role of the quantity $x_{t+1}$ in the state evolution equations \eqref{eq:taut_def} and \eqref{eq:xt_def}.

\begin{prop}
Under the assumption that $s^t = \beta + \tau_t Z$, where $Z$ is i.i.d.\ $\sim \mc{N}$(0,1)  and independent of $\beta$,  the quantity $x^{t+1}$ defined in \eqref{eq:xt_def} satisfies
\be
x_{t+1} = \frac{1}{nP} \expec[\beta^* \beta^{t+1}], \quad  1-x_{t+1} = \frac{1}{nP} \expec[ \norm{\beta - \beta^{t+1}}^2],
\label{eq:xt_betatbeta}
\ee
\label{prop:se_cons}
and consequently, $\tau^2_{t+1} = \sigma^2 + \frac{\expec[ \norm{\beta - \beta^{t+1}}^2]}{n}$.
\end{prop}
\begin{proof}
For convenience of notation, we relabel  the $N$ i.i.d.\ random variables $\{Z_k\}_{k \in [N]}$ as \\
$\{U^\ell_j\}_{j \in [M], \ell \in [L]}$.  For any $\ell$, $U^\ell$ denotes the length $M$ vector $\{U^\ell_j\}_{j \in [M]}$, and $U$ is the length $N$ vector  $\{U^\ell\}_{\ell \in [L]}$. We  have
\be
\begin{split}
& \frac{1}{nP} \expec[\beta^* \beta^{t+1}] = \frac{1}{nP} \expec[\beta^* \,  \eta^t(\beta + \tau_t U)] \\
& \stackrel{(a)}{=}  \  \frac{1}{nP} \sum_{\ell=1}^L \expec[ \sqrt{n P_\ell} \ \eta^t_{\textsf{sent}(\ell)}(\beta_\ell + \tau_t U^\ell)  ] \\
& \stackrel{(b)}{=} \frac{1}{nP} \sum_{\ell =1}^L \expec \left[ \sqrt{nP_\ell} \, 
\frac{\sqrt{nP_\ell} \cdot e^{\sqrt{n P_\ell} (\sqrt{nP_\ell} + \tau_t U^\ell_1)/ \tau^2_t}}
{e^{\frac{\sqrt{n P_\ell} (\sqrt{nP_\ell} + \tau_t U^\ell_1)}{\tau^2_t}}  + 
\sum_{j =2}^M e^{\frac{\sqrt{n P_\ell}  \tau_t U^\ell_j}{\tau^2_t}}  } \right] \\
& = \sum_{\ell=1}^{L} \frac{P_\ell}{P} \, \expec \left[
\frac{e^{\frac{\sqrt{n P_\ell}}{\tau_t} \, (U^{\ell}_1  + \frac{\sqrt{n P_\ell}}{\tau_t})} }
{e^{\frac{\sqrt{n P_\ell}}{\tau_t} \, (U^{\ell}_1  + \frac{\sqrt{n P_\ell}}{\tau_t})} + \sum_{j=2}^M e^{\frac{\sqrt{n P_\ell}}{\tau_t}U^{\ell}_j} } \right] = x_{t+1}.
\end{split}
\label{eq:Ebbt}
\ee
In $(a)$ above, the index of the non-zero term in section $\ell$ is denoted by $\textsf{sent}(\ell)$. $(b)$ is obtained by assuming that $\textsf{sent}(\ell)$ is the first entry in section $\ell$ --- this assumption is valid because the prior on $\beta$ is uniform over $\mcb(P_1,\ldots,P_L)$.

Next,  consider
\be
\frac{1}{nP} \expec[ \norm{\beta - \beta^{t+1}}^2 ] = 1 + \frac{ \expec[\norm{\beta^{t+1}}^2]- 2 \expec[\beta^* \beta^{t+1}]}{nP}.
\label{eq:betat_beta_sq}
\ee
Under the assumption that $s^t = \beta + \tau_t Z$, recall from Section \ref{subsec:test_stat} that $\beta^{t+1}$ can be expressed as $\beta^{t+1} = \expec[\beta \mid s^t]$. 
We therefore have
\be
\begin{split}
& \expec[\norm{\beta^{t+1}}^2] = \expec[ \,  \norm{\expec[ \beta | s^t ]}^2 \, ]=   \expec[ \, (\expec[ \beta | s^t ] - \beta + \beta)^* \expec[ \beta | s^t]]  \\
& \stackrel{(a)}{=} \expec[ \, \beta^* \expec[ \beta | s^t ]  \,] = \expec[ \, \beta^* \beta^{t+1}],
\end{split}
\label{eq:betat_sq}
\ee
where step $(a)$ follows because $ \expec[ \, (\expec[ \beta | s^t ] - \beta)^* \expec[ \beta | s^t ] \, ] =0$ due to the orthogonality principle. Substituting \eqref{eq:betat_sq} in \eqref{eq:betat_beta_sq} and using \eqref{eq:Ebbt} yields \[ \frac{1}{nP} \expec[ \norm{\beta - \beta^{t+1}}^2 ] = 1 - \frac{ \expec[ \, \beta^* \beta^{t+1} \, ]}{nP}  = 1 - x_{t+1}. \]
The last claim then follows from \eqref{eq:taut_def}.
\end{proof}

Hence $x_{t+1}$ can be interpreted as the expectation of the (power-weighted) fraction of correctly decoded sections in step $t+1$. We emphasize that this interpretation is accurate only in the limit as $n, M, L \to \infty$, when $s^t$ is distributed as  $\beta + \bar{\tau}_t Z$,  with  $\bar{\tau}_t := \lim \tau_t$.
In Section \ref{sec:amp_proof}  (Lemmas \ref{lem:conv_expec} and \ref{lem:lim_xt_taut}), we derive a closed-form expression for $ \bar{x}_{t+1} := \lim x_{t+1}$ under an exponentially decaying power allocation of the form $P_\ell \propto 2^{-2\mc{C} \ell/L}$.  We show that for rates $R< \mc{C}$,  
\be
\bar{x}_{t} = \frac{ (1+ \snr) - (1+ \snr)^{1- \xi_{t-1}}}{\snr},  \quad  \bar{\tau}_t^2 = \sigma^2 + P(1-\bar{x}_t),  
\label{eq:xt1_formula}
\ee
for $t \geq 0$ where $\xi_{-1}=0$ and 
\be
\begin{split}
\xi_{t} & = \min \left\{ \left(\frac{1}{2\mc{C}}\log\left(\frac{\mc{C} }{R}\right) +  \xi_{t-1}\right), \ 1  \right\}.
\end{split}
\label{eq:2c_alph}
\ee
A direct consequence of \eqref{eq:xt1_formula} and \eqref{eq:2c_alph} is that $\bar{x}_t$ strictly increases with $t$ until it reaches one, and the number of steps $T^*$ until  $\bar{x}_{T^*}=1$  is
$T^* = \left\lceil \frac{2 \mc{C}}{\log(\mc{C}/R)} \right\rceil$.

The constants $\{ \xi_t \}_{t\geq 0}$ have a nice interpretation in the large system limit: at the end of step $t+1$, the first $\xi_t$ fraction of sections in $\beta^{t+1}$ will be correctly decodable with high probability, i.e., the true non-zero entry in these sections will have almost all the posterior probability mass. The other $(1- \xi_t)$ fraction of sections \emph{will not} be correctly decodable from $\beta^{t+1}$ as the power allocated to these sections is not large enough. An additional $\tfrac{1}{2\mc{C}}\log\left(\tfrac{\mc{C} }{R}\right)$ fraction of sections become correctly decodable in each step until $T^*$, when all the sections are correctly decodable with high probability. Fig. \ref{fig:eta_prog} illustrates when  various sections of $\beta$ become decodable for a finite-sized SPARC with $L=512, M=1024$, and $R=0.7 \mc{C}$.

As $\bar{x}_{t}$ increases to $1$,  \eqref{eq:xt1_formula} implies that $\bar{\tau}^2_t$, the variance of the ``noise"  in the  AMP test statistic, decreases monotonically  from 
$\bar{\tau}^2_0= \sigma^2+P$ down to $\bar{\tau}^2_{T^*} = \sigma^2$.  In other words, the initial observation $y=A\beta + w$ is effectively transformed by the AMP decoder into a cleaner statistic $s^{T^*} = \beta + w'$, where $w'$ is Gaussian with the \emph{same} variance as the measurement noise $w$.

\begin{figure}[t]
\centering
\includegraphics[width=3.4in]{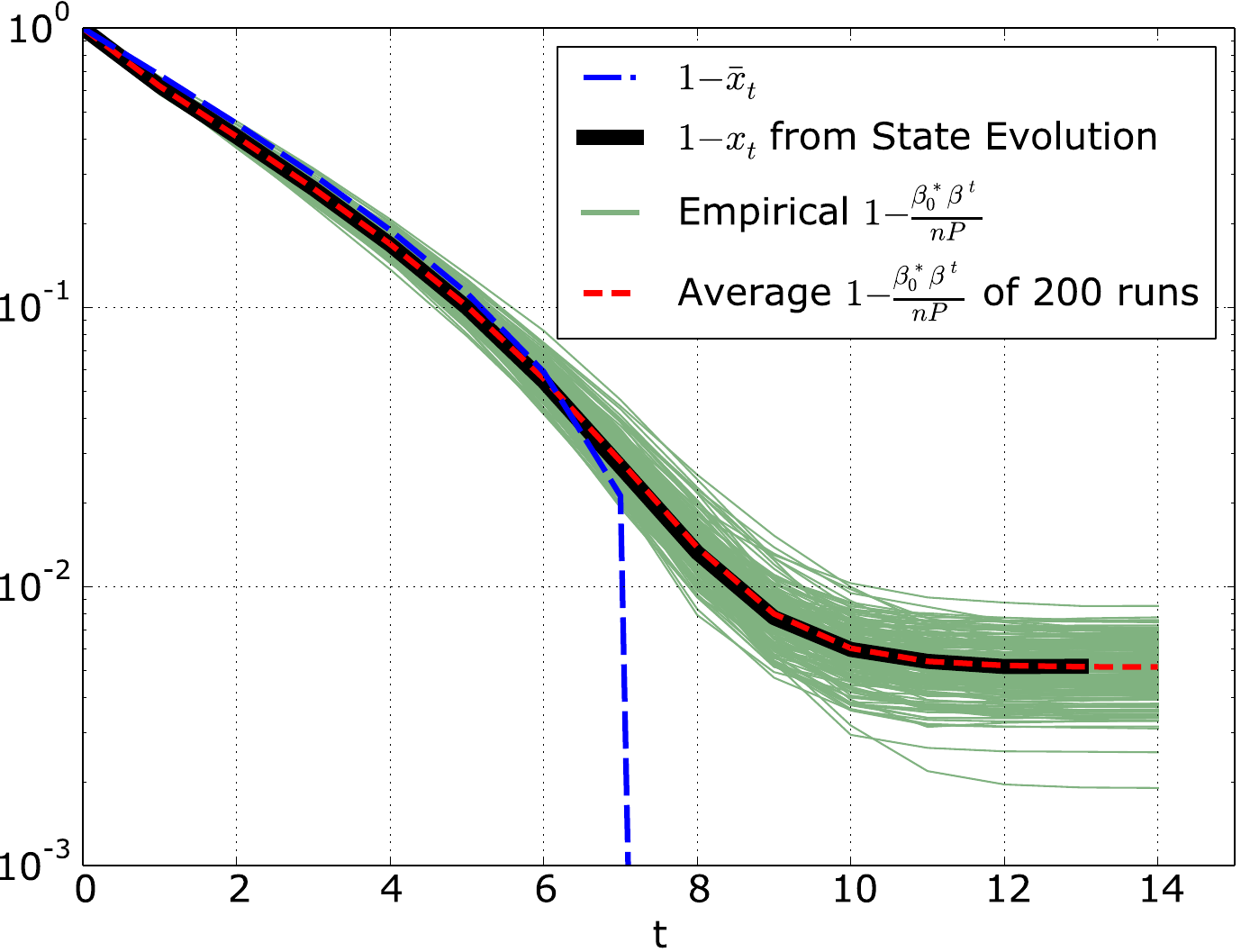}
\caption{\small{Comparison of state evolution predictions with AMP performance. The SPARC parameters are $M=512, L=1024, \snr =15, R=0.7 \mc{C}, \, P_\ell \propto 2^{-2\mc{C}{\ell}/{L}}$. The average of the $200$ trials (green curves) is the dashed red curve, which is almost indistinguishable from the state evolution prediction (black curve).}}
\vspace{-7pt}
\label{fig:SE_example}
\end{figure}
To summarize, for any fixed $R < \mc{C}$, when the AMP decoder is run for a finite number of steps $T^* = \left\lceil \frac{2 \mc{C}}{\log(\mc{C}/R)} \right\rceil$, then in the large system limit $\lim \frac{1}{n} \expec \norm{\beta - \beta^{T^*}}^2$ equals zero.

For finite-sized dictionaries, the test statistic $s^t$ will not be precisely distributed as $\beta + \tau_t Z$. Nevertheless,  computing $x_{t+1}$ numerically  via the state evolution equations \eqref{eq:taut_def} and \eqref{eq:xt_def}  yields an estimate for the expected weighted fraction of correctly decoded sections after each step. 
 Figure \ref{fig:SE_example} shows the trajectory of $(1-x_t)$ vs $t$ for  a SPARC with the parameters specified in the figure. The empirical average of $1- (\beta_0^* \beta^t)/nP$ matches almost exactly with $1-x_{t}$. The theoretical limit $1-\bar{x}_{t}$ given in \eqref{eq:xt1_formula}  is also shown in the figure.

\subsection{Derivation of the AMP}

We describe a min-sum-like message passing algorithm for SPARC decoding from which the AMP decoder is obtained as a first-order approximation.  The aim is to highlight the similarities and differences from the derivation of the AMP in \cite{BayMont11}.   The derivation here is not required for the analysis in the remainder of the paper.

Consider the factor graph for the  model
$y = A\beta + w$,
where $\beta \in \mcb(P_1, \ldots, P_L)$. Each row of $A$ corresponds to a constraint (factor) node, while each column corresponds to a variable node. We use the indices $a,b$ to denote factor nodes, and indices $i,j$ to denote variable nodes. The AMP updates  in \eqref{eq:amp1}--\eqref{eq:amp2} are obtained via  a first-order approximation to the following message passing algorithm that iteratively computes estimates of $\beta$ from $y$.

 For $i \in [N]$, $a \in [n]$, set $\beta^0_{ j \to a}=0$,  and compute the following for $t \geq 0$:
\begin{align}
z^t_{a \to i}  & = y_a - \sum_{j \in [N] \bks i} A_{aj}  \beta^t_{j \to a},  \label{eq:z_update} \\
\beta^{t+1}_{i \to a} & = \eta_i^t \left( {s}_{i \to a}  \right),
\label{eq:beta_update}
\end{align}
where $\eta_i^t (\cdot)$ is the estimation function defined in \eqref{eq:eta_def}, and for $i \in \text{sec}(\ell)$, the entries of the test statistic ${s}_{i \to a} \in \mathbb{R}^M$ are defined as
\be
\begin{split}
({s}_{i \to a})_i &  = \sum_{b \in [n] \bks a}  A_{bi} z^t_{b \to i},  \\
({s}_{i \to a})_j & =  \sum_{b \in [n] }  A_{bj} z^t_{b \to j}, \quad j \in \text{sec}(\ell) \bks i.
\end{split}
\ee

 It is useful to compare  the $\beta$-update in \eqref{eq:beta_update} to the message passing algorithm from which the traditional AMP is derived (cf.\ equation $(1.2)$ in \cite{BayMont11}). In  \cite{BayMont11}, the vector $x$ to be recovered is assumed to be i.i.d.\ across entries; hence we have a single estimating function $\eta^t$ in this case, which for $i \in [N]$, $a \in [n]$,  generates the message
 \be
 x^{t+1}_{i \to a}  = \eta^t \Bigg(\sum_{b \in [n] \bks a}  A_{bi} z^t_{b \to i}\Bigg).
 \label{eq:xia_def}
 \ee
In \eqref{eq:xia_def}, each outgoing message  from the $i$th variable node depends only on its own incoming messages. In contrast, in \eqref{eq:beta_update},  each outgoing message from a variable node depends on the incoming  messages of all the other nodes in the \emph{same section}.  This is due to the constraint that $\beta$ has exactly one non-zero entry in each section, which ensures that entries of $\beta^t$ within each section are dependent, while entries in different sections are mutually independent.

The derivation of the AMP updates  in \eqref{eq:amp1}--\eqref{eq:amp2} starting from the messaging passing algorithm  \eqref{eq:z_update}--\eqref{eq:beta_update} is given in Appendix \ref{app:amp_derive}.

\section{Performance of the AMP Decoder} \label{sec:AMP_perf}

Before giving the main result, we state two lemmas that specify the limiting behaviour of the state evolution parameters defined in \eqref{eq:taut_def}, \eqref{eq:xt_def}.  Treating $x_{t+1}$ in \eqref{eq:xt_def} as a function of $\tau$,
we can define
\be
x(\tau) := \sum_{\ell=1}^{L} \frac{P_\ell}{P} \, \expec \left[
\frac{e^{ \frac{\sqrt{n P_\ell}}{\tau} \, (U^{\ell}_1  + \frac{\sqrt{n P_\ell}}{\tau})} }{ e^{ \frac{\sqrt{n P_\ell}}{\tau} \, (U^{\ell}_1  + \frac{\sqrt{n P_\ell}}{\tau})} + \sum_{j=2}^M e^{ \frac{\sqrt{n P_\ell}}{\tau}U^{\ell}_j} } \right],
\label{eq:xt_tau_def}
\ee
where $\{ U^\ell_j\}$ are i.i.d.\ $\sim \mc{N}(0,1)$ for $j \in[M], \ \ell \in [L]$.
\begin{lem}
For any power allocation $\{ P_\ell \}_{\ell=1, \ldots, L}$ that is non-increasing with $\ell$, we have
\be
\bar{x}(\tau) := \lim x(\tau) = \lim \,  \sum_{\ell=1}^{ \lfloor \xi^*(\tau) L \rfloor} \frac{P_\ell}{P}, 
\label{eq:xt_tau_def1}
\ee
where $\xi^*(\tau)$ is the supremum of all $\xi \in (0,1]$ that satisfy
\[ \lim L P_{ \lfloor \xi L \rfloor} >  2(\ln 2) R \, \tau^2. \]
 If $ \lim L P_{ \lfloor \xi L \rfloor} \leq  2(\ln 2) R \, \tau^2$ for all $\xi > 0$, then $\bar{x}(\tau)=0$. (The rate $R$ is measured in bits.)
\label{lem:conv_expec}
\end{lem}
\begin{proof} In Appendix \ref{app:conv_exp}. \end{proof}

Since the entries of $A$ are i.i.d., the assumption  that $\{ P_\ell \}$ is non-decreasing with $\ell$  can be made without loss of generality. 
Recalling that ${x}_{t+1}$ is the expected power-weighted fraction of correctly decoded sections after step $(t+1)$,   for any power allocation $\{P_\ell \}$, Lemma \ref{lem:conv_expec}  may be interpreted as follows: in the large system limit,  sections $\ell$ such that $\ell \leq  \lfloor \xi^*(\bar{\tau}_t) L \rfloor$ will be correctly decoded in step
$(t+1)$. All sections satisfying this condition will be decodable in step $(t+1)$ (i.e., will have most of the posterior probability mass on the correct  term); conversely all sections whose power falls below the threshold will not be decodable in this step.

The performance of the AMP decoder will be analyzed with the following exponentially decaying power allocation:
 \be  P_\ell = P \cdot  \frac{2^{2\mc{C}/L} -1}{1- 2^{-2\mc{C}}} \cdot 2^{-2\mc{C}\ell/L}, \quad  \ell \in [L].  \label{eq:exp_power_alloc} \ee
For the  power allocation in \eqref{eq:exp_power_alloc}, we have for $\xi \in (0,1]$
\be
\lim L P_{ \lfloor \xi L \rfloor}  =  \sigma^2 (1+\snr)^{1-\xi} \ln(1+\snr).
\label{eq:cell}
\ee

\begin{lem}
For the power allocation $\{ P_\ell \}$ given in \eqref{eq:exp_power_alloc}, we have for $t=0,1,\ldots$:
\begin{align}
\bar{x}_{t} & := \lim x_{t} = \frac{ (1+ \snr) - (1+ \snr)^{1- \xi_{t-1}}}{\snr} \label{eq:limxt1}, \\
\bar{\tau}^2_{t}&  := \lim \tau^2_{t}  = \sigma^2 + P(1 - \bar{x}_t) = \sigma^2\left( 1 + \snr \right)^{1-\xi_{t-1}} \label{eq:limtaut1}
\end{align}
where $\xi_{-1}=0$, and for $t \geq 0$,
\be
\begin{split}
 \xi_{t} & = \min \left\{ \left(\frac{1}{2\mc{C}}\log\left(\frac{\mc{C} }{R}\right) +  \xi_{t-1}\right), \ 1  \right\}.
\end{split}
\label{eq:lim_alph}
\ee
\label{lem:lim_xt_taut}
\end{lem}
\begin{proof} In Appendix \ref{app:lim_xt_taut}. \end{proof}

We observe from Lemma \ref{lem:lim_xt_taut} that $\xi_t$ increases  in each step  by
$ \tfrac{1}{2\mc{C}}\log\left(\tfrac{\mc{C} }{R}\right)$ until it  equals $1$. Also note that $\bar{\tau}^2_t$ strictly decreases with 
$t$ until it reaches
$\sigma^2$ (when 
$\xi_t$ reaches $1$), after which it remains constant. Thus  the number of steps until $\xi_t$ reaches one (i.e., $\bar{\tau}^2_t$ stops decreasing) equals
\be
 T^*  = \left\lceil \frac{2 \mc{C}}{\log(\mc{C}/R)} \right\rceil.
 \label{eq:Tstar_def}
 \ee

Our main result is proved for the following AMP decoder, which uses the asymptotic values $\{\bar{\tau}^2_t\}$ defined in Lemma \ref{lem:lim_xt_taut}, and runs for exactly $T^*$ steps.  Set $\beta^0 =0$ and compute
\begin{align}
z^t & =y - A\beta^t + \frac{z^{t-1}}{\bar{\tau}^2_{t-1}}\left( P - \frac{\norm{\beta^t}^2}{n} \right), \label{eq:asymp_amp1}\\
\beta^{t+1}_i & = \eta^t_{i}( \beta^t + A^*z^t), \quad  \text{ for } i \in [N] \label{eq:asymp_amp2}
\end{align}
where for $i \in \text{sec}(\ell), \   \ell  \in [L]$, 
\be
\eta^t_i(s) = \sqrt{nP_\ell}\,  \frac{ e^{s_i \sqrt{n P_\ell}/ \bar{\tau}^2_t} }
{\sum_{j \in \text{sec}(\ell)} \, e^{ s_j \sqrt{n P_\ell}/\bar{\tau}^2_t }}.
\label{eq:asymp_eta_def}
\ee
The only difference from the earlier decoder  described in \eqref{eq:taut_def}--\eqref{eq:eta_def} is that we now use the limiting values $\{ \bar{\tau}^2_t \}$ from  Lemma \ref{lem:lim_xt_taut}  instead of $\{ \tau^2_t \}$.  The algorithm terminates after generating $\beta^{T^*}$, where $T^*$ is defined in \eqref{eq:Tstar_def}. 
The decoded codeword
$\hat{\beta} \in \mcb(P_1,\ldots, P_L)$ is obtained  by setting the maximum of $\beta^{T^*}$  in each section $\ell$ to $\sqrt{nP_{\ell}}$ and the remaining entries to $0$.  

The \emph{section error rate}  of a decoder for a SPARC $\mc{S}$ is defined as
\be
\mc{E}_{sec}(\mc{S}) := \frac{1}{L} \sum_{\ell =1}^{L}  \mathbf{1}{\{ \hat{\beta}_\ell \neq \beta_{0_\ell} \}}.
\label{eq:sec_err_def}
\ee

\begin{thm}
Fix any rate $R < \mc{C}$, and $\textsf{a} >0$. Consider a sequence of rate $R$ SPARCs $\{ \mc{S}_n \}$ indexed by block length $n$, with design matrix parameters $L$ and $M=L^{\textsf{a}}$ determined according to  \eqref{eq:llogl_nR}, and an exponentially decaying power allocation given by \eqref{eq:exp_power_alloc}. Then  the section error rate of the AMP decoder (described in  \eqref{eq:asymp_amp1}--\eqref{eq:asymp_eta_def}, and run for $T^*$ steps) converges to zero almost surely, i.e., for any $\e >0$,
\be \lim_{n_0 \to \infty} P\left(  \mc{E}_{sec}(\mc{S}_n)   <  \e, \ \forall n \geq n_0 \right) = 1. \label{eq:pezero} \ee
\label{thm:main_amp_perf}
\end{thm}

\emph{Remarks}:
\begin{enumerate}
\item The probability measure in \eqref{eq:pezero} is over the Gaussian design matrix ${A}$, the Gaussian channel noise $w$, and the  the message $\beta$ distributed uniformly in $\mcb(P_1, \ldots, P_L)$.

\item As in \cite{AntonyFast}, we can construct a concatenated code with  an inner SPARC of  rate $R$ and an outer Reed-Solomon (RS)  code of rate 
$(1-2\e)$. If $M$ is a prime power, a RS code defined over a finite field of order $M$  defines a one-to-one mapping  between a symbol of the RS codeword and a section of the SPARC.  
The concatenated code has rate $R(1-2\e)$, and decoding complexity that is polynomial in $n$. The decoded message $\hat{\beta}$ equals $\beta$ whenever the section error rate of the SPARC is less than $\e$.   Thus for any $\e >0$,  the theorem guarantees that the probability of \emph{message} decoding error for a sequence of rate $R(1-2\e)$ SPARC-RS concatenated codes  will tend to zero, i.e.,
$ \lim P(\hat{\beta} \neq \beta) =0$.
\end{enumerate}

The proof  of Theorem \ref{thm:main_amp_perf}  is given in Section \ref{sec:amp_proof}.

\subsection{Empirical Performance at Finite Blocklengths} \label{sec:ex_results}

In this section, we make two modifications to the SPARC construction used  in
Theorem~\ref{thm:main_amp_perf} to improve the empirical performance at finite
block lengths. First, we introduce a power allocation that yields several
orders of magnitude improvement in section error rate for rates $R$ that are
not very close to the capacity $\mc{C}$. Second, we use a Hadamard design
matrix (instead of Gaussian), which facilitates a decoder with  $O(N \log N)$
running time and a memory requirement of $O(N)$. In comparsion, with a Gaussian
design matrix  the running time and memory of the AMP decoder are both
$O(nN)$.  We mention that  the recent work \cite{BarbSchKrz15} considers  an
AMP decoder with a spatially coupled Hadamard-based design matrix. In our case, the Hadamard
design matrix  is not spatially coupled, rather it is the modified  power
allocation that yields low section error rates.

\subsubsection*{Modified Power Allocation}

We define a power allocation characterized by two parameters $a,f$.
For $f \in [0,1] $, let
\be
P_\ell = \begin{cases}
        \kappa \cdot 2^{-2a\mc{C} \ell/L}, & 1 \leq \ell \leq fL\\
        \kappa \cdot 2^{-2a\mc{C} f},   & fL+1 \leq \ell \leq L
\end{cases}
\label{eq:mixed_power_alloc}\ee
where 
\[
\kappa = \frac{P\left(2^{2a\mc{C} /L}-1\right)}{1-2^{-2a\mc{C} f}\left(1-L(1-f)(2^{2a\mc{C} /L}-1)\right)}.
 \]
The normalizing constant $\kappa$ ensures  that the total power across sections
is $P$.  For intuition, first assume that $f=1$. Then
\eqref{eq:mixed_power_alloc} implies that $P_\ell \propto 2^{-a2\mc{C} \ell/L}$
for $\ell \in [L]$. Setting $a=1$ recovers the original power allocation of
\eqref{eq:exp_power_alloc}, while $a=0$ allocates $\frac{P}{L}$ to each section.
Increasing $a$ increases the power allocated to the initial sections which
makes them more likely to decode correctly, which in turn helps by decreasing
the effective noise variance $\bar{\tau}^2_t$ in subsequent AMP iterations.
However, if $a$ is too large, the final sections may have too little power to
decode correctly.  
 
Hence we want the parameter $a$ to be large enough to ensure that the AMP gets
started on the right track, but not much larger.  This intuition can be made
precise in the large system limit using Lemma~\ref{lem:conv_expec}:  recall that for a section $\ell$ to be correctly decoded in step
$(t+1)$,  the limit of  $L P_\ell$ must exceed a threshold proportional to $R\bar{\tau}^2_t$. For rates close to $\mc{C}$, we need $a$ to be close to $1$ for the initial sections to cross this threshold and get decoding started
correctly. On the other hand, for rates such as $R=0.6 \mc{C}$,  $a=1$ allocates more power than necessary to the initial sections, leading to poor decoding performance in the final sections.  

In addition, we found that the section error rate can be further improved by \emph{flattening} the power allocation in the final sections.  For a given $a$, \eqref{eq:mixed_power_alloc} has an exponential power allocation until section $fL$, and constant power for the remaining $(1-f)L$ sections.  The allocation in \eqref{eq:mixed_power_alloc}  is continuous, i.e.\ each section in the flat
part is allocated the same power as the final section in the exponential part. Flattening boosts the power given to the final sections compared to an exponentially decaying allocation. The two parameters $(a,f)$ let us trade-off between the conflicting objectives of assigning enough power to the initial sections  and ensuring that the final sections have enough power to be decoded
correctly. 

\textbf{The constants $\bar{\tau}^2_t$ and $\bar{x}_t$}:
Analogous to Lemma~\ref{lem:lim_xt_taut}, the large system limit values of the state evolution parameters for the power allocation in
\eqref{eq:mixed_power_alloc} can be obtained from Lemma~\ref{lem:conv_expec}. Set $\bar{\tau}_0^2= \sigma^2 + P$, and for $t \geq 0$ compute
 \begin{align}
    \bar{\xi}_t &= \min \Big\{ \frac{1}{2a\mc{C}} \log \Big(  \frac { a\mc{C} P\,  2^{2a\mc{C} f}} {R\bar{\tau}_t^2[ 2^{2a\mc{C} f} + (1-f)2a\mc{C}\ln 2 -1]} \Big),  \nonumber \\
 & \qquad  \qquad  1 \Big\}, \label{eq:contpa_xit} \\ \bar{x}_{t+1}&=\frac
        {1-2^{-2a\mc{C}\bar{\xi}_t}} {1+2^{-2a\mc{C} f}((1-f) 2 a\mc{C} \ln 2-1)}, \\
    \bar{\tau}^2_{t+1} & =\sigma^2+P(1-\bar{x}_{t+1}).
\label{eq:contpa_limit_sys} \end{align} We note that setting $a=f=1$ in
\eqref{eq:contpa_xit}--\eqref{eq:contpa_limit_sys} recovers the limiting state
evolution parameters for the exponential power allocation, which were obtained
in Lemma~\ref{lem:lim_xt_taut}.
 
\begin{figure}[t] \centering
    \includegraphics[width=3.5in, height=2.5in]{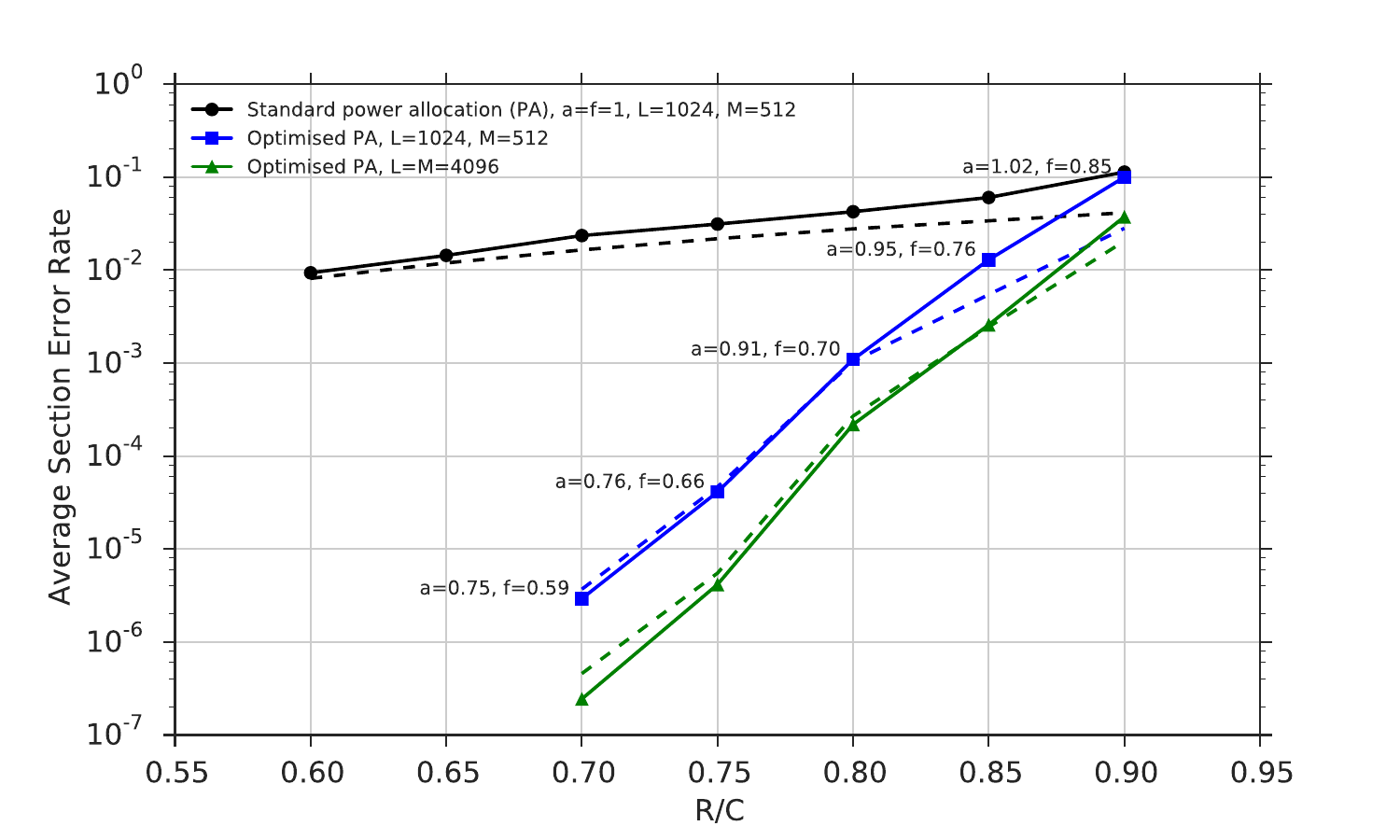}
    \caption{\small{Section error rate vs $R/\mc{C}$ at $\snr=15, \mc{C}=2$ bits.
        The top solid curve shows the average section error rate of the AMP
        over $1000$ trials with $P_\ell \propto 2^{-2\mc{C}\ell/L}$.
        The  solid curve in the middle shows the section error rate using the
        power allocation in \eqref{eq:mixed_power_alloc} with the $(a,f)$  values shown.
        The SPARC parameters for both these curves are $M=512, L=1024$.
        The bottom solid  curve shows the section error rate with the same
        $(a,f)$ values, but $L=M=4096$.  In all cases, the dashed lines show the state evolution prediction \eqref{eq:vt_def} of the section error rate.    Missing points at $R=0.6\mc{C}$ and $0.65\mc{C}$ indicate  no errors observed over $1000$ trials.}}
    \vspace{-7pt}
\label{fig:amp_sec_err_rate} 
\end{figure}

\textbf{Experimental Results}: Fig.~\ref{fig:amp_sec_err_rate} shows the
performance of the AMP at different rates.  Given the values of $M, L$, the
block length $n$ is determined by the rate $R$ according to \eqref{eq:ml_nR}.
For example, with $M=512, L=1024$, we have 
$n=7680$ for $R=0.6 \mc{C}$, and  $n=5120$ for $R=0.9\mc{C}$. The solid curve
at the top  shows the average section error rate of the AMP (over $1000$ runs)
with an exponentially decaying power allocation where $P_\ell \propto
2^{-2\mc{C} \ell/L}$. The solid curve in the middle shows the average section
error rate  with the power allocation in \eqref{eq:mixed_power_alloc}, with
values of $(a,f)$ obtained via a rough optimization around an initial guess of
$a=f=R/\mc{C}$.  The solid curve at the bottom shows the average section error
rate with $L=M=4096$, and the power allocation in \eqref{eq:mixed_power_alloc}
with same $(a,f)$ values as before.

In all cases, the decoder described in
\eqref{eq:asymp_amp1}--\eqref{eq:asymp_eta_def} was used.  The constants
$\{\bar{\tau}^2_t\}$ required by the decoder are  specified by
Lemma~\ref{lem:lim_xt_taut} for the exponential allocation, and by
\eqref{eq:contpa_xit}--\eqref{eq:contpa_limit_sys}  for the modified
allocation.  The simulations for Fig.~\ref{fig:amp_sec_err_rate}  were run
using Hadamard design matrices, which are described below.

Across trials, we observed good concentration around the average section error
rates. For example, with $M=512, L=1024$ and $R=0.75 \mc{C}$, $958$ of the
$1000$ trials had zero errors, and the remaining $42$ had only one section in
error, for an average section error rate of $4.10\times10^{-5}$. Further, all
the section errors were in the flat part of the power allocation, as expected.
Increasing $L$ tends to improve this concentration, while increasing $M$
reduces the average section error rate. This improvement in the section error
rate is illustrated by the bottom curve in Fig.~\ref{fig:amp_sec_err_rate}.

The dashed curves in Fig.~\ref{fig:amp_sec_err_rate} show the
section error rate predictions for the two power allocations obtained from
state evolution. Recall from Section~\ref{subsec:sevol} that $\bar{x}_{t+1}$ in
\eqref{eq:xt_def} can be interpreted as the expectation of the (power-weighted)
fraction of correctly decoded sections after step $t+1$. Using arguments
similar to Proposition~\ref{prop:se_cons}, we can show that under the
assumption that the test statistic $s^t \sim \beta + \bar{\tau}_t Z$,  the
\emph{non-weighted} expectation of the correctly decoded sections after step
$(t+1)$ is given by

\be
\begin{split} 
& \frac{1}{nP} \sum_{\ell=1}^L \, \frac{P/L}{P_\ell} \, \expec[\beta^*_\ell \beta^{t+1}_\ell]   \\
& = \sum_{\ell=1}^{L} \,  \frac{1}{L} \,  \expec \left[ \frac{e^{\frac{\sqrt{n P_\ell}}{\tau_t} \, (U^{\ell}_1  + \frac{\sqrt{n P_\ell}}{\tau_t})} }
{e^{\frac{\sqrt{n P_\ell}}{\tau_t} \, (U^{\ell}_1  + \frac{\sqrt{n P_\ell}}{\tau_t})} + \sum_{j=2}^M e^{\frac{\sqrt{n P_\ell}}{\tau_t}U^{\ell}_j} } \right] := {v}_{t+1}. 
\end{split}
 \label{eq:vt_def}
\ee

Thus  $v_{T^*}$ is an estimate of the section error rate. We observe that  the
empirical  section error rate in Fig.~\ref{fig:amp_sec_err_rate} is close to
the $v_{T^*}$, especially for the larger dictionary.

It is  evident  that judicious power allocation can yield significant
improvements in section error rates. An interesting open question is to find
good rules of thumb for the power allocation as a function of rate and $\snr$.
For any given allocation,  one can determine  whether the section error rate
goes to zero in the large system limit. Indeed, using
Lemma~\ref{lem:conv_expec} with $\bar{\tau}^2_0 = \sigma^2+P$,  we see that
those sections $\ell$ for which the indicator in \eqref{eq:xt_tau_def1} is
positive are decoded in the first step; this also gives the value of
$\bar{x}_1$. Then with $\bar{\tau}^2_1 = \sigma^2 + P(1 -\bar{x}_1)$ we can
determine which sections are decoded in step $2$, and so on.  The section error
rate goes to zero if and only if $\bar{x}_{T^*} =1$. The proof of this is
essentially identical to that of Theorem~\ref{thm:main_amp_perf}. 

Thus Lemma~\ref{lem:conv_expec} gives a straightforward way to check whether a
power allocation is good in the large system limit. This can provide some
guidance for the finite length case, but the challenge is to choose between
several power allocations for which $\bar{x}_{T^*} =1$. One way to compare
these allocations may be via the state evolution prediction $v_{T^*}$
from \eqref{eq:vt_def},  but this needs additional investigation.

\subsubsection*{Reducing the decoding complexity using Hadamard Dictionaries} \label{subsec:hadamard}

The computational complexity of the decoder in
\eqref{eq:asymp_amp1}--\eqref{eq:asymp_eta_def} is determined by the
matrix-vector multiplications $A\beta^t$ and $A^*z^t$, whose running time is
$O(nN)$ if performed in the straightforward way. The remaining operations are
$O(N)$. As the number of iterations is finite, the decoding complexity scales
linearly with the size of the design matrix.  With a Gaussian design matrix,
the memory requirement  is also proportional to  $nN$ as the entire matrix has
to be stored.  This is the major bottleneck in scaling the AMP decoder to work
with large design matrices.  

To reduce the decoding complexity and the  required memory, we generate $A$
from a Hadamard matrix as follows. Let $N=ML$ be a power of $2$, and let $m=
\log_2 N$. With $H_0=1$, recursively define the $2^m \times 2^m$  matrix $H_{m}$ as

\ben
H_{m} = \begin{pmatrix} H_{m-1} & H_{m-1} \\ H_{m-1} & -H_{m-1}   \end{pmatrix}.
\label{eq:had_def}
\een

The design matrix $A$ is generated by picking $n$ rows uniformly at random from
$H_m$ and scaling the resulting matrix by $\frac{1}{\sqrt{n}}$ so that each
column has norm one.\footnote{Strictly speaking, we generate $A$ by uniformly
sampling from all rows of $H_m$ \emph{except} the first. This is because
the first row is all ones, while the others have an equal number of $1$s and
$-1$s. } Thus the $k$th element of the codeword is  $(A\beta)_k =  \sum_{j \in
[N]} A_{kj} \beta_j$, where $A_{kj} \in \{\tfrac{1}{\sqrt{n}}, \tfrac{-1}{\sqrt{n}} \}$ for $k\in[n], j \in [N]$.

For $A$ generated as above, the matrix-vector multiplications $A\beta^t$ and
$A^*z^t$ can be performed efficiently using the \emph{fast Walsh-Hadamard
Transform} (WHT) \cite{Shanks69}. Let $\mc{S}_n$ denote the set of $n$ indices
of the rows of $H_m$ that constitute $A$.  To compute $A\beta^t$, compute the
length-$N$ WHT of $\beta^t$ and keep only the elements indexed by $\mc{S}_n$.
To compute $A^* z^t$, first extend $z^t \in \mathbb{R}^n$ to a vector
$\tilde{z}^t \in \mathbb{R}^N$ by embedding $z^t$ in the indices corresponding to
$\mc{S}_n$, and setting the remaining entries to zero.  Since $H_m$ is
symmetric, the length-$N$ WHT of $\tilde{z}^t$ equals $A^*z^t$.  

The fast WHT has $O(N \log N)$ running time. Further, we do not need to store
$A$; only the vectors $\beta^t$ and $z^t$ need to be kept in
memory. Hence the running time and memory requirement of the decoder are now
$O(N\log N)$ and $O(N)$, respectively.  These substantial improvements allow
the use of much larger dictionaries (e.g.,  $M=L=4096$) for which AMP decoding
with Gaussian matrices is infeasible with standard computing resources. For
given values of $n,M,L$ and power allocation $\{P_\ell\}$, we found the
empirical performance with  a Hadamard dictionary to be very similar to the
Gaussian case. 

\section{Proof of Theorem \ref{thm:main_amp_perf}} \label{sec:amp_proof}

The main ingredients in the proof of Theorem \ref{thm:main_amp_perf} are two technical lemmas (Lemma \ref{lem:hb_cond} and Lemma \ref{lem:main_lem}).  We first lay down the notation that will be used in the proof.  We then state the two lemmas and use them to prove Theorem \ref{thm:main_amp_perf}.

\subsection{Definitions and Notation for the Proof}
For consistency and ease of comparison, we use notation similar to  \cite{BayMont11}.  Define the following column vectors recursively for $t\geq 0$, starting with $\beta^0=0$ and $z^0=y$.
\begin{equation}
\begin{split}
h^{t+1}  := \beta_0 - (A^*z^t + \beta^t), \qquad &  q^t  :=\beta^t - \beta_0, \\
b^t := w-z^t,\qquad & m^t :=-z^t.
\end{split}
\label{eq:hqbm_def}
\end{equation}
Recall that $\beta_0$ is the message vector chosen by the transmitter. Due to the symmetry of the code construction, we can assume that the non-zeros of $\beta_0$ are in the first entry of each section.

Define  $\mathscr{S}_{t_1, t_2}$ to be the sigma-algebra generated by
\[ b^0, ..., b^{t_1 -1}, m^0, ..., m^{t_1 - 1}, h^1, ..., h^{t_2}, q^0, ..., q^{t_2}, \text{ and }  \beta_0, w. \]
Lemma \ref{lem:hb_cond} iteratively computes the conditional distributions $b^t |_{ \mscrs_{t, t}}$ and $h^{t+1} |_{ \mscrs_{t+1, t}}$. Lemma \ref{lem:main_lem} then uses this conditional distributions to show the convergence of various inner products involving $h^{t+1}, q^t, b^t$, and $m^t$ to deterministic constants.  

For $t \geq 1$, let
\be
\lambda_t := \frac{-1}{\bar{\tau}^2_{t-1}}\left( P - \frac{\norm{\beta^t}^2}{n} \right).
\label{eq:lambda_t_def}
\ee
We then have
\begin{equation}
b^{t} + \lambda_t m^{t-1} = A q^t,
\label{eq:bmq}
\end{equation}
which follows from \eqref{eq:amp1} and \eqref{eq:hqbm_def}. We also have
\be
h^{t+1} + q^t = A^* m^t.
\label{eq:hqm}
\ee
From \eqref{eq:bmq} and \eqref{eq:hqm}, we have the matrix equations
\be
X_t = A^* M_t, \quad Y_t =AQ_t,
\ee
where
\be
\begin{split}
X_t  &= [h^1 + q^0 \mid h^2+q^1 \mid \ldots \mid h^t + q^{t-1}], \\
Y_t  &= [ b^0 \mid b^1 + \lambda_1 m^0 \mid \ldots \mid b^{t-1} + \lambda_{t-1} m^{t-2}], \\
M_t  &= [m^0 \mid \ldots \mid m^{t-1} ], \\
Q_t  &=  [q^0 \mid \ldots \mid q^{t-1} ].
\end{split}
\label{eq:XYMQt}
\ee
The notation $[c_1 \mid c_2 \mid \ldots \mid c_k]$ is used to denote a matrix with columns $c_1, \ldots, c_k$.  Additionally define the matrices 
\be
\begin{split} B_t &:= [b^0 | \ldots | b^{t-1}], \qquad  H_t = [h^1 | \ldots | h^{t}],\\
\Lambda_t &:= \text{diag}(\lambda_0, \ldots, \lambda_{t-1}) 
\end{split}
\ee 
Note that $M_0, Q_0, B_0$, $H_0$, and $\Lambda_0$ are all-zero vectors.  Using the above we see that
 \be Y_t = B_t + \Lambda_t [0 | M_{t-1}] \quad \text{ and } \quad X_t = H_t + Q_{t}. \label{eq:XtYt_rel} \ee

 We use  $m^t_{\|}$ and $q^t_{\|}$ to denote the projection of $m^t$ and $q^t$ onto the column space of $M_t$ and $Q_t$, respectively. Let
 $\vec{\alpha}_t := (\alpha^t_0, \ldots, \alpha^t_{t-1})$ and $\vec{\gamma}_t :=  (\gamma^t_0, \ldots, \gamma^t_{t-1})$ be the coefficient vectors of these projections, i.e.,
 \be
 m^t_{\| } = \sum_{i=0}^{t-1} \alpha^t_i m^i, \quad  q^t_{\|} = \sum_{i=0}^{t-1} \gamma^t_i q^i.
 \label{eq:mtqt_par}
 \ee
 The projections of $m^t$ and $q^t$ onto the orthogonal complements of $M^t$ and $Q^t$, respectively,  are denoted by
 \be
 m^t_{\perp} := m^t - m^t_{\|}, \quad  q^t_{\perp} := q^t - q^t_{\|}
  \label{eq:mtqt_perp}
 \ee
With $\bar{\tau}^2_t$  and $\bar{x}_t$ as defined in Lemma \ref{lem:lim_xt_taut}, for $t \geq 0$ define
\be
\bar{\sigma}^2_t : =\bar{\tau}_t^2 - \sigma^2 = P(1-\bar{x}_t),
\label{eq:sigt_def}
 \ee
 Let $(\bar{\sigma}^{\perp}_0)^2 := \bar{\sigma}_0^2$ and $(\bar{\tau}^{\perp}_0)^2 := \bar{\tau}_0^2$, and for $t > 0$ define 
\be
\begin{split}
& (\bar{\sigma}_{t}^{\perp})^2 := \bar{\sigma}_{t}^2 \left(1 - \frac{ \bar{\sigma}_{t}^2 }{ \bar{\sigma}_{t-1}^2 }\right),  \text{ and } (\bar{\tau}^{\perp}_{t})^2 := \bar{\tau}_{t}^2 \left(1 - \frac{\bar{\tau}_{t}^2}{\bar{\tau}_{t-1}^2}\right).
\label{eq:sigperp_defs}
\end{split}
\ee

 Given two random vectors $X, Y$ and a sigma-algebra $\mscrs$, $X |_\mscrs \stackrel{d}{=} Y$ implies that the conditional distribution of $X$  given $\mscrs$ equals the distribution of $Y$. For random variables $X,Y$, the notation $X \stackrel{a.s.}{=} Y$ means that $X$ and $Y$ are equal almost surely. We use the notation $\vec{o}_t(n^{-\delta})$ to denote a vector in $\mathbb{R}^t$  such that each of its coordinates is  $o(n^{-\delta})$ (here $t$  is fixed).  The identity matrix is denoted by $\mathsf{I}$.

The notation `$\lim$'  is used to denote the large system limit as $n,M, L \to \infty$; recall that the three quantities are related as
$L \log M =nR$, with $M=L^b$.
We keep in mind that (given $R$ and $b$) the block length $n$ uniquely determines the dimensions of all the quantities in the system including
$A, \beta_0, w,  h^{t+1}, q^{t}, b^t, m^t$. Thus we have a sequence  indexed by $n$  of each of these random quantities,  associated with the  sequence of SPARCs
$\{ \mc{S}_n \}$.

We next characterize (in Lemma \ref{lem:hb_cond})  the conditional distribution of the vectors $h^{t+1}$ and $b^t$ given the matrices in \eqref{eq:XYMQt} as well as $\beta_0$ and $w$.  This shows that $h^{t+1}$ and $b^t$ can each be expressed as the sum of an i.i.d.\ Gaussian random vector and a deviation term. Lemma \ref{lem:main_lem} then shows that these deviation terms are small, in the sense that their section-wise maximum absolute value and norm converge to $0$ almost surely. Lemma \ref{lem:main_lem} also provides convergence results for various inner products and functions involving $\{h^{t+1}, q^t, b^t, m^t\}$.  
These will be used to show that the performance of the AMP decoder in the large system limit is accurately predicted by the state evolution equations \eqref{eq:limxt1} and \eqref{eq:limtaut1}. In particular, it is shown that the squared error $\frac{1}{n} \norm{\beta^{t} - \beta}^2$  converges almost surely to $P(1-\bar{x}_{t})$, for $0 \leq t \leq T^*$.

\subsection{Conditional Distribution Lemma}

A key ingredient in the proof is the  distribution of $A$ conditioned on the sigma algebra $\mscrs_{t_1,t}$ where $t_1$ is either $t+1$ or $t$. Observing that conditioning on $\mscrs_{t_1,t}$ is equivalent to conditioning on the linear constraints\footnote{While conditioning on the linear constraints,  we emphasize that only $A$ is treated as random.}
$ A Q_{t_1} = Y_{t_1}$ and $A^*M_t=X_t$,
we have the following lemma.

\begin{lem}\cite[Lemma $10$, Lemma $12$]{BayMont11}
For $0 \leq t \leq T^*$, the conditional distribution of the vectors in \eqref{eq:bmq} and \eqref{eq:hqm} satisfies the following, provided $n >t$ and $M_t$ and $Q_t$ have full column rank.
\begin{align*}
A^*m^t &|_{\mscrs_{t+1,t}} \stackrel{d}{=} X_t (M_t^* M_t)^{-1} M_t^* m^t_{\parallel}  \\
& + \,  Q_{t+1}(Q^*_{t+1} Q_{t+1})^{-1} Y^*_{t+1} m^\perp_t \, +\,  \mathsf{P}^\perp_{Q_{t+1}} \tilde{A}^* m^t_{\perp}, \\
A q^t &|_{\mscrs_{t,t}} \stackrel{d}{=} Y_t (Q_t^* Q_t)^{-1} Q_t^* q^t_{\parallel} \, + \, M_{t}(M^*_{t} M_{t})^{-1} X^*_{t} q^\perp_t \\
&+\,  \mathsf{P}^\perp_{M_t} \hat{A} q^t_{\perp},
\end{align*}
where $m^t_{\parallel}, m^\perp_t, q^t_{\|}, q^\perp_t$ are defined in \eqref{eq:mtqt_par} and \eqref{eq:mtqt_perp}.  Here $\tilde{A}, \hat{A} \stackrel{d}{=} A$ are random matrices independent of $\mscrs_{t+1,t}, \mscrs_{t,t}$, and $\mathsf{P}^\perp_{M_t}= \mathsf{I}-\mathsf{P}_{M_t}$ where $\mathsf{P}_{M_t}= M_t(M^*_t M_t)^{-1}M_t^*$ is the orthogonal projection matrix onto the column space of $M_t$; similarly,
$\mathsf{P}^\perp_{Q_{t+1}}= \mathsf{I}-\mathsf{P}_{Q_{t+1}}$, where $\mathsf{P}_{Q_{t+1}}= Q_{t+1} (Q^*_{t+1} Q_{t+1})^{-1}Q_{t+1}^*$.
\label{lem:A_conddist}
\end{lem}
The distributional characterization of $A^*m^t$ and $Aq^t$ in Lemma \ref{lem:A_conddist} together with \eqref{eq:bmq} and \eqref{eq:hqm}  leads to the following lemma.

\begin{lem}[Conditional Distribution Lemma]
For the vectors $h^{t+1}$ and $b^t$ defined in \eqref{eq:hqbm_def}, the following hold for $1 \leq t \leq T^*$, provided $n >t$ and $M_t$ and $Q_t$ have full column rank.
\begin{equation}
\begin{split}
& h^{1} \lvert_{\mscrs_{1, 0}} \stackrel{d}{=} \bar{\tau}_0 Z_0 + \Delta_{1,0},  \\
& h^{t+1} \lvert_{\mscrs_{t+1, t}} \stackrel{d}{=} \frac{\bar{\tau}_t^2}{\bar{\tau}_{t-1}^2} h^{t} + \bar{\tau}_{t}^{\perp} \, Z_t + \Delta_{t+1,t}, 
\end{split}
\label{eq:Ha_dist} 
\end{equation}
\begin{equation}
b^{0} \lvert_{\mscrs_{0, 0}} \stackrel{d}{=} \bar{\sigma}_0 Z'_0, \quad b^{t} \lvert_{\mscrs_{t, t}}\stackrel{d}{=} \frac{\bar{\sigma}_t^2}{\bar{\sigma}_{t-1}^2} b^{t-1} +  \bar{\sigma}_{t}^{\perp} \, Z'_t + \Delta_{t,t}. \label{eq:Ba_dist}
\end{equation}
where $Z_0, Z_t \in \mathbb{R}^N$ and $Z'_0, Z'_t \in \mathbb{R}^n$ are i.i.d.\ standard Gaussian random vectors that are independent of the corresponding conditioning sigma algebras. The deviation terms are 
\begin{equation}
\begin{split}
\Delta_{1,0} &= \left[ \left(\frac{\norm{m^0}}{\sqrt{n}}  - \bar{\tau}_0\right)\mathsf{I} -\frac{\norm{m^0}}{n} \mathsf{P}_{q^0}\right] Z_0 \\
&\qquad + q^0 \left(\frac{\norm{q^0}^2}{n}\right)^{-1} \left(\frac{(b^0)^*m_0}{n} - \frac{\norm{q^0}^2}{n}\right), \label{eq:D10}
\end{split}
\end{equation} 
and for $t >0$,
\be
\begin{split}
\Delta_{t,t} & =  \sum_{r=0}^{t-2} \gamma^t_r b^r + \left(\gamma^t_{t-1} -  \frac{\bar{\sigma}_t^2}{\bar{\sigma}_{t-1}^2}\right) b^{t-1}  \\
&+ \left[  \left(\frac{\norm{q^t_{\perp}}}{\sqrt{n}} - \bar{\sigma}_{t}^{\perp}\right) \mathsf{I}  - \frac{\norm{q^t_{\perp}} }{\sqrt{n}} \mathsf{P}_{M_t}\right]Z'_t    \\
& + M_t\Big(\frac{M_{t}^* M_{t}}{n}\Big)^{-1} \\
& \quad \cdot \left(\frac{H_t^* q^t_{\perp}}{n} - \frac{M_t^*}{n}\left[\lambda_t m^{t-1} - \sum_{r=1}^{t-1} \lambda_{r} \gamma^t_{r} m^{r-1}\right]\right),
\end{split}
\label{eq:Dtt} 
\ee
\be
\begin{split}
\Delta_{t+1,t} & =  \sum_{r=0}^{t-2} \alpha^t_r h^{r+1} +  \left(\alpha^t_{t-1} -  \frac{\bar{\tau}_t^2}{\bar{\tau}_{t-1}^2}\right) h^{t} \\
&+ \left[\left(\frac{\norm{m^t_{\perp}}}{\sqrt{n}} - \bar{\tau}_{t}^{\perp}\right)  \mathsf{I} -\frac{\norm{m^t_{\perp}}}{\sqrt{n}} \mathsf{P}_{Q_{t+1}}\right]Z_t \\
&+ Q_{t+1} \left(\frac{Q_{t+1}^* Q_{t+1}}{n}\right)^{-1} \\
&\quad \cdot  \left(\frac{B^*_{t+1} m_t^{\perp}}{n} - \frac{Q_{t+1}^*}{n}\left[q^t - \sum_{i=0}^{t-1} \alpha^t_i q^i\right]\right). 
\end{split}
\label{eq:Dt1t} 
\ee
\label{lem:hb_cond}
\end{lem}

\begin{proof}
We first demonstrate  \eqref{eq:Ba_dist}.  By \eqref{eq:hqbm_def} it follows that
\ben
b^{0}\lvert_{\mscrs_{0,0}} = -A\beta_0 = A q^0 \overset{d}{=} \frac{\norm{q^0}}{\sqrt{n}} Z'_0,
\een
where $Z'_0 \in \mathbb{R}^n$ is an i.i.d.\ standard Gaussian random vector, independent of $\mscrs_{0,0}$.  The result follows since $\norm{q^0} = \norm{\beta_0} = \sqrt{nP} = \sqrt{n \bar{\sigma}_0}.$

For the case $t \geq 1$, we use Lemma \ref{lem:A_conddist} to write
\be
\begin{split}
&b^t  \lvert_{\mscrs_{t, t}}  = (A q^t - \lambda_t m^{t-1}) \lvert_{\mscrs_{t, t}} \\
&\overset{d}{=} Y_t(Q_t^* Q_t)^{-1} Q_t^* q^t_{\parallel} + M_t(M_t^*M_t)^{-1} X_t^*  q_{\perp}^t \\
&\qquad + \mathsf{P}^{\perp}_{M_{t}} \tilde{A} q^t_{\perp} - \lambda_t m^{t-1}\\
& = B_t(Q_t^* Q_t)^{-1} Q_t^* q^t_{\parallel} + [ 0 | M_{t-1}] \Lambda_t (Q_t^* Q_t)^{-1} Q_t^* q^t_{\parallel} \\
&\qquad  + M_t(M_t^*M_t)^{-1} H_t^*  q_{\perp}^t + \mathsf{P}^{\perp}_{M_{t}} \tilde{A} q^t_{\perp}- \lambda_t m^{t-1}. \nonumber
\end{split}
\label{eq:lemma13a}
\ee
The last equality above is obtained using \eqref{eq:XtYt_rel}.  Noticing that $\mathsf{P}^{\perp}_{M_{t}} \tilde{A} q^t_{\perp} = (\mathsf{I}  - \mathsf{P}_{M_t})\tilde{A}q^t_{\perp}$ and $B_t(Q_t^* Q_t)^{-1} Q_t^* q^t_{\parallel} = \sum_{i=0}^{t-1} \gamma^t_i b^i$,  it follows that
\be
\begin{split}
& b^t |_{\mscrs_{t,t}} \\
& \stackrel{d}{=}   (\mathsf{I}  - \mathsf{P}^{\parallel}_{M_t})\tilde{A}q^t_{\perp} + \sum_{i=0}^{t-1} \gamma^t_i b^i +  [0 | M_{t-1}] \Lambda_t (Q_t^* Q_t)^{-1} Q^*_t q^t_{\parallel} \\
&\qquad  +  M_t(M_t^* M_t)^{-1} H_t^* q^t_{\perp} - \lambda_t m^{t-1} \\
&\overset{d}{=}  (\mathsf{I}  - \mathsf{P}^{\parallel}_{M_t}) \frac{\norm{q^t_{\perp}}}{\sqrt{n}} Z'_t + \sum_{i=0}^{t-1} \gamma^t_i b^i  + M_t(M_t^* M_t)^{-1} H_t^* q^t_{\perp}\\
& + [0 | M_{t-1}] \Lambda_t (Q_t^* Q_t)^{-1} Q^*_t q^t_{\parallel}  - \lambda_t m^{t-1},
\end{split}
\label{eq:Bbtdef}
\ee
where $Z'_t \in \mathbb{R}^n$ is an i.i.d.\ standard Gaussian random vector. All the quantities in the RHS of \eqref{eq:Bbtdef} except $Z'_{t}$ are in the conditioning sigma-field.  We can rewrite \eqref{eq:Bbtdef} as 
\begin{align*}
b^{t}\lvert_{\mscrs_{t, t}} \overset{d}{=}& \frac{\bar{\sigma}_t^2}{\bar{\sigma}_{t-1}^2} b^{t-1} + \bar{\sigma}_{t}^{\perp} Z'_t + \Delta_{t,t}, 
\end{align*}
where
\begin{align*}
\Delta_{t,t} &= \sum_{r=0}^{t-2} \gamma^t_r b^r +  \left(\gamma^t_{t-1} - \frac{\bar{\sigma}_t^2}{\bar{\sigma}_{t-1}^2}\right) b^{t-1}  \\
&+ \left[ \left(\frac{\norm{q^t_{\perp}}}{\sqrt{n}} - \bar{\sigma}_{t}^{\perp}\right)\mathsf{I}  - \frac{\norm{q^t_{\perp}} }{\sqrt{n}} \mathsf{P}_{M_t}\right]Z'_t   \\
&+ [0 | M_{t-1}] \Lambda_t (Q_t^* Q_t)^{-1} Q^*_t q^t_{\parallel} + M_t(M_t^* M_t)^{-1} H_t^* q^t_{\perp} \\
& - \lambda_t m^{t-1}. 
\end{align*}
The above definition of $\Delta_{t,t}$ equals that given in \eqref{eq:Dtt} since
\begin{align*}
&  M_t\left(\frac{M_{t}^* M_{t}}{n}\right)^{-1}  \frac{M_t^*}{n}\left(\lambda_t m^{t-1} - \sum_{i=0}^{t-2} \lambda_{i+1} \gamma^t_{i+1} m^i\right) \nonumber 
\\
& \ +  [0 | M_{t-1}] \Lambda_t (Q_t^* Q_t)^{-1} Q^*_t q^t_{\parallel}  - \lambda_t m^{t-1} \nonumber \\
&= \lambda_t m^{t-1} - \sum_{i = 0}^{t-2} \lambda_{i+1} \gamma^t_{i+1} m^i + \sum_{j=0}^{t-2} \lambda_{j+1} \gamma^t_{j+1} m^j  - \lambda_t m^{t-1} \\
&= 0. \nonumber
\end{align*}
This completes the proof of \eqref{eq:Ba_dist}.  Result \eqref{eq:Ha_dist} can be shown similarly.  
\end{proof}

The conditional distribution representation in Lemma \ref{lem:hb_cond} implies that for each $t \geq 0$, $h^{t+1}$ is the sum of an i.i.d.\ $\mc{N}(0, \bar{\tau}_t^2)$ random vector plus a deviation term.  Indeed, if we assume that $h^{t}$ has the representation  $\bar{\tau}_{t-1} \breve{Z}_{t-1} +  \Delta_{t}$, then Lemma \ref{lem:hb_cond} implies
\be
\begin{split}
h^{t+1} \lvert_{\mscrs_{t+1, t}} &\stackrel{d}{=} \frac{\bar{\tau}_t^2}{\bar{\tau}_{t-1}^2} h^{t} + \tau^{\perp}_t Z_t + \Delta_{t+1,t} \\
&\stackrel{d}{=}   \frac{\bar{\tau}_t^2}{\bar{\tau}_{t-1}} \breve{Z}_{t-1}   +  \bar{\tau}^{\perp}_t Z_t  + \Delta_{t}  + \Delta_{t+1,t}  \stackrel{d}{=}  \bar{\tau}_t \breve{Z}_t.
\label{eq:bt_simp}
\end{split}
\ee
 To obtain the last equality, we combine the independent Gaussians $\breve{Z}_{t-1}$ and  $Z_t$ using the expression for $\bar{\tau}_t^{\perp}$ in \eqref{eq:sigperp_defs}.   It can be similarly seen that $b^t$ is the sum of an i.i.d.\ $\mc{N}(0, \bar{\sigma}_t^2)$ random vector and a deviation term.  The next lemma  shows that these deviation terms are $o(n^{-\delta})$ for some $\delta > 0$. 

\subsection{Main Convergence Lemma}
 \begin{defi}
 A function $\phi: \mathbb{R}^m \to \mathbb{R}$ is pseudo-Lipschitz of order $k$ (denoted by $\phi \in  PL(k)$) if there exists a constant $C >0$ such that for all $x,y \in \mathbb{R}^m$,
 \be  \abs{\phi(x) - \phi(y)} \leq C ( 1 + \norm{x}^{k-1} + \norm{y}^{k-1} ) \norm{x-y} . \ee
 \end{defi}

In the lemma below, $\delta \in (0,\tfrac{1}{2})$ is a generic positive number whose exact value is not required.   The value of $\delta$ in each  statement of the lemma may be different. We will say that a  sequence $x_n$ converges  to a constant $c$ at rate $n^{-\delta}$ if  $ \lim_{n \to \infty} n^{\delta}(x_n -c) = 0$.
 
 \begin{lem}
The following  statements hold for $0 \leq t \leq T^*$, where $T^*= \left \lceil \frac{2\mc{C}}{\log\left(\mc{C}/R \right)} \right \rceil$. 
\begin{enumerate}[(a)]
\item The following statements hold almost surely:
\begin{equation}
\begin{split}
&\max_{j \in sec(\ell)} \abs{[\Delta_{t+1, t}]_{j}} = o\left( n^{-\delta} \sqrt{\log M}\right), \\
& \max_{j \in sec(\ell)} \abs{h^{t+1}_{j}} \leq c_{t+1} \sqrt{\log M} \quad \text{for } \ell \in [L], \label{eq:Ha} 
\end{split}
\end{equation}
\begin{equation}
\lim \frac{\norm{\Delta_{t,t}}^2}{n} = 0, \label{eq:Ba}
\end{equation}
where $c_{t+1} > 0$ is a constant not depending on $N,n$. The convergence rate in \eqref{eq:Ba} is $n^{-\delta}$.

\item  
i) Consider the following functions defined on $\mathbb{R}^{M} \times \mathbb{R}^{M} \times \mathbb{R}^{M} \rightarrow \mathbb{R}$. For $x,y,z \in \mathbb{R}^M$,  $-1\leq r \leq s \leq t$, and $\ell \in [L]$, let
\be
\begin{split}
\phi_{1, \ell }(x,y,z) &:= x^*y/ M, \\
\phi_{2,\ell }(x,y,z)&:= \norm{\eta^{r}_\ell(z-x)}^2/ \log M, \\
\phi_{3, \ell }(x,y,z)&:=  {[} \eta^r_\ell( z-x) - z {]}^*{[} \eta^s_\ell (z-y) - z ]/ \log M , \\
\phi_{4, \ell }(x,y,z) &:=  y^*{[}  \,\eta^r_{\ell}(z-x) - z]/ \log M, \\
\end{split}
\label{eq:phih_fns}
\ee
where for $r \geq 0$,  $\eta^r_\ell(\cdot)$ is the restriction of $\eta^r$ to section $\ell$, i.e.,  for $x \in \mathbb{R}^M$, 
\[  \eta^{r}_{\ell, i}(x) := \sqrt{n P_\ell} \,  \frac{\exp\left(\frac{x_i \sqrt{n P_\ell}}{\tau^2_r} \right)}
{\sum_{j=1}^M \, \exp\left(\frac{x _j \sqrt{n P_\ell}}{\tau^2_r} \right)},   \ i=1, \ldots,M. \]
(Also, $\eta^{-1}_{\ell,i}(\cdot) := 0$ for $i \in [M]$.) Then, for $k \in \{1,2,3,4 \}$ and arbitrary constants $(a_0, \ldots, a_{t},$ $b_0, \ldots, b_t)$, we have 
\begin{equation}
\begin{split}
&\lim n^\delta \left | \frac{1}{L} \sum_{\ell=1}^L \phi_{k, \ell} \Big( \sum_{r=0}^{t} a_r h^{r+1}_{\ell}, \sum_{s=0}^t  b_s h^{s+1}_{\ell}, \beta_{0_{\ell}}\Big) -c_k \right |
\label{eq:Hb2}
\end{split}
\end{equation}
almost surely equals $0$, where 
\[  c_k := \lim \frac{1}{L} \sum_{\ell=1}^L \mathbb{E} \left[  \phi_{k, \ell}\left(\sum_{r=0}^t a_r \bar{\tau}_r \breve{Z}_{r_{\ell}}, \sum_{s=0}^t b_s \bar{\tau}_s \breve{Z}_{s_{\ell}}, \beta_{\ell}\right)\right]  \]
Here $\breve{Z}_0, ..., \breve{Z}_t$ are length-$N$ Gaussian random vectors independent of $\beta$, and $\breve{Z}_{r_\ell}, \beta_\ell, \beta_{0,\ell}, h^{r+1}_{\ell}$ denote the $\ell$th section of the respective vectors. For
$0\leq s \leq t$,  $\{ \breve{Z}_{s_j}\}_{j \in [N]}$ are  i.i.d.\ $\sim \mc{N}(0,1)$, and for each $ i \in [N]$, $(\breve{Z}_{0_i}, \ldots, \breve{Z}_{t_i})$ are jointly Gaussian with $\mathbb{E}[\bar{\tau}_r \breve{Z}_{r_i} \bar{\tau}_t \breve{Z}_{t_i}] = \bar{\tau}_t^2$ for $0 \leq r \leq t$. The limit defining $c_k$ exists and is finite for each $\phi_{k,\ell}$ in \eqref{eq:phih_fns}.

ii) For all pseudo-Lipschitz functions $\phi_b: \mathbb{R}^{t + 2} \rightarrow \mathbb{R}$ of order two, we have 
\begin{equation}
\begin{split}
\lim n^{\delta} & \Big[   \frac{1}{n} \sum_{i=1}^n \phi_b(b_i^0, ..., b_i^t, w_i)  \\ 
& \quad  -  \mathbb{E}[ \phi_b(\bar{\sigma_0} \hat{Z}_0, ..., \bar{\sigma_t} \hat{Z}_t, \sigma Z_w)] \Big] =0 \quad a.s.
\end{split}
 \label{eq:Bb}
\end{equation}
The  random variables $(\hat{Z}_0, ...,\hat{Z}_t)$  are jointly Gaussian with 
$\hat{Z}_s \sim \mc{N}(0,1)$ for $0\leq s \leq t$ and $\mathbb{E}[\bar{\sigma_s} \hat{Z}_s \bar{\sigma_t} \hat{Z}_t] = \bar{\sigma}_t^2$. Further,  $( \hat{Z}_0, ...,\hat{Z}_t )$  are independent of  $Z_w \sim \mc{N}(0,1)$.


\item 
\begin{align}
&\lim \frac{(h^{t+1})^* q^0}{n}\overset{a.s.}{=} 0,  \label{eq:Hc}\\
&\lim \frac{(b^{t})^* w}{n}\overset{a.s.}{=} 0.  \label{eq:Bc}
\end{align}
The convergence rate in both \eqref{eq:Hc} and \eqref{eq:Bc}  is $n^{-\delta}$.


\item For all $0 \leq r  \leq t$, 
\begin{align}
 & \lim \frac{(h^{r+1})^* h^{t+1}}{N} \overset{a.s}{=}  \bar{\tau}_{t}^2, \label{eq:Hd} \\
& \lim \frac{(b^{r})^* b^{t}}{n}  \overset{a.s.}{=} \bar{\sigma}^2_{t}, \label{eq:Bd}
\end{align}
where $\bar{\sigma}_s$ is defined in \eqref{eq:sigt_def}. The convergence rate in both \eqref{eq:Hd} and \eqref{eq:Bd} is $n^{-\delta}$.


\item For all $0 \leq r \leq t$, 
\begin{align}
&\lim \frac{(q^{0})^* q^{t+1}}{n}  \overset{a.s.}{=} \bar{\sigma}^2_{t+1}, \quad \lim \frac{(q^{r+1})^* q^{t+1}}{n}  \overset{a.s.}{=} \bar{\sigma}^2_{t+1}, \label{eq:He} \\
 & \lim \frac{(m^r)^* m^t}{n} \overset{a.s.}{=} \bar{\tau}_t^2. \label{eq:Be}
\end{align}
The convergence rate in both \eqref{eq:He} and \eqref{eq:Be} is $n^{-\delta}$.

\item For all $0 \leq r, s \leq t$, 
\begin{align}
 \lim \frac{(h^{s+1})^* q^{r+1}}{n} &\overset{a.s}{=} \lim \lambda_{r+1}\lim \frac{(m^r)^* m^s}{n} \nonumber \\
 &\overset{a.s.}{=} - \frac{\bar{\sigma}^2_{r+1} \bar{\tau}_{\max(r,s)}^2}{\bar{\tau}_r^2},
\label{eq:Hf} \\
 \lim \frac{(b^{r})^* m^{s}}{n} &\overset{a.s}{=} \bar{\sigma}^2_{\max(r,s)}.
\label{eq:Bf}
\end{align}
The convergence rate in both \eqref{eq:Hf} and \eqref{eq:Bf}  is $n^{-\delta}$.

\item The vectors $(\gamma^{t+1}_0, \ldots, \gamma^{t+1}_t)$ and $(\alpha^{t}_0, \ldots,  \alpha^{t}_{t-1})$ converge entry-wise to the following limits at rate $n^{-\delta}$.
\begin{align}
& \lim (\gamma^{t+1}_0, \ldots, \gamma^{t+1}_{t-1}, \gamma^{t+1}_t) \overset{a.s.}{=} \left(0, \ldots, 0, \frac{\bar{\sigma}_{t+1}^2}{\bar{\sigma}_{t}^2}\right) ,  \label{eq:Hg}\\
&\lim (\alpha^{t}_0, \ldots, \alpha^{t}_{t-2}, \alpha^{t}_{t-1}) \overset{a.s.}{=} \left(0, \ldots, 0, \frac{\bar{\tau}_{t}^2}{\bar{\tau}_{t-1}^2}\right), \ t \geq 1.  \label{eq:Bg}
\end{align}

\item 
\begin{align}
& \lim \frac{\norm{q_{\perp}^{t+1}}^2}{n} \overset{a.s.}{=} (\bar{\sigma}^{\perp}_{t+1})^2, \label{eq:Hh} \\
&  \lim \frac{\norm{m_{\perp}^t}^2}{n} \overset{a.s.}{=} (\bar{\tau}^{\perp}_{t})^2, \label{eq:Bh} 
\end{align}
where $\bar{\sigma}^{\perp}_{t}, \bar{\tau}^{\perp}_{t}$, defined in \eqref{eq:sigperp_defs}, are strictly positive for $t \leq T^*$. The convergence rate in both \eqref{eq:Hh} and \eqref{eq:Bh}  is $n^{-\delta}$.
\end{enumerate}
\label{lem:main_lem}
\end{lem}

The lemma is proved in Section \ref{sec:lem1_proof}. 

\subsection{Comments on Lemmas \ref{lem:hb_cond} and \ref{lem:main_lem}} \label{subsec:lem_comments}

To prove Theorem \ref{thm:main_amp_perf}, the main result we need from Lemma \ref{lem:main_lem} is that for each $t >0$, $\frac{\norm{q^t}^2}{n} = 
\frac{\norm{\beta^t - \beta_0}^2}{n}$ converges to $\bar{\sigma}_t^2$ with probability $1$. This result is used   in Section \ref{subsec:proof_thm1} below to prove Theorem \ref{thm:main_amp_perf}.  The convergence of $\frac{\norm{q^t}^2}{n}$ is shown in part (e) of Lemma \ref{lem:main_lem} by appealing to part (b).i, which shows that within the functions  listed in \eqref{eq:phih_fns}, 
 $h^{t+1} = \beta_0-(A^* z^t + \beta^t)$ (the difference between the true signal and the test statistic) can be replaced by $\bar{\tau_t}\breve{Z}_t$ in the large system limit.
 
 While the results in Lemmas \ref{lem:hb_cond} and \ref{lem:main_lem} are similar to those found in \cite[Lemma $1$]{BayMont11}, there are a few key differences.
\begin{itemize}

\item The functions  listed in \eqref{eq:phih_fns}  all act \emph{section-wise}  on the vectors $\{ h^t \}_{t >0}$. Recall that the structure of $\beta_0$ implies that  $h^t \in \mathbb{R}^{ML}$ are section-wise independent, where the section size  $M=L^\mathsf{a}= \Theta((n/\log n)^{\mathsf{a}})$.  This is in contrast to the  functions considered in \cite{BayMont11,bayMontLASSO} (and in part (b).ii), which act \emph{component-wise} on vectors whose components are i.i.d. 

\item To prove part (b).i of Lemma \ref{lem:main_lem} for the section-wise functions in \eqref{eq:phih_fns}, we first need to show that the  deviation terms $\Delta_{t+1, t}$ (defined in Lemma \ref{lem:hb_cond}) can be  neglected in the large system limit. This is done by showing  in part (a) of Lemma \ref{lem:main_lem} (see \eqref{eq:Ha}) that
\ben
\max_{j \in sec(\ell)} \abs{[\Delta_{t+1, t}]_{j}} = o\left( n^{-\delta} \sqrt{\log M}\right).
\een
To prove  this, we require the inner product convergence results given the other parts of the lemma to hold with a convergence rate of $n^{-\delta}$ for some $\delta >0$. This is another difference from \cite[Lemma $1$]{BayMont11}, where a minimum rate of convergence was not needed. In our case, without an $n^{-\delta}$ convergence rate, we would only have that the deviation terms satisfied $\max_{j \in sec(\ell)} \abs{[\Delta_{t+1, t}]_{j}} = o(\sqrt{\log M})$, and we would not be able to neglect them.

\item Other  differences between Lemmas \ref{lem:hb_cond},\ref{lem:main_lem} and \cite[Lemma 1]{BayMont11} include:
\begin{itemize} 
\item[--] Lemma \ref{lem:hb_cond} characterizes the the conditional distribution of the vectors $h^{t+1}$ and $b^t$, given the matrices in \eqref{eq:XYMQt} as well as $\beta_0$ and $w$, as the sum of an ideal distribution and a deviation term.  Lemma \ref{lem:hb_cond} should be compared to \cite[Lemma $1$(a)]{BayMont11}, which is  a similar  distributional characterization of $h^{t+1}$ and $b^t$, however it does not use the ideal distribution.  We found that working with the ideal distribution throughout Lemma \ref{lem:main_lem} simplified our proof.

\item[--] Lemma \ref{lem:main_lem} gives explicit values for the deterministic limits in parts (c)--(h), which are required in other parts of our proof.  
\end{itemize}
\end{itemize}

\subsection{Proof of  Theorem \ref{thm:main_amp_perf} } \label{subsec:proof_thm1}
From the definition in \eqref{eq:sec_err_def}, the event that the section error rate is larger than $\e$ can be written as
\be
\{ \mc{E}_{sec}(\mc{S}_n)  > \e \} =  \left\{ \sum_{\ell =1}^{L}  \mathbf{1} \{ \hat{\beta}_\ell \neq \beta_{0_\ell} \} > L \e \right\}.
\label{eq:sec_err_event}
\ee
When a section $\ell$ is decoded in error, the correct non-zero entry has no more than half the total mass of section
$\ell$ at the termination step $T^*$. That is, $\beta^{T^*}_{\textsf{sent}(\ell)} \leq  \frac{1}{2} \sqrt{n P_ \ell}$
where $\textsf{sent}(\ell)$ is the index of the non-zero entry in section $\ell$ of the true message $\beta_0$.  Since $\beta_{0_\textsf{sent}(\ell)} = \sqrt{n P_\ell}$, we therefore have
\be
\mathbf{1} \{ \hat{\beta}_\ell \neq \beta_{0_\ell} \} \ \  \Rightarrow \ \ \norm{ \beta^{T^*}_\ell - \beta_{0_\ell}}^2 \geq  \frac{n P_\ell}{4}, \quad \ell \in [L].
\label{eq:sec_error_implies}
\ee
Hence when \eqref{eq:sec_err_event} holds, we  have
\be
\begin{split}
&\norm{\beta^{T^*} - \beta_0}^2 = \sum_{\ell=1}^L \, \norm{\beta^{T^*}_\ell - \beta_{0_\ell}}^2 \stackrel{(a)}{\geq}
\sum_{\ell=1}^L   \mathbf{1} \{ \hat{\beta}_\ell \neq \beta_{0_\ell} \}  \frac{nP_\ell}{4} \, \\
&\quad \stackrel{(b)}{\geq} \, L\e \frac{nP_L}{4}  \stackrel{(c)}{\geq} \,
\frac{n \, \e \, \sigma^2 \ln(1 + \snr)}{4},
\label{eq:sec_error_chain}
\end{split}
\ee
where $(a)$ follows from   \eqref{eq:sec_error_implies}; $(b)$ is obtained using \eqref{eq:sec_err_event}, and the fact that
$P_\ell > P_L$ for $\ell \in [L-1]$ for the exponentially decaying power allocation in \eqref{eq:exp_power_alloc}; $(c)$ is obtained using the first-order Taylor series lower bound
$L P_L \geq \sigma^2 \ln(1+\tfrac{P}{\sigma^2})$. We therefore conclude that
\be
\{ \mc{E}_{sec}(\mc{S}_n)  > \e \} \ \Rightarrow \  \left\{ \frac{\norm{\beta^{T^*} - \beta_0}^2 }{n}  \geq \frac{\e \, \sigma^2 \ln(1 + \snr)}{4} \right\}.
\label{eq:sec_error_exp}
\ee
Now, from \eqref{eq:He} of Lemma \ref{lem:main_lem}(e), we know that
\be
\lim  \frac{\norm{\beta^{T^*} - \beta_0}^2 }{n}  = \lim \frac{\norm{q^{T^*}}^2}{n}  \stackrel{a.s.}{=}   P(1 - \bar{x}_{T^*} )\stackrel{(a)}{=} 0,
\label{eq:betaTst_conv}
\ee
where $(a)$ follows from Lemma \ref{lem:lim_xt_taut}, which implies that  $\xi_{T^*- 1}=1$ for  $T^*=\left\lceil \frac{2 \mc{C}}{\log(\mc{C}/R)} \right\rceil$,  and hence $\bar{x}_{T^*} =1$. Thus we have shown in  \eqref{eq:betaTst_conv} that
$\frac{\norm{\beta^{T^*} - \beta_0}^2 }{n}$ converges almost surely to zero, i.e.,
\be
\lim_{n_0 \to \infty} P\left( \frac{\norm{\beta^{T^*} - \beta_0}^2 }{n}  < \e, \ \forall n \geq n_0 \right) = 1
\ee
for any $\e >0$. From \eqref{eq:sec_error_exp}, this implies that for $\e'= \frac{4\e}{\sigma^2 \ln(1 + \snr)}$, 
\be
\lim_{n_0 \to \infty} P\left(  \mc{E}_{sec}(\mc{S}_n)   \leq  \e', \ \forall n \geq n_0 \right) = 1.
\ee

\section{Proof of Lemma  \ref{lem:main_lem}} \label{sec:lem1_proof}

\subsection{Useful Probability and Linear Algebra Results}

We  list some results that will be used  in the proof of Lemma \ref{lem:main_lem}. Most of these can be found in \cite[Section III.G]{BayMont11}, but we summarize them here for completeness.

\begin{fact}
Let  $u \in \mathbb{R}^N$ and $v \in \mathbb{R}^n$ be deterministic vectors such that $\lim_{n \to \infty} \norm{u}^2/n$ and $\lim_{n \to \infty} \norm{v}^2/n$ both exist and are finite. Let $\tilde{A} \in \mathbb{R}^{n \times N}$ be a matrix with independent  $\mc{N}(0, 1/n)$ entries. Then:

(a)
\begin{equation}
\tilde{A} u \overset{d}{=} \frac{\norm{u}}{\sqrt{n}} Z_u  \quad \text{ and }  \quad \tilde{A}^*v \overset{d}{=} \frac{\norm{v}}{\sqrt{n}} Z_v,
\label{eq:Au_dist}
\end{equation}
where $Z_u \in \mathbb{R}^n$ and $Z_v \in \mathbb{R}^N$  are each i.i.d. standard Gaussian random vectors. Consequently,
\begin{equation}
\lim_{n \rightarrow \infty} \frac{\norm{\tilde{A}u}^2}{n} \overset{a.s.}{=}  \lim_{n \rightarrow \infty}  \frac{\norm{u}^2}{n} \sum_{i=1}^n \frac{Z^2_{u,i}}{n} \overset{a.s.}{=}  \lim_{n \rightarrow \infty}  \frac{\norm{u}^2}{n},
\end{equation}
\begin{equation}
\label{fact1b}
\lim_{n \rightarrow \infty} \frac{\norm{\tilde{A}^*v}^2}{N} \overset{a.s.}{=} \lim_{n \rightarrow \infty}  \frac{\norm{v}^2}{n} \sum_{j=1}^N \frac{Z^2_{v,j}}{N} \overset{a.s.}{=}  \lim_{n \rightarrow \infty}  \frac{\norm{v}^2}{n}.
\end{equation}

(b) Let $\mc{W}$  be a $d$-dimensional subspace of $\mathbb{R}^n$ for $d \leq n$. Let $(w_1, ..., w_d)$ be an orthogonal basis of $\mc{W}$ with $\norm{w_i}^2 = n$ for
$i \in [d]$, and let  $\mathsf{P}_\mc{W}$ denote  the orthogonal projection operator onto $\mc{W}$.  Then for $D = [w_1\mid \ldots \mid w_d]$, we have $\mathsf{P}_{\mc{W}} \tilde{A} u \overset{d}{=} \frac{\norm{u}}{\sqrt{n}}  \mathsf{P}_\mc{W} Z_u  \overset{d}{=} \frac{\norm{u}}{\sqrt{n}} Dx$ where $x \in \mathbb{R}^d$ is a random vector with i.i.d. $\mc{N}(0, 1/n)$ entries. Therefore
$\lim_{n \to \infty} n^{\delta}\norm{x} \stackrel{a.s.}{=} 0$ for any constant $\delta \in [0,0.5)$. (The limit is taken with $d$ fixed.)
\label{fact:gauss_p0}
\end{fact}

\begin{fact}[Strong Law for Triangular Arrays]
Let $\{ X_{n,i}: \, i \in [n], n \geq 1\}$ be a triangular array of random variables  such that for each $n$ $(X_{n,1}, \ldots, X_{n,n})$ are mutually independent, have zero mean, and satisfy
\be
\frac{1}{n} \sum_{i=1}^n \expec \abs{X_{n,i}}^{2+\kappa} \leq c n^{\kappa/2} \quad \text{ for some } \kappa \in (0,1) \text { and } c < \infty.
\ee
Then $\frac{1}{n} \sum_{i=1}^n X_{n,i} \to 0$ almost surely as $n \to \infty$.
\label{fact:slln}
\end{fact}

\begin{fact}
Let $v \in \mathbb{R}^n$ be a random vector with i.i.d. entries $\sim p_V$ where the measure $p_V$ has bounded second moment. Then for any function $\psi$ that is pseudo-Lipschitz  of order two: 
\be
\lim_{n \to \infty} \frac{1}{n} \sum_{i=1}^n \psi(v_i) \stackrel{a.s.}{=}  \expec_{p_V}[\psi(V)]
\ee
with convergence rate $n^{-\delta}$, for some $\delta \in (0,1/4)$.
\label{fact:lip_slln}
\end{fact}

\begin{fact}[Stein's lemma]
For zero-mean jointly Gaussian random variables $Z_1, Z_2$, and any function $f:\mathbb{R} \to \mathbb{R}$ for which $\expec[Z_1 f(Z_2)]$ and $\expec[f'(Z_2)]$  both exist, we have $\expec[Z_1 f(Z_2)] = \expec[Z_1Z_2] \expec[f'(Z_2)]$.
\label{fact:stein}
\end{fact}

\begin{fact}
Let $v_1, \ldots, v_t$ be a sequence of vectors in $\mathbb{R}^n$ such that for $i \in [t]$
\[ \frac{1}{n} \norm{v_i - \mathsf{P}_{i-1}(v_i)}^2 \geq c, \]
where $c$ is a positive constant and $\mathsf{P}_{i-1}$ is the orthogonal projection onto the span of $v_1, \ldots, v_{i-1}$.Then the matrix $C \in \mathbb{R}^{t \times t}$
with $C_{ij} = v^*_i v_j / n$ has minimum eigenvalue $\lambda_{\min} \geq c'$, where $c'$ is a strictly positive constant (depending  only on $c$ and $t$).
\label{fact:eig_proj}
\end{fact}

\begin{fact}
Let $\{ S_n \}_{n \geq 1}$ be a sequence of $t \times t$ matrices such that $\lim_{n \to \infty} S_n = S_{\infty}$ where the limit is element-wise. Then if $\liminf_{n \to \infty} \lambda_{\min}(S_n) \geq c$ for a positive constant $c$, then 
$\lambda_{\min}(S_\infty) \geq c$.
\label{fact:eig_conv}
\end{fact}

\begin{fact}
\label{lem:maxZbound}
Let $Z_1, Z_2, \ldots$ be i.i.d. standard Gaussian random variables. For any constant $K >1$, with probability $1$ we have 
\[ \max_{j \in [M]} \, \abs{Z_j}  \leq \sqrt{ 2 K \log M} \text{  for all  sufficiently large } M. \]
\end{fact}

\begin{proof}
 For  $x > 0$, we have
$P( \max_{j \in [M]} Z_j \, > x) = 1 - (P(Z_1  \leq x ))^M = 1 - (1 - \mc{Q}(x))^M$,
where $\mc{Q}(x) =  \int_{x}^\infty \frac{1}{\sqrt{2\pi}} e^{-u^2/2} du$. Using  $\mc{Q}(x) < e^{-x^2/2}$ for $x >0$ and setting $x = \sqrt{2K \ln M}$, we obtain 
\ben
P( \max_{j \in [M]} Z_j \, > \sqrt{2K \ln M}) \leq 1 - \left(1 - \frac{1}{M^K}\right)^M \leq  
\frac{1}{M^{K-1}},
\een
where we have used $(1-y)^M \geq (1-My)$ for $y \in (0,1)$. Hence for $K > 1$, we have
\ben
\sum_{M =1}^{\infty} P( \max_{j \in [M]} Z_j \, > \sqrt{2K \ln M}) \leq  \sum_{M=1}^\infty \frac{1}{M^{K-1}} < \infty.
\een
Therefore the Borel-Cantelli lemma implies that  {with probability 1},  the event $\{ \max_{j \in [M]} Z^{(M)}_j \, > \sqrt{2K \ln M}\}$ occurs only for finitely many $M$. By a symmetrical argument, we can show that with probability 1, the event  $\{ \min_{j \in [M]} Z_j \, < - \sqrt{2K \ln M}\}$ also occurs only for finitely many $M$.
\end{proof}

\subsection{Proof of Lemma \ref{lem:main_lem}} 

The proof proceeds by induction on $t$.  We label as $\mathcal{H}^{t+1}$ the results \eqref{eq:Ha}, \eqref{eq:Hb2}, \eqref{eq:Hc}, \eqref{eq:Hd}, \eqref{eq:He}, \eqref{eq:Hf}, \eqref{eq:Hg}, \eqref{eq:Hh} and similarly as $\mathcal{B}^t$ the results \eqref{eq:Ba}, \eqref{eq:Bb}, \eqref{eq:Bc}, \eqref{eq:Bd}, \eqref{eq:Be}, \eqref{eq:Bf}, \eqref{eq:Bg}, \eqref{eq:Bh}.  The proof consists of four steps:

\begin{enumerate}

\item $\mathcal{B}_0$ holds.

\item $\mathcal{H}_1$ holds.

\item If $\mathcal{B}_r, \mathcal{H}_s$ holds for all $r < t $ and $s \leq t $, then $\mathcal{B}_t$ holds.

\item if $\mathcal{B}_r, \mathcal{H}_s$ holds for all $r \leq t $ and $s \leq t$, then $\mathcal{H}_{t+1}$ holds.
\end{enumerate}

\subsubsection{Step 1: Showing $\mathcal{B}_0$ holds} \label{subsub:step1}
We wish to show that \eqref{eq:Ba}, \eqref{eq:Bb}, \eqref{eq:Bc}, \eqref{eq:Bd}, \eqref{eq:Be}, \eqref{eq:Bf}, \eqref{eq:Bg}, and \eqref{eq:Bh} hold when $t = 0$.

\textbf{(a)} $\Delta_{0,0} = 0$ so there is nothing to prove.

\textbf{(b)} From Lemma \ref{lem:hb_cond} we note $b^0 \stackrel{d}{=} \bar{\sigma}_0 Z$ where $Z \in \mathbb{R}^n$ is a standard Gaussian vector.  We will first use Fact  \ref{fact:slln} to show that 
\begin{equation}
\begin{split}
 \lim n^{\delta} & \left[\frac{1}{n}\sum_{i=1}^n \phi_b(\bar{\sigma}_0 Z_i, w_i)   \right. \\
&\quad \left. - \frac{1}{n}\sum_{i=1}^n \mathbb{E}_{Z}\left\{\phi_b(\bar{\sigma}_0 Z_i, w_i)\right\}\right] =0 \quad  a.s. ,
\end{split}
\label{eq:1b2}
\end{equation}
 Let $\tilde{Z}$ be an independent copy of $Z$.  To apply Fact \ref{fact:slln}, we need to verify that
\begin{equation*}
\frac{1}{n} \sum_{i=1}^n \mathbb{E}\lvert n^{\delta} \phi_b(\bar{\sigma}_0 \tilde{Z}_i, w_i) - n^{\delta} \mathbb{E}_{Z}\left\{ \phi_b(\bar{\sigma}_0 Z_i, w_i)\right\}\lvert^{2 + \kappa}   \, \leq \, cn^{\kappa/2}.
\end{equation*}
for some constants $\kappa \in (0,1)$ and $c>0$.  Dropping the subscript $i$ on $\tilde{Z}, Z$ for brevity, we have
\begin{small}
\begin{equation}
\begin{split}
 &\mathbb{E}_{\tilde{Z}}\abs{ \phi_b(\bar{\sigma}_0 \tilde{Z}, w_i) - \mathbb{E}_{Z}\left\{\phi_b(\bar{\sigma}_0 Z, w_i)\right\}}^{2 + \kappa}\\
 & \overset{(a)}{\leq} \mathbb{E}_{\tilde{Z}, Z}\left \lvert \phi_b(\bar{\sigma}_0 \tilde{Z}, w_i) - \phi_b(\bar{\sigma}_0 Z, w_i)\right\lvert^{2 + \kappa} \\
& \overset{(b)}{\leq} c' \abs{\bar{\sigma}_0}^{2 + \kappa}\,  \mathbb{E}_{\tilde{Z}, Z}\left\{ | \tilde{Z} - Z|^{2 + \kappa}\left( 1 +  |\bar{\sigma}_0 \tilde{Z}| + |w_i| + |\bar{\sigma}_0 Z|\right)^{2 + \kappa}\right\} \\
& \leq c_0  \abs{\bar{\sigma}_0}^{2 + \kappa} \left[  \mathbb{E}_{\tilde{Z}, Z} \left\{|\tilde{Z} - Z|^{2 + \kappa} \left(1 +  |\bar{\sigma}_0 \tilde{Z}|^{2 + \kappa} + |\bar{\sigma}_0 Z|^{2 + \kappa}\right)\right\}  \right.\\
&\qquad \left. +  |w_i|^{2 + \kappa}  \mathbb{E}_{\tilde{Z}, Z}\left\{|\tilde{Z} - Z|^{2 + \kappa}\right\} \right] \\
& \overset{(c)}{\leq} c_1 + c_2 |w_i|^{2 + \kappa},
\end{split}
\label{eq:upperbound0}
\end{equation}
\end{small}
where $c', c_0,c_1, c_2$ are positive constants. In the chain above, $(a)$ uses Jensen's inequality, $(b)$ holds because $\phi_b \in PL(2)$, and $(c)$ uses the fact that $Z,\tilde{Z}$ are i.i.d.\ $\mc{N}(0,1)$. Using \eqref{eq:upperbound0}, we obtain 
\ben
\begin{split}
\frac{1}{n} &\sum_{i=1}^n \mathbb{E}\lvert n^{\delta} \phi_b(\bar{\sigma}_0 \tilde{Z}, w_i) - n^{\delta} \mathbb{E}_{Z}\left\{ \phi_b(\bar{\sigma}_0 Z, w_i)\right\}\lvert^{2 + \kappa}  \\
& \qquad \leq \frac{n^{\delta (\kappa +2)}}{n}  \sum_{i=1}^n (c_1 + c_2 |w_i|^{2 + \kappa}) \leq c n^{\kappa/2},
\end{split}
\een
for  $\delta < \frac{\kappa/2}{\kappa +2}$ since the $w_i$'s are i.i.d.\ $\mc{N}(0,\sigma^2)$. Thus \eqref{eq:1b2} holds.

Finally considering the expectation in \eqref{eq:1b2}, Fact \ref{fact:lip_slln} implies
\be
\frac{1}{n} \sum_{i=1}^n \expec_Z \left\{\phi_b(\bar{\sigma}_0 Z, w_i)\right\} \, \stackrel{n \to \infty}{\longrightarrow} \,\expec\left\{ \phi_b(\bar{\sigma}_0 \hat{Z}_0, \sigma Z_w)\right\} \ a.s.,
\label{eq:lip_lln0}
\ee
at rate $n^{-\delta}$.  Combining   \eqref{eq:1b2} and \eqref{eq:lip_lln0} yields the result.

\textbf{(c)} The function $\phi_b(b^0_i, w_i) := b^0_i w_i \in PL(2)$.  By $\mc{B}_0(b)$, $\lim \frac{(b^0)^*w}{n} \overset{a.s.}{=} \mathbb{E}\{\bar{\sigma}_0 \hat{Z}_0 \sigma Z_w)\} = 0$ and the convergence rate is $n^{-\delta}$.

\textbf{(d)} The function $\phi_b(b^0_i, w_i) := (b^0_i)^2 \in PL(2)$.  By $\mc{B}_0(b)$, $\lim \frac{\norm{b^0}^2}{n} \overset{a.s.}{=} \mathbb{E}\{(\bar{\sigma}_0 \hat{Z}_0)^2\} = \bar{\sigma}_0^2$ and the convergence rate is $n^{-\delta}$.

\textbf{(e)} Recall $m^0 = b^0 - w$.  The function $\phi_b(b^0_i, w_i) := (b^0_i - w_i)^2 \in PL(2)$.  By $\mc{B}_0(b)$, $\lim \frac{\norm{m^0}^2}{n} \overset{a.s.}{=} \mathbb{E}\{(\bar{\sigma}_0 \hat{Z}_0 - \sigma Z_w)^2\} = \bar{\sigma}_0^2 + \sigma^2 = \bar{\tau}_0^2$ and the convergence rate is $n^{-\delta}$.

\textbf{(f)} The function $\phi_b(b^0_i, w_i) := b^0_i (b^0_i - w_i) \in PL(2)$.  By $\mc{B}_0(b)$, $\lim \frac{(b^0)^*m^0}{n} \overset{a.s.}{=} \mathbb{E}\{\bar{\sigma}_0 \hat{Z}_0(\bar{\sigma}_0 \hat{Z}_0 - \sigma Z_w)\} = \bar{\sigma}_0^2$ and the convergence rate is $n^{-\delta}$.

\textbf{(g)} For $t=0$, nothing to prove.

\textbf{(h)} Since $M_0$ is the empty matrix, $m^0_\perp = m^0$, so the result is already shown in $\mc{B}_0(e)$. 

\subsubsection{Step 2: Showing $\mathcal{H}_1$ holds}   \label{subsub:step2}
We wish to show that \eqref{eq:Ha}, \eqref{eq:Hb2}, \eqref{eq:Hc}, \eqref{eq:Hd}, \eqref{eq:He}, \eqref{eq:Hf}, \eqref{eq:Hg}, and \eqref{eq:Hh} hold when $t = 0$.

\textbf{(a)}  From the definition of $\Delta_{1,0}$ in Lemma \ref{lem:hb_cond} \eqref{eq:D10}, we have 
\be
\begin{split}
\Delta_{1,0} & = \left[ \left(\frac{\norm{m^0}}{\sqrt{n}} - \bar{\tau}_0\right)\mathsf{I} -\frac{\norm{m^0}}{n} \mathsf{P}_{q^0}\right] Z_0 \\\
& \qquad + q^0 \left(\frac{\norm{q^0}^2}{n}\right)^{-1} \left(\frac{(b^0)^*m_0}{n} - \frac{\norm{q^0}^2}{n}\right) \\ 
& = \left(\frac{\norm{m^0}}{\sqrt{n}} - \bar{\tau}_0\right)Z_0 -\frac{\norm{m^0}}{\sqrt{n}} \frac{q^0}{\sqrt{P}} \frac{Z}{\sqrt{n}} \\
& \quad   + \frac{q^0}{P} \left(\frac{(b^0)^*m^0}{n} - P \right),\label{eq:newDelta10}
\end{split}
\ee
where the second equality follows from Fact \ref{fact:gauss_p0} with $Z \in \mathbb{R} \sim \mc{N}(0,1)$.  It follows from \eqref{eq:newDelta10} that
\be
\begin{split}
&  \max_{j \in sec(\ell)} \abs{[\Delta_{1,0}]_j} \leq \left \lvert \frac{\norm{m^0}}{\sqrt{n}} - \bar{\tau}_0\right \lvert \max_{j \in sec(\ell)} \abs{Z_{0_j}}  \\
& \qquad + \frac{\norm{m^0}}{\sqrt{n}} \sqrt{\frac{nP_{\ell}}{P}} \frac{\abs{Z}}{\sqrt{n}} + \frac{\sqrt{n P_{\ell}}}{P} \left \lvert \frac{(b^0)^*m^0}{n} - P \right \lvert.\label{eq:newDelta10eq1}
\end{split}
\ee
We show all terms on the RHS of the above are $o( n^{-\delta} \sqrt{\log M})$ almost surely. Recall $\sqrt{nP_{\ell}} = \Theta(\sqrt{\log M})$.  By $\mc{B}_0$(e), $\norm{m^0}^2/n \overset{a.s.}{\to} \bar{\tau}_0^2$ at rate $n^{-\delta}$. This along with the Fact \ref{lem:maxZbound} implies that the first term is $o(n^{-\delta} \sqrt{\log M})$ almost surely.  Similarly from $\mc{B}_0$(e) and the fact that $\abs{Z}/\sqrt{n}$ is almost surely $o(n^{-\delta})$ the second term is $o( n^{-\delta} \sqrt{\log M} )$; finally by $\mc{B}_0$(f) the third term is also $o(n^{-\delta} \sqrt{\log M})$ almost surely.  We have therefore shown that $\max_{j \in sec(\ell)} \abs{[\Delta_{1, 0}]_{j}} \overset{a.s}{=} o\left( n^{-\delta} \sqrt{\log M}\right).$

Next, from Lemma \ref{lem:hb_cond} \eqref{eq:Ha_dist} it follows,
\ben
\begin{split}
\max_{j \in sec(\ell)} \abs{h^1_j} &\leq \abs{\bar{\tau}_0} \max_{j \in sec(\ell)} \abs{Z_{0_j}} + \max_{j \in sec(\ell)} \abs{\Delta_{{1,0}_j}} \\
&\overset{a.s}{\leq} \abs{\bar{\tau}_0} \left(\sqrt{3 \log M}\right) + o( n^{-\delta} \sqrt{\log M}),
\end{split}
\een
where we have used Fact \ref{lem:maxZbound} for the second inequality.  This completes the proof.
 
\textbf{(b)}   The proof of this part involves several claims which are fairly straightforward but tedious to verify, so we give only the main steps, referring the reader to \cite{steps24b} for details.  Throughout we use generic $\phi_{k, \ell}(x, y, z)$ since the steps are identical for all $k \in \{1, 2, 3, 4\}$.  From Lemma \ref{lem:hb_cond} \eqref{eq:Ha_dist},
\begin{equation*}
\begin{split}
&\phi_{k, \ell}(a_0 h_{\ell}^1, \, b_0 h_{\ell}^1, \, \beta_{0_{\ell}}) \lvert_{\mscrs_{1,0}} \\
&\overset{d}{=} \phi_{k, \ell}\left(a_0 \bar{\tau}_0 Z_{0_{\ell}} + a_0 [\Delta_{1,0}]_{\ell}, \, b_0 \bar{\tau}_0 Z_{0_{\ell}} + b_0 [\Delta_{1,0}]_{\ell}, \, \beta_{0_{\ell}}\right).
\end{split}
\end{equation*}
By $\mc{H}_{1}$(a), $\max_{j \in sec(\ell)} \abs{[\Delta_{1, 0}]_{j}} \overset{a.s.}{=} o(n^{-\delta'} \sqrt{\log M} )$ for each $\ell \in [L]$ and some $\delta'>0$.  In \cite{steps24b}, the first step of the proof uses this to show for each of the functions in \eqref{eq:phih_fns},
\begin{equation*}
\begin{split}
&\frac{1}{L}  \sum_{\ell=1}^L  \left \lvert \phi_{k, \ell}\left(a_0 \bar{\tau}_0 Z_{0_{\ell}} + a_0 [\Delta_{1,0}]_{\ell}, \, b_0 \bar{\tau}_0 Z_{0_{\ell}} + b_0 [\Delta_{1,0}]_{\ell}, \, \beta_{0_{\ell}}\right) \right. \\
& \qquad  \left. - \phi_{k, \ell}\left(a_0 \bar{\tau}_0 Z_{0_{\ell}}, \, b_0 \bar{\tau}_0 Z_{0_{\ell}}, \, \beta_{0_{\ell}}\right)\right \lvert  \overset{a.s.}{=} o({n^{- \delta'} \log M}).
\end{split}
\end{equation*}
Choosing $\delta \in (0, \delta')$ ensures that we can drop the deviation term $\Delta_{1,0}$.

The second step of the proof appeals to Fact \ref{fact:slln} to show that 
\begin{equation*}
\begin{split}
& \lim n^{\delta} \left[\frac{1}{L}\sum_{\ell=1}^L \phi_{k, \ell}\left(a_0 \bar{\tau}_0 Z_{0_{\ell}}, \, b_0 \bar{\tau}_0 Z_{0_{\ell}}, \, \beta_{0_{\ell}}\right) \right. \\
& \left. \qquad - \frac{1}{L}\sum_{\ell=1}^L \mathbb{E}_{Z_0}\left\{\phi_{k, \ell}\left(a_0 \bar{\tau}_0 Z_{0_{\ell}}, \, b_0 \bar{\tau}_0 Z_{0_{\ell}}, \, \beta_{0_{\ell}} \right)\right\}\right] \overset{a.s.}{=}  0.
\end{split}
\end{equation*}
Let $\tilde{Z}_0$ be an independent copy of $Z_0$.  In order to use Fact \ref{fact:slln} to get the above result we must prove the following for each function in \eqref{eq:phih_fns}, for some $0 \leq \kappa \leq 1$, and $c>0$ some constant.
\begin{equation}
\begin{split}
\label{eq:2b4a}
& \frac{1}{L}\sum_{\ell=1}^L\mathbb{E}_{\tilde{Z}_0, Z_0} \left\lvert n^{\delta}  \phi_{k, \ell}\left(a_0 \bar{\tau}_0 \tilde{Z}_{0_{\ell}}, \, b_0 \bar{\tau}_0 \tilde{Z}_{0_{\ell}}, \, \beta_{0_{\ell}}\right) \right. \\
& \qquad  \left. -  n^{\delta} \phi_{k, \ell}\left(a_0 \bar{\tau}_0 Z_{0_{\ell}}, \, b_0 \bar{\tau}_0 Z_{0_{\ell}}, \, \beta_{0_{\ell}}\right) \right\lvert^{2 + \kappa}\leq cL^{\kappa/2}.
\end{split}
\end{equation}
Note that the exact condition required by Fact \ref{fact:slln} follows from \eqref{eq:2b4a} by an application of Jensen's inequality.  In \cite{steps24b}, it is shown that for each function in \eqref{eq:phih_fns} and each $\ell \in [L]$,
\be
\begin{split}
&\mathbb{E}_{\tilde{Z}_0, Z_0} \left\lvert  {\phi}_{k, \ell}\left(a_0 \bar{\tau}_0 \tilde{Z}_{0_{\ell}}, \, b_0 \bar{\tau}_0 \tilde{Z}_{0_{\ell}}, \, \beta_{0_{\ell}}\right) \right. \\
& \quad \left. - {\phi}_{k, \ell}\left(a_0 \bar{\tau}_0 Z_{0_{\ell}}, \, b_0 \bar{\tau}_0 Z_{0_{\ell}}, \, \beta_{0_{\ell}}\right)\right\lvert^{2 + \kappa} 
\stackrel{a.s.}{=} O((\log M)^{2+\kappa}). 
\label{eq:hatA_tildeA_bnd}
\end{split}
\ee
Bound \eqref{eq:hatA_tildeA_bnd} implies \eqref{eq:2b4a} holds if $\delta(2+\kappa)$ is chosen to be smaller than $\frac{1}{2}\kappa$. (Recall $L= \Theta(n/ \log n)$).

The final step of the proof is to show that 
\ben
\begin{split}
&\lim n^{\delta}\left[ \frac{1}{L} \sum_{\ell=1}^L \mathbb{E}_{Z_0}\left[ \phi_{k, \ell}\left(a_0 \bar{\tau}_0 Z_{0_{\ell}}, b_0 \bar{\tau}_0 Z_{0_{\ell}}, \beta_{0_\ell}\right)\right] \right. \\
&\quad \left. - \frac{1}{L} \sum_{\ell=1}^L \mathbb{E}_{(\breve{Z}_0, \beta)}\left[ \phi_{k, \ell}\left(a_0 \bar{\tau}_0 \breve{Z}_{0_{\ell}}, b_0 \bar{\tau}_0 \breve{Z}_{0_{\ell}}, \beta_{\ell}\right)\right] \right]  \overset{a.s.}{=} 0
\end{split}
\een
But the above holds because the uniform distribution of the non-zero entry in $\beta_\ell$ over the $M$ possible locations and the i.i.d.\ distribution of $Z_{0}$ (and of $\breve{Z}_0$) together ensure that $\forall \beta_{0} \in \mcb$, we have
\be
\begin{split} 
&\mathbb{E}_{Z_0}\left[ \phi_{k, \ell}\left(a_0 \bar{\tau}_0 Z_{0_{\ell}}, b_0 \bar{\tau}_0 Z_{0_{\ell}}, \beta_{0_\ell}\right)\right]  \\
&= \mathbb{E}_{(\breve{Z}_0, \beta)}\left[ \phi_{k, \ell}\left(a_0 \bar{\tau}_0 \breve{Z}_{0_{\ell}}, b_0 \bar{\tau}_0 \breve{Z}_{0_{\ell}}, \beta_{\ell}\right)\right],
\  \forall \,\, \ell \in [L].
\end{split}
\ee

The existence of the limit of  $\frac{1}{L} \sum_{\ell=1}^L \mathbb{E}_{(\breve{Z}_0, \beta)} [ \phi_{k, \ell} (a_0 \bar{\tau}_0 \breve{Z}_{0_{\ell}}, b_0 \bar{\tau}_0 \breve{Z}_{0_{\ell}}, \beta_{\ell} )]$  for $k=1$ follows from the law of large numbers; for $k=2,3,4$, the limit follows from  Appendix \ref{app:qrqs_lim}.

\textbf{(c)} Using the fourth function in \eqref{eq:phih_fns} with $r=-1$, $\lim \frac{(h^1)^*q^0}{n} \overset{a.s.}{=} \lim -\frac{1}{n} \mathbb{E}\{\bar{\tau}_0 \breve{Z}_{0}^* \beta\} = 0$ by $\mc{H}_1(b)$ and the convergence rate is $n^{-\delta}$.

\textbf{(d)} Using the first function in \eqref{eq:phih_fns}, $\lim \frac{\norm{h^1}^2}{N} \overset{a.s.}{=} \lim \frac{ \bar{\tau}_0^2}{N} \, \mathbb{E}\norm{\breve{Z}_{0}}^2 = \bar{\tau}_0^2$ by $\mc{H}_1(b)$ and the convergence rate is $n^{-\delta}$.

\textbf{(e)} Using the third function in \eqref{eq:phih_fns}, by $\mc{H}_1$(b) we have for  $r=0$ or $r=1$:
\ben
\begin{split}
&\lim \frac{(q^r)^* q^1}{n} \\
&\overset{a.s.}{=} \lim \frac{1}{n} \mathbb{E}[(\eta^{r-1}(\beta - \bar{\tau}_{r-1} \breve{Z}_{r-1}) - \beta)^*(\eta^{0}(\beta - \bar{\tau}_{0} \breve{Z}_{0}) - \beta)] \\
&= \bar{\sigma}_1^2,
\end{split}
\een 
and the convergence rate is $n^{-\delta}$. The last equality above is shown in Appendix \ref{app:qrqs_lim}.

\textbf{(f)} Using the fourth function in \eqref{eq:phih_fns} with $r=0$, by $\mc{H}_1(b)$ we have
\be
\begin{split}
\lim n^\delta & \left( \frac{1}{n}(h^1)^*q^1 - \right. \\
& \left. \frac{1}{n} \sum_{\ell = 1}^L \mathbb{E}\{\bar{\tau}_0 \breve{Z}_{0_{\ell}}^*[\eta^0_{\ell}(\beta - \bar{\tau}_0 \breve{Z}_{0}) - \beta_{\ell}]\} \right) =0 \quad a.s.
\end{split}
\label{eq:2d1}
\ee
 Consider a single term in the expectation in \eqref{eq:2d1}, say $\ell=1$.  We have
\begin{equation}
\begin{split}
&\mathbb{E}\{\bar{\tau}_0 \breve{Z}_{0_{(1)}}^*[\eta^0_{(1)}(\beta - \bar{\tau}_0 \breve{Z}_{0}) - \beta_{(1)}]\} \\
& \qquad = \bar{\tau}_0 \sum_{i=1}^M \mathbb{E}\{\breve{Z}_{0_i} [\eta^0_{i}(\beta - \bar{\tau}_0 \breve{Z}_{0}) - \beta_{i}]\}
\label{eq:2d2} 
\end{split}
\end{equation}
where $\beta_{(1)} = (\beta_{1}, \beta_{2}, \ldots, \beta_{M})$ and $\breve{Z}_{0_{(1)}} = (\breve{Z}_{0_1}, \breve{Z}_{0_2}, \ldots , \breve{Z}_{0_M}).$  Note that for each $i$, the function $\eta^0_i(\cdot)$ depends on all the $M$ indices in the section containing $i$.  For each $i \in [M]$, we evaluate the expectation on the RHS of \eqref{eq:2d2} using the law of iterated expectations:
\begin{equation}
\begin{split}
&\mathbb{E}\{\breve{Z}_{0_i} [\eta^0_i(\beta - \bar{\tau}_0 \breve{Z}_{0}) - \beta_{i}]\} \\
& \quad = \mathbb{E}\left[\mathbb{E}\left\{\breve{Z}_{0_i} [\eta_{0_i}(\beta - \bar{\tau}_0 \breve{Z}_{0}) - \beta_{i}] \mid \beta_{(1)}, \breve{Z}_{0_{(1) \setminus i}}\right\}\right]
\end{split}
\label{eq:2d3}
\end{equation}
where the inner expectation is over $\breve{Z}_{0_i}$ conditioned on $\{\beta_{(1)}, \breve{Z}_{0_{(1) \setminus i}}\}$.  Since $\breve{Z}_{0_i}$ is independent of $\{\beta_{(1)}, \breve{Z}_{0_{(1) \setminus i}}\}$, the latter just act as constants in the inner expectation  over $\breve{Z}_{0_i} \sim \mathcal{N}(0,1)$.
Applying Stein's lemma (Fact \ref{fact:stein}) to the inner expectation, we obtain
\begin{equation*}
\begin{split}
& \expec\left[ \mathbb{E}  \left\{ \breve{Z}_{0_i} [\eta^0_i(\beta - \bar{\tau}_0 \breve{Z}_{0}) - \beta_{i}] \mid  \beta_{(1)}, \breve{Z}_{0_{(1)} \setminus i}\right \} \right] \\
&= \mathbb{E}\left[ \mathbb{E} \left\{ \frac{\partial}{\partial \breve{Z}_{0_i}}[\eta^0_i(\beta -\bar{\tau}_0 \breve{Z}_{0}) - \beta_{i}] \mid \beta_{(1)}, \breve{Z}_{0_{(1)} \setminus i}\right\}\right] \\
& \overset{(a)}{=} -\frac{\bar{\tau}_0}{\bar{\tau}^2_0}\mathbb{E}\left[\mathbb{E}\left\{\eta^0_i(\beta - \bar{\tau}_0 \breve{Z}_{0}) \right. \right. \\
&\qquad  \qquad  \qquad \left. \left. \cdot \Big(\sqrt{nP_{1}} - \eta^0_i(\beta - \bar{\tau}_0 \breve{Z}_{0})\Big) \Big \lvert \  \beta_{(1)}, \breve{Z}_{0_{(1)} \setminus i}\right\} \right] \\
&\overset{(b)}{=} -\frac{1}{\bar{\tau}_0}\mathbb{E}\left[\eta^0_i(\beta - \bar{\tau}_0 \breve{Z}_{0})\left(\sqrt{nP_{1}} - \eta^0_i(\beta - \bar{\tau}_0 \breve{Z}_{0})\right)\right]
\end{split}
\end{equation*}
where $(a)$ follows from the definition of $\eta^t_i$ in \eqref{eq:eta_def} which implies
$\frac{\partial \eta^t_i(s)}{\delta s_i} = \frac{\eta^t_i(s)}{\bar{\tau}_t^2}\left(\sqrt{nP_{\ell}} - \eta^t_i(s)\right) \text{ for } i \in sec(\ell)$, and $(b)$ from the law of iterated expectation.  Using the above in \eqref{eq:2d3} and \eqref{eq:2d2}, we have
\begin{equation}
\begin{split}
&\mathbb{E}\left[ \bar{\tau}_0 \breve{Z}_{0_{(1)}}^*[\eta^0_{(1)}(\beta - \bar{\tau}_0 \breve{Z}_{0}) - \beta_{(1)}]\right] \\
& =  \sum_{i=1}^M  \mathbb{E}\left[ \eta^0_i(\beta - \bar{\tau}_0 \breve{Z}_{0})\left(\eta^0_i(\beta - \bar{\tau}_0 \breve{Z}_{0}) - \sqrt{nP_{1}}\right)\right].
\end{split}
\label{eq:2d5} 
\end{equation}
The  argument  above can be repeated for each section $\ell \in [L]$ to obtain a relation analogous to \eqref{eq:2d5}.  Using this for the expectation in \eqref{eq:2d1}, we  obtain
\begin{equation*}
\begin{split}
\lim \frac{1}{n}(h^1)^*q^1 &\overset{a.s.}{=} \lim \left( \frac{1}{n}\mathbb{E}\left[\norm{\eta^{0}(\beta - \bar{\tau}_0 \breve{Z}_{0})}^2\right] - P\right) \\
&= -\bar{\sigma}_1^2,
\end{split}
\end{equation*}
with convergence rate $n^{-\delta}$. The last equality above follows from Appendix \ref{app:qrqs_lim}.

Finally, recall from $\mc{B}_0 (e)$   that $\norm{m^0}^2/n \stackrel{a.s.}{\to} \bar{\tau}^2_0$ at rate $n^{-\delta}$.  Further, from \eqref{eq:lambda_t_def}, we observe that
\ben
\begin{split}
 \lambda_1 &= \frac{1}{\bar{\tau}^2_0}\left(  \frac{\norm{\beta^1}^2}{n} - P \right)  \\
 &\stackrel{a.s.}{\to} 
 \lim \frac{1}{\bar{\tau}_0^2}\left( \frac{\mathbb{E}\left[ \norm{\eta^{0}(\beta - \bar{\tau}_0 \breve{Z}_{0})}^2\right]}{n} - P\right)
 = \frac{-\bar{\sigma}^2_1}{\bar{\tau}^2_0},
\end{split}
\een
where the convergence at rate $n^{-\delta}$ follows from $\mc{H}_{1}$(b) applied to the second function in \eqref{eq:phih_fns}.

\textbf{(g)} Note that $Q_{1}^* Q_{1}$ is invertible since $Q_{1}^* Q_{1} = \norm{q^0} = nP > 0$
\ben
\gamma^1_0 = \left( \frac{Q_1^* Q_1}{n} \right)^{-1} \frac{Q_1^* q^0}{n} = \frac{(q^0)^* q^1}{nP} \overset{a.s.}{\to} \frac{\bar{\sigma}_1^2}{P} = 
\frac{\bar{\sigma}_1^2}{\bar{\sigma}_0^2},
\een
where the limit follows from $\mc{H}_1$(e).

\textbf{(h)} Let $\mathsf{P}_{Q_{1}} = Q_{1}(Q_{1}^* Q_{1})^{-1} Q_{1}^*$ be the projection matrix onto the column space of $Q_{1} = q^0$. Note that $Q_{1}^* Q_{1}$ is invertible since $Q_{1}^* Q_{1} = nP > 0$.  Then,
\begin{equation}
\begin{split}
\frac{\norm{q_{\perp}^{1}}^2}{n} &= \norm{ q^1 - \mathsf{P}_{Q_{1}} q^1}^2 \\
&=  \frac{\norm{q^{1}}^2}{n} - \frac{(q^{1})^*q^0}{n} \cdot \left(\frac{(q^0)^*q^0}{n}\right)^{-1} \cdot\frac{(q^0)^*q^{1}}{n}.
\end{split}
\label{eq:stepheq1}
\end{equation}
Using the representation in \eqref{eq:stepheq1}, it follows by $\mathcal{H}_{1}$(e) that
\begin{equation*}
 \norm{q_{\perp}^{1}}^2/n \overset{a.s.}{\to}  \bar{\sigma}^2_{1} - (\bar{\sigma}^4_{1}/\bar{\sigma}_0^2) = (\bar{\sigma}_1^{\perp})^2.
\label{eq:stepheq2}
\end{equation*}
Finally note that $\bar{\sigma}^2_r = \sigma^2\left( \left( 1 + \snr \right)^{1-\xi_{r-1}} -1\right)$  with $\xi_{r-1}$ defined in \eqref{eq:lim_alph}. The definition of   $\xi_{r-1}$ implies that $(\bar{\sigma}_r^{\perp})^2$ is strictly positive for $r \leq T^*$, where $T^*=  \left\lceil \frac{2 \mc{C}}{\log(\mc{C}/R)} \right \rceil$.

\subsubsection{Step 3: Showing $\mathcal{B}_t$ holds} \label{subsub:step3} 
We wish to show that \eqref{eq:Ba}, \eqref{eq:Bb}, \eqref{eq:Bc}, \eqref{eq:Bd}, \eqref{eq:Be}, \eqref{eq:Bf}, \eqref{eq:Bg}, and \eqref{eq:Bh} hold assuming $\mathcal{B}_r, \mathcal{H}_s$ holds for all $r < t $ and $s \leq t $.

\textbf{(a)}  Let $\textbf{M}_t := \frac{M_{t}^* M_{t}}{n}$ and $v := \frac{H_t^* q^t_{\perp}}{n} - \frac{M_t}{n}^*\left[\lambda_t m^{t-1} - \sum_{r=1}^{t-1} \lambda_{r} \gamma^t_{r} m^{r-1}\right]$.  From the definition of $\Delta_{t,t}$ in Lemma \ref{lem:hb_cond} \eqref{eq:Dtt}, we have 
\be
\begin{split}
\Delta_{t,t} = & \, \sum_{r=0}^{t-2} \gamma^t_r b^r + \left(\gamma^t_{t-1} -  \frac{\bar{\sigma}_t^2}{\bar{\sigma}^2_{t-1}}\right) b^{t-1} + \left(\frac{\norm{q^t_{\perp}}}{\sqrt{n}} - \bar{\sigma}_{t}^{\perp}\right)Z'_t \\
&\quad  - \frac{\norm{q^t_{\perp}}}{\sqrt{n}}\sum_{s=0}^{t-1} \tilde{m}^{s} \frac{\bar{Z}'_{t_s}}{\sqrt{n}}+ M_t \textbf{M}_t^{-1}v, \label{eq:newDeltatt1}
\end{split}
\ee
where we have used Fact \ref{fact:gauss_p0} to write
\ben
\frac{\norm{q^t_{\perp}} }{\sqrt{n}} \mathsf{P}_{M_t}Z'_t \overset{d}{=} \frac{\norm{q^t_{\perp}} \tilde{M}_t \bar{Z}'_t}{n} =  \frac{\norm{q^t_{\perp}}}{\sqrt{n}}\sum_{s=0}^{t-1} \tilde{m}^{s} \frac{\bar{Z}'_{t_s}}{\sqrt{n}}.
\een
The matrix $\tilde{M}_t= [\tilde{m}^0| \ldots| \tilde{m}^t ] \in \mathbb{R}^{n \times t}$ forms an orthogonal basis for the column space of $M_t$ such that  $\norm{\tilde{m}^s}= \sqrt{n}$, $\forall s$, and $\bar{Z}'_t \in \mathbb{R}^t$ is an independent i.i.d.\ $\mc{N}(0,1)$ random vector.  Using $M_t \textbf{M}_t^{-1}v = \sum_{r=0}^{t-1} m^r [\textbf{M}_t^{-1}v]_{r+1}$ and $\norm{\tilde{m}^s}^2 = n$ in \eqref{eq:newDeltatt1}, we obtain the bound
\be
\begin{split}
&\frac{\norm{\Delta_{t,t}}^2}{n} \leq (3t+1)\left[\sum_{r=0}^{t-2} (\gamma^t_r)^2 \frac{\norm{b^r}^2}{n}  \right.   \\
&\quad \left. + \left(\gamma^t_{t-1} -  \frac{\bar{\sigma}_t^2}{\bar{\sigma}_{t-1}^2}\right)^2 \frac{\norm{b^{t-1}}^2}{n}  + \left(\frac{\norm{q^t_{\perp}}}{\sqrt{n}} - \bar{\sigma}_{t}^{\perp}\right)^2 \frac{\norm{Z'_t}^2}{n}  \right. \\
&\quad \left.  +  \frac{\norm{q^t_{\perp}}^2}{n} \frac{\norm{\bar{Z}'_{t}}^2}{n}  +  \sum_{r=0}^{t-1} \frac{\norm{m^r}^2}{n} \left([\textbf{M}_t^{-1}v]_{r+1}\right)^2 \right] .\label{eq:newDeltatteq1}
\end{split}
\ee
 We show that each term on the RHS of \eqref{eq:newDeltatteq1} is almost surely $o(n^{-\delta})$.   Note that by $\mc{H}_t$(g), $\gamma^t_j \overset{a.s.}{\to} 0$ for $0 \leq j \leq t-2$ and $\gamma^t_{t-1} \overset{a.s.}{\to} \frac{\bar{\sigma}_t^2}{\bar{\sigma}_{t-1}}$.  By $\mc{B}_0$(d) -- $\mc{B}_{t-1}$(d), $\norm{b^r}^2/n \overset{a.s.}{\to} \bar{\sigma}_r^2$ for $0 \leq r \leq t-1$. These imply that the first and second terms in \eqref{eq:newDeltatteq1} are $o(n^{-\delta})$ almost surely.  By $\mc{H}_t$(h), $\norm{q^t_{\perp}}^2/n \overset{a.s.}{\to} (\bar{\sigma}_t^{\perp})^2$; noting that $\norm{Z'_{t}}^2$ and $\norm{\bar{Z}'_{t}}^2$ are $\chi^2_t$ random variables, it follows that the third and fourth terms are $o(n^{-\delta})$ almost surely.    Finally, by $\mc{B}_0$(e) -- $\mc{B}_{t-1}$(e), $\norm{m^r}^2/n \stackrel{a.s.}{\to} \bar{\tau}_r^2$ for $0 \leq r \leq t-1$.  Therefore to prove  convergence for the fifth term, we will show that $ \left([\textbf{M}_t^{-1}v]_{r+1}\right)^2 \overset{a.s.}{\to} 0$  at rate $n^{-\delta}$.   Note that
\be
\begin{split}
&[\textbf{M}_t^{-1}v]_{r+1} = \\
&\  \begin{cases}
\lambda_{r+1} \gamma^t_{r+1} +  \left[\, \left(\tfrac{M_t^*M_t}{n} \right)^{-1} \tfrac{H_t^* q_{\perp}^t}{n} \, \right]_{r+1}  \text{ for } 0 \leq  r \leq t-2, \\
- \lambda_t  + \left[\, \left(\tfrac{M_t^*M_t}{n} \right)^{-1} \tfrac{H_t^* q_{\perp}^t}{n} \, \right]_{t},  \text{ for }  r = t-1. \label{eq:coeff_compare}
\end{cases}
\end{split}
\ee
We show that each of the above coefficients is  $o(n^{-\delta})$. Indeed, for $1 \leq i \leq t$,
\ben
\begin{split}
&\left[ \frac{H_{t}^* q_{\perp}^t}{n}\right]_{i} = \frac{(h^{i})^* q_{\perp}^t}{n} 
= \frac{(h^{i})^* q^t}{n} - \sum_{r=0}^{t-1} \gamma^t_r \frac{(h^{i})^* q^r}{n} \\
& \stackrel{a.s.}{\to}   \lim\left[\lambda_t \frac{(m^{i-1})^* m^{t-1}}{n} -\sum_{r=0}^{t-1} \gamma^t_r \lambda_{r} \frac{(m^{i-1})^* m^{r-1}}{n}\right], 
\end{split}
\een
where the convergence (at rate $n^{-\delta}$) follows from $\mc{H}_t$(f) and $\mc{H}_t$(g) (convergence of $\vec{\gamma}^t$ to finite values). Therefore, 
\be 
\left[ \frac{H_{t}^* q_{\perp}^t}{n}\right] \stackrel{a.s.}{\to}  \lim\left[\lambda_t \frac{(M_t)^* m^{t-1}}{n} - \sum_{r=0}^{t-2} \gamma^t_{r+1} \lambda_{r+1} \frac{(M_t)^* m^{r}}{n}\right]
\label{eq:Htqt_conv} 
\ee
at rate $n^{-\delta}$. Using \eqref{eq:Htqt_conv} in \eqref{eq:coeff_compare} we see that each coefficient of \eqref{eq:coeff_compare} is $o(n^{-\delta})$, which completes the proof. 

\textbf{(b)} Using the characterization for $b^t$ obtained in Lemma \ref{lem:hb_cond} \eqref{eq:Ba_dist}, we have
\begin{equation*}
\begin{split}
&\phi_b(b_i^0, \ldots, b_i^t, w_i) \Big{\lvert}_{\mscrs_{t, t}} \, \\
& \overset{d}{=}  \, \phi_b\left(b_i^0, \ldots, b_i^{t-1}, \frac{\bar{\sigma}_t^2}{\bar{\sigma}_{t-1}^2} b^{t-1}_i +  \bar{\sigma}_{t}^{\perp} \, Z'_{t_i} + \left[\Delta_{t,t}\right]_i, w_i\right).
\end{split}
\end{equation*}
The deviation term $\Delta_{t,t}$ in the RHS of the above can be dropped. Indeed, defining 
\ben
\begin{split}
a_i &= \left(b_i^0, \ldots, b_i^{t-1}, \frac{\bar{\sigma}_t^2}{\bar{\sigma}_{t-1}^2} b^{t-1}_i +  \bar{\sigma}_{t}^{\perp} \, Z'_{t_i} + \left[\Delta_{t,t}\right]_i, w_i \right), \\
c_i &= \left(b_i^0, \ldots, b_i^{t-1}, \frac{\bar{\sigma}_t^2}{\bar{\sigma}_{t-1}^2} b^{t-1}_i +  \bar{\sigma}_{t}^{\perp} \, Z'_{t_i}, w_i\right),
\end{split}
\een
we can show that almost surely
\be 
\begin{split}
 & \frac{1}{n} \left\lvert \sum_{i=1}^n \phi_b\left(a_i\right) - \sum_{i=1}^n \phi_b\left(c_i\right) \right\lvert 
 \leq  \frac{1}{n} \sum_{i=1}^n \left\lvert \phi_b\left(a_i\right) - \phi_b\left(c_i\right) \right\lvert \\
&\ \stackrel{(a)}{\leq}   \frac{C}{n} \sum_{i=1}^n (1 + \norm{a_i} + \norm{c_i}) \left \lvert \left[\Delta_{t,t}\right]_i \right \rvert  \\
 & \ \overset{(b)}{\leq} C  \sqrt{\sum_{i=1}^n \frac{(1 + \norm{a_i} + \norm{c_i})^2}{n}} \cdot \sqrt{\frac{\norm{\Delta_{t,t}}^2}{n}}
\stackrel{(c)}{=} o(n^{-\delta'}).
\end{split}
\label{eq:3b1}
\ee
In \eqref{eq:3b1}, $(a)$ holds because $\phi_b \in PL(2)$; $(b)$ is obtained using H{\"o}lder's inequality, and $(c)$ follows from $\mc{B}_t$(a) if  $\sum_{i=1}^n \frac{{\norm{a_i}}^2}{n}$ and $\sum_{i=1}^n \frac{{\norm{c_i}}^2}{n}$ are bounded and finite. This holds almost surely since
\begin{align*}
&\sum_{i=1}^n\frac{{\norm{a_i}}^2}{n} \leq C\left[\sum_{i=1}^n \frac{{\norm{c_i}}^2}{n} + \frac{\norm{\Delta_{t,t}}^2}{n}\right] \\
&\leq C'\left[\sum_{r=0}^{t-1} \frac{\norm{b^r}^2}{n} + \left(\frac{\bar{\sigma}_t^4}{\bar{\sigma}_{t-1}^4}\right)\frac{\norm{b^{t-1}}^2}{n} + ( \bar{\sigma}_{t}^{\perp})^2 \frac{\norm{Z'_{t}}^2}{n} \right.\\
& \qquad \qquad \left. + \frac{\norm{w}^2}{n} + \frac{\norm{\Delta_{t,t}}^2}{n}\right].
\end{align*}
The RHS above is finite almost surely by $\mc{B}_0$(d) -- $\mc{B}_{t-1}$(d), $\mc{B}_t$ (a), and the Gaussianity of $w$ and $Z'_{t}$.  Thus by choosing $\delta < \delta'$, we can work with $c_i$ instead of $a_i$. 
Next, we use Fact \ref{fact:slln} to show that
\begin{equation}
\lim n^{\delta}\left[ \frac{1}{n}\sum_{i=1}^n \phi_b(c_i) -  \frac{1}{n}\sum_{i=1}^n \mathbb{E}_{Z'_t}\left\{\phi_b(c_i)\right\}\right] \overset{a.s.}{=} 0,
\label{eq:3b5}
\end{equation}
To appeal to Fact \ref{fact:slln}, we need to verify that
\begin{equation}
\frac{1}{n}\sum_{i=1}^n \mathbb{E} 
\left\lvert n^{\delta}\phi_b\left(c_i\right) - \mathbb{E}_{Z'_t}\left\{n^{\delta} \phi_b\left(c_i\right)\right\}\right\lvert^{2 + \kappa} \leq cn^{\kappa/2}.
\label{eq:3b4}
\end{equation}
Let $\tilde{Z}'_t$ be an independent copy of $Z'_t$.  In what follows we drop $i$ indices on $Z'_t$ and $\tilde{Z}'_t$ 
for brevity and define $\tilde{\kappa} = \kappa + 2$.  Using steps similar to \eqref{eq:upperbound0}, we can show that 
\begin{align}
& \expec \left\lvert \phi_b\left(c_i\right) - \mathbb{E}_{Z'_t}\left\{ \phi_b\left(c_i\right)\right\}\right\lvert^{\tilde{\kappa}} \nonumber \\
& \leq
\kappa' \abs{ \bar{\sigma}_{t}^{\perp}}^{\tilde{\kappa}} \mathbb{E}_{\tilde{Z}'_t, Z'_t}\left\{ \abs{Z'_t- \tilde{Z}'_t}^{\tilde{\kappa}}\left(1 +  \abs{ \bar{\sigma}_{t}^{\perp} Z'_t}^{\tilde{\kappa}} + \abs{ \bar{\sigma}_{t}^{\perp} \tilde{Z}'_t}^{\tilde{\kappa}}\right)\right\} \nonumber \\
&\quad + \kappa'  \left(\sum_{r=0}^{t-2}\abs{b_i^r}^{\tilde{\kappa}} + \left(\left \lvert 1 +\frac{\bar{\sigma}_t^2}{\bar{\sigma}^2_{t-1}} \right \lvert \abs{b_i^{t-1}}\right)^{\tilde{\kappa}}  +  \abs{w_i}^{\tilde{\kappa}}  \right) \nonumber \\
& \qquad  \cdot  \abs{ \bar{\sigma}_{t}^{\perp}}^{\tilde{\kappa}} \mathbb{E}_{\tilde{Z}'_t, Z'_t}\left\{  \abs{ \bar{\sigma}_{t}^{\perp}}^{\tilde{\kappa}} \abs{Z'_t- \tilde{Z}'_t}^{\tilde{\kappa}}\right\} \nonumber \\
& \stackrel{(a)}{\leq} \kappa_1 + \kappa_2 \left(\sum_{r=0}^{t-2}\abs{b_i^r}^{\tilde{\kappa}} + \left(\left \lvert 1 +\frac{\bar{\sigma}_t^2}{\bar{\sigma}^2_{t-1}} \right \lvert \abs{b_i^{t-1}}\right)^{\tilde{\kappa}}  +  \abs{w_i}^{\tilde{\kappa}}  \right),
\label{eq:2pluskap_bound}
\end{align}
for some constants $\kappa', \kappa_1, \kappa_2 >0$. In \eqref{eq:2pluskap_bound}, $(a)$ holds since $\tilde{Z}'_t, Z'_t$ are $\mc{N}(0,1)$. Substituting \eqref{eq:2pluskap_bound} in the LHS of \eqref{eq:3b4}, and applying induction hypotheses $\mc{B}_0$(d) -- $\mc{B}_{t-1}$(d) shows that the condition \eqref{eq:3b4} is satisfied if $\delta < \frac{\kappa/2}{\tilde{\kappa}} = \frac{\kappa/2}{\kappa +2}$. Thus \eqref{eq:3b5} holds, and we now need to show that the limit of
\begin{equation*}
\begin{split}
&\frac{n^{\delta} }{n}\sum_{i=1}^n \left[\mathbb{E}_{Z'_t}\left\{\phi_b\left(b_i^0, \ldots, b_i^{t-1}, \frac{\bar{\sigma}_t^2}{\bar{\sigma}^2_{t-1}} b^{t-1}_i +  \bar{\sigma}_{t}^{\perp} \, Z'_{t_i}, w_i\right)\right\}  \right. \\
&\qquad \left. - \mathbb{E}\{\phi_b(\bar{\sigma}_0 \hat{Z}_0,\ldots, \bar{\sigma}_t \hat{Z}_t, \sigma Z_w)\}\right]
\end{split}
\end{equation*}
is almost surely $0$ with $\mathbb{E}[\bar{\sigma}_r \hat{Z}_r \bar{\sigma}_t \hat{Z}_t] = \bar{\sigma}_t^2$ for all $0 \leq r \leq t$.  Define the function
\ben
\begin{split}
&\phi_b^{NEW}(b^0_i, \ldots, b^{t-1}_i, w_i) \\
&\quad := \mathbb{E}_{Z'_t}\left\{\phi_b\left(b_i^0, \ldots, b_i^{t-1}, \frac{\bar{\sigma}_t^2}{\bar{\sigma}_{t-1}^2} b^{t-1}_i +  \bar{\sigma}_{t}^{\perp} \, Z'_{t_i}, w_i\right)\right\}.
\end{split}
\een
It can be verified that  $\phi_b^{NEW} \in PL(2)$, and hence the induction hypothesis $\mc{B}_{t-1}$(b) implies that the limit of
\begin{equation*}
\begin{split}
&n^{\delta}  \Bigg[\frac{1}{n}\sum_{i=1}^n  \mathbb{E}_{Z'_t}\Big\{\phi_b(b_i^0, ..., b_i^{t-1}, \frac{\bar{\sigma}_t^2}{\bar{\sigma}_{t-1}^2} b^{t-1}_i +  \bar{\sigma}_{t}^{\perp} \, Z'_{t_i}, w_i)\Big\} -\\
&\mathbb{E} \, \mathbb{E}_{Z'_t} \Big\{\phi_b(\bar{\sigma}_0 \hat{Z}_0, ..., \bar{\sigma}_{t-1} \hat{Z}_{t-1}, \frac{\bar{\sigma}_t^2}{\bar{\sigma}_{t-1}} \hat{Z}_{t-1} +  \bar{\sigma}_{t}^{\perp} \, Z'_{t}, \sigma Z_w) \Big\} \Bigg]
\end{split}
\end{equation*}
is almost surely $0$.

The proof is completed by noting that  $\Big( (\bar{\sigma}_t^2/\bar{\sigma}_{t-1}) \hat{Z}_{t-1} +  \bar{\sigma}_{t}^{\perp} \, Z'_{t} \Big)$ is a Gaussian random variable with variance  $(\bar{\sigma}_t^2/\bar{\sigma}_{t-1})^2 + (\bar{\sigma}_{t}^{\perp})^2 = \bar{\sigma}_t^2$, where we have used the definition of $\bar{\sigma}_{t}^{\perp}$ from \eqref{eq:sigperp_defs} and the fact that $\hat{Z}_{t-1}$ and $Z'_{t}$ are independent.  Note also that for $0 \leq r \leq t-1$,
\ben
\mathbb{E}\left[\bar{\sigma}_r \hat{Z}_r ( \frac{\bar{\sigma}_t^2}{\bar{\sigma}_{t-1}} \hat{Z}_{t-1} +  \bar{\sigma}_{t}^{\perp} Z'_{t} )\right] \overset{(a)}{=} \frac{\bar{\sigma}_t^2}{\bar{\sigma}_{t-1}} \bar{\sigma}_r \mathbb{E}[\hat{Z}_r \hat{Z}_{t-1}]  \overset{(b)}{=} \bar{\sigma}_t^2,
\een
where (a) holds since $\hat{Z}_r$, $Z'_{t}$ are independent and (b) because $ \bar{\sigma}_r  \bar{\sigma}_{t-1} \mathbb{E}\left[\hat{Z}_r \hat{Z}_{t-1} \right] =  \bar{\sigma}_{t-1}^2$.

\textbf{(c)} The function $\phi_b(b^0_i, \ldots, b^t_i, w_i) := b^t_i w_i \in PL(2)$.  By $\mc{B}_0(b)$, $\lim \frac{(b^t)^*w}{n} \overset{a.s.}{=} \mathbb{E}\{\bar{\sigma}_t \hat{Z}_t \sigma Z_w)\} = 0$, and the convergence rate is $n^{-\delta}$.

\textbf{(d)} The function $\phi_b(b^0_i, \ldots, b^t_i, w_i) := b^r_i b^t_i \in PL(2)$ for $0 \leq r \leq t$.  By $\mc{B}_0(b)$, $\lim \frac{(b^r)^*b^t}{n} \overset{a.s.}{=} \mathbb{E}\{\bar{\sigma}_r \hat{Z}_r \bar{\sigma}_t \hat{Z}_t\} = \bar{\sigma}_t^2$  and the convergence rate is $n^{-\delta}$.

\textbf{(e)} Recall $m^r = b^r - w$ for $0 \leq r \leq t$.  The function $\phi_b(b^0_i, \ldots, b^t_i, w_i) := (b^r_i - w_i)(b^t_i - w_i) \in PL(2)$.  By $\mc{B}_0(b)$, $\lim \frac{(m^r)^*m^t}{n} \overset{a.s.}{=} \mathbb{E}\{(\bar{\sigma}_r \hat{Z}_r - \sigma Z_w)(\bar{\sigma}_t \hat{Z}_t - \sigma Z_w)\} =\mathbb{E}\{(\bar{\sigma}_r \hat{Z}_r \bar{\sigma}_t \hat{Z}_t \} + \sigma^2 = \bar{\sigma}_t^2 + \sigma^2 = \bar{\tau}_t^2$  and the convergence rate is $n^{-\delta}$.

\textbf{(f)} The function $\phi_b(b^0_i, \ldots, b^t_i, w_i) := b^r_i (b^s_i - w_i) \in PL(2)$ for $0 \leq r,s \leq t$.  By $\mc{B}_0(b)$, $\lim \frac{(b^r)^*m^s}{n} \overset{a.s.}{=} \mathbb{E}\{\bar{\sigma}_r \hat{Z}_r(\bar{\sigma}_s \hat{Z}_s - \sigma Z_w)\} = \bar{\sigma}_{\max(r,s)}^2$  and the convergence rate is $n^{-\delta}$.

\textbf{(g)} Note that $\vec{\alpha}^t = \left( \frac{M_t^* M_t}{n} \right)^{-1} \frac{M_t^* m^t}{n}$.  We first show that the matrix $\frac{M_t^* M_t}{n}$ is invertible with a finite limit. From the induction hypotheses $\mc{B}_{0}$(e)--$\mc{B}_{t-1}$(e), $\lim \frac{1}{n}(m^r)^*m^s \overset{a.s.}{=} \bar{\tau}^2_{\max(r,s)}$ at rate $n^{-\delta}$ for $0 \leq r,s \leq (t-1).$ Further, $\mc{B}_{0}$(h)--$\mc{B}_{t-1}$(h) and Fact \ref{fact:eig_proj} together imply that the smallest eigenvalue of the matrix $\frac{M_t^* M_t}{n}$ is bounded from below by a positive constant for all $n$; then Fact \ref{fact:eig_conv} implies that its inverse has a finite limit. Further, the inverse converges  to its limit at rate $n^{-\delta}$ as each entry in $\frac{M_t^* M_t}{n}$ converges at this rate.  Next, using $\mc{B}_{0}$(e)--$\mc{B}_{t-1}$(e), 
\be
\begin{split}
\lim \vec{\alpha}^t &= \lim \left( \frac{M_t^* M_t}{n} \right)^{-1} \frac{M_t^* m^t}{n} \\
& \overset{(a)}{=} C^{-1} \textsf{e}_t \bar{\tau}_t^2 \overset{(b)}{=} \left(0,\ldots,0, \frac{\bar{\tau}^2_{t}}{\bar{\tau}^2_{t-1}}\right)^*,
\label{eq:alphat_conv}
\end{split}
\ee
In step $(a)$, the matrix $C \in \mathbb{R}^{t \times t}$ has entries $C_{i,j} = \bar{\tau}_{\max(i-1,j-1)}^2$ for $1 \leq i,j \leq t$ and $\textsf{e}_t \in \mathbb{R}^t$ denotes the all-ones column vector.  The equality $(b)$ is obtained  as follows:  first, note that $C^{-1}  \mathsf{e}_{t}$ is the solution to $Cx = \mathsf{e}_{t}$. Next, since all the entries in the last column of $C$ are equal to $\bar{\tau}^2_{t-1}$, by inspection the solution to $Cx = \mathsf{e}_{t}$ is $x = [0,\ldots,0, (\bar{\tau}^2_{t-1})^{-1}]^*$, which yields $(b)$ in \eqref{eq:alphat_conv}.

\textbf{(h)}   Let $\mathsf{P}_{M_{t}} = M_{t}(M_{t}^* M_{t})^{-1} M_{t}^*$ be the projection matrix onto the column space of $M_{t}$. Note that $M_{t}^* M_{t}$ is invertible with a finite limit in $\mc{B}_t$(g).  Then,
\begin{equation}
\begin{split}
&\frac{\norm{m_{\perp}^t}^2}{n} = \norm{(\mathsf{I} - \mathsf{P}_{M_{t}})m^{t}}^2 \\
& =  \frac{\norm{m^{t}}^2}{n} - \frac{(m^{t})^*M_t}{n} \cdot \left(\frac{M_t^*M_t}{n}\right)^{-1} \cdot\frac{M_t^*m^{t}}{n}.
\end{split}
\label{eq:stepbteq1}
\end{equation}
Using the representation in \eqref{eq:stepbteq1}, it follows by $\mathcal{B}_{0}$(e) - $\mathcal{B}_{t}$(e),
\begin{equation*}
\lim \frac{\norm{m_{\perp}^{t}}^2}{n} \overset{(a)}{=}  \bar{\tau}^2_{t} - \bar{\tau}_{t}^2 \textsf{e}_t^* C^{-1} \textsf{e}_t \bar{\tau}_{t}^2 \overset{(b)}{=} \bar{\tau}^2_{t} - \frac{ \bar{\tau}_{t}^4}{\bar{\tau}^2_{t-1}} = (\bar{\tau}_t^{\perp})^2.
\label{eq:stepbteq2}
\end{equation*}
In step $(a)$, the matrix $C \in \mathbb{R}^{t \times t}$ has entries $C_{i,j} = \bar{\tau}_{\max(i-1,j-1)}^2$ for $1 \leq i,j \leq t$ and $\textsf{e}_t \in \mathbb{R}^t$ denotes the all-ones column vector.  The equality $(b)$ follows from the same reasoning as in \eqref{eq:alphat_conv}.

Finally since $\bar{\tau}^2_r = \sigma^2 \left( 1 + \snr \right)^{1-\xi_{r-1}} $ for $0 \leq r \leq t$, the definition of $\xi_{r-1}$ in \eqref{eq:lim_alph} implies that $(\bar{\tau}_t^{\perp})^2$ is strictly positive for $r \leq T^*$, where $T^*=  \left\lceil \frac{2 \mc{C}}{\log(\mc{C}/R)} \right \rceil$.

\subsubsection{Step 4:} \label{subsub:step4} \emph{Showing $\mathcal{H}_{t+1}$ holds.}

\textbf{(a)}  Let $\textbf{Q}_{t+1} := \frac{Q_{t+1}^* Q_{t+1}}{n}$ and $v' := \frac{B^*_{t+1} m_t^{\perp}}{n} - \frac{Q_{t+1}^*}{n}\left[q^t - \sum_{i=0}^{t-1} \alpha^t_i q^i\right]$.  From the definition of $\Delta_{t+1,t}$ in Lemma \ref{lem:hb_cond} \eqref{eq:Dt1t}, we have 
\be
\begin{split}
&\Delta_{t+1,t} = \sum_{r=0}^{t-2} \alpha^t_r h^{r+1} +  \left(\alpha^t_{t-1} -  \frac{\bar{\tau}_t^2}{\bar{\tau}_{t-1}^2}\right) h^{t}  \\
& + \left(\frac{\norm{m^t_{\perp}}}{\sqrt{n}} - \bar{\tau}_{t}^{\perp}\right) Z_t - \frac{\norm{m^t_{\perp}}}{\sqrt{n}}\sum_{r'=0}^{t} \tilde{q}^{r'} \frac{\bar{Z}_{t+1_r}}{\sqrt{n}} \\
&  + Q_{t+1} \textbf{Q}_{t+1}^{-1}v',
\end{split}
 \label{eq:newDeltat1t1}
\ee
where we have used Fact \ref{fact:gauss_p0} to write
\ben
\frac{\norm{m^t_{\perp}} }{\sqrt{n}} \mathsf{P}_{Q_{t+1}}Z_t \overset{d}{=} \frac{\norm{m^t_{\perp}} \tilde{Q}_{t+1} \bar{Z}_{t+1}}{n} =  \frac{\norm{m^t_{\perp}}}{\sqrt{n}}\sum_{s=0}^{t} \tilde{q}^{s} \frac{\bar{Z}_{t+1_r}}{\sqrt{n}}.
\een
The matrix $\tilde{Q}_{t+1}=[\tilde{q}^0| \ldots| \tilde{q}^t]$ forms an orthogonal basis for the columns of $Q_{t+1}$ such that $\norm{\tilde{q}^s} = \sqrt{n}$ and $\bar{Z}_{t+1} \in \mathbb{R}^{t+1}$ is an independent  i.i.d.\ $\mc{N}(0,1)$ random vector. It follows from \eqref{eq:newDeltat1t1} that
\be
\label{eq:newDeltat1teq1}
\begin{split}
&\max_{j \in sec(\ell)} \abs{\Delta_{t+1,t}} \leq \sum_{r=0}^{t-2} \abs{\alpha^t_r} \max \abs{h^{r+1}_j} \\
& \ +  \left \lvert \alpha^t_{t-1} -  \frac{\bar{\tau}_t^2}{\bar{\tau}_{t-1}^2}\right \lvert \max \abs{h^{t}_j} + \left \lvert \frac{\norm{m^t_{\perp}}}{\sqrt{n}} - \bar{\tau}_{t}^{\perp}\right \lvert \max \abs{Z_{t_j}}  \\
&\ +  \frac{\norm{m^t_{\perp}}}{\sqrt{n}} K \sqrt{nP_{\ell}} \sum_{r'=0}^{t} \frac{\abs{\bar{Z}_{t+1_{r'}}}}{\sqrt{n}} \\
& \ + K \sqrt{nP_{\ell}} \sum_{r=0}^{t} \left \lvert [\textbf{Q}_{t+1}^{-1}v']_{r+1}\right \lvert.
\end{split}
\ee
In the above we have used  $Q_{t+1} \textbf{Q}_{t+1}^{-1}v' = \sum_{r=0}^{t} q^r [\textbf{Q}_{t+1}^{-1}v']_{r+1}$, and the fact that both $\max \abs{q^r_j}$ and $\max \abs{\tilde{q}^r_j}$ are bounded by $K \sqrt{nP_{\ell}}$ for some constant $K>0$.  

We show that all terms on the RHS of \eqref{eq:newDeltat1teq1} are $o(n^{-\delta} \sqrt{\log M})$ almost surely.   This is true of the first two terms by $\mc{H}_1$(a)--$\mc{H}_t$(a), and  $\mc{B}_t$(g), which says that almost surely $\abs{\alpha^t_r} \in o(n^{-\delta})$ for $0 \leq r \leq t-2$ and $\left \lvert \alpha^t_{t-1} -  (\bar{\tau}_t^2/\bar{\tau}^2_{t-1})\right \lvert \in o(n^{-\delta})$.  By Fact \ref{lem:maxZbound} and $\mc{B}_t$(h) the third term is almost surely $o(n^{-\delta} \sqrt{\log M})$.  Considering the fourth term, $\norm{m^t_{\perp}}/\sqrt{n}$ has a bounded limiting value by $\mc{B}_t$(h), $\sqrt{nP_{\ell}} = \Theta(\sqrt{\log M})$ and $\abs{\bar{Z}_{t+1, r'}}/\sqrt{n} \in o(n^{-\delta})$ a.s. for $0 \leq r' \leq t$.  Finally, the fifth term is almost surely $o(n^{-\delta}\sqrt{\log M})$ if we can show that $\left \lvert [\textbf{Q}_{t+1}^{-1}v']_{r+1}\right \lvert \in o(n^{-\delta})$ for each $0 \leq r \leq t$.  We prove this in what follows.

Note that
\be
\begin{split}
&[\textbf{Q}_{t+1}^{-1}v']_{r+1} = \\
& \quad \begin{cases}
\alpha^t_{r+1} +  \left[\, \left(\tfrac{Q_{t+1}^*Q_{t+1}}{n} \right)^{-1} \tfrac{B_{t+1}^* m_{\perp}^t}{n} \, \right]_{r+1}  \text{ for } 0 \leq  r \leq t, \\
- 1 + \left[\, \left(\tfrac{Q_{t+1}^*Q_{t+1}}{n} \right)^{-1} \tfrac{B_{t+1}^* m_{\perp}^t}{n} \, \right]_{t},  \text{ for }  r = t. \label{eq:coeff_compare_H}
\end{cases}
\end{split}
\ee
We show that each of the above coefficients is  $o(n^{-\delta})$. Indeed, for $1 \leq i \leq t+1$,
\ben
\begin{split}
&\left[ \frac{B_{t+1}^* m_{\perp}^t}{n}\right]_{i} = \frac{(b^{i-1})^* m_{\perp}^t}{n}  = \frac{(b^{i-1})^*(m^t - m_{\parallel}^t)}{n} \\
&\qquad = \frac{(b^{i-1})^* m^t}{n} - \sum_{r=0}^{t-1} \alpha^t_r \frac{(b^{i-1})^* m^r}{n} \\
&\qquad  \stackrel{a.s.}{\to}   \lim\left[\frac{(q^{i})^* q^{t}}{n} -\sum_{r=0}^{t-1} \alpha^t_r \frac{(q^{i})^* q^{r}}{n}\right], 
\end{split}
\een
where the convergence (at rate $n^{-\delta}$) follows from $\mc{B}_t$(e), $\mc{B}_t$(f), and $\mc{B}_t$(g) (convergence of $\vec{\alpha}^t$ to finite values). Therefore, at rate $n^{-\delta}$,
\be 
\left[ \frac{B_{t+1}^* m_{\perp}^t}{n}\right] \stackrel{a.s.}{\to}  \lim\left[\frac{(Q_{t+1})^* q^{t}}{n} - \sum_{r=0}^{t-1} \alpha^t_{r+1} \frac{(Q_{t+1})^* q^{r}}{n}\right],
\label{eq:Btmt_conv} 
\ee
and substituting \eqref{eq:Btmt_conv} in \eqref{eq:coeff_compare_H} we see that each coefficient of \eqref{eq:coeff_compare_H} is $o(n^{-\delta})$.  This completes the proof demonstrating $\max_{j \in sec(\ell)} \abs{[\Delta_{t+1, t}]_{j}} \overset{a.s}{=} \Theta\left( n^{-\delta} \sqrt{\log M}\right)$.
 
Next, from Lemma \ref{lem:hb_cond} \eqref{eq:Ha_dist} it follows,
\begin{align*}
&\max_{j \in sec(\ell)} \abs{h^{t+1}_j} \\
&\leq \frac{\bar{\tau}_t^2}{\bar{\tau}_{t-1}^2} \max_{j \in sec(\ell)} \abs{ h^{t}_j}  + \abs{\bar{\tau}_t^{\perp}} \max_{j \in sec(\ell)} \abs{Z_{t_j}} + \max_{j \in sec(\ell)} \abs{\Delta_{{t+1,t}_j}} \\
& \overset{a.s}{\leq} \frac{\bar{\tau}_t^2}{\bar{\tau}_{t-1}^2} c_{t}\sqrt{\log M} + \abs{\bar{\tau}_0}\Theta\left(\sqrt{\log M}\right) + \Theta\left( n^{-\delta} \sqrt{\log M}\right).
\end{align*}
The second inequality above comes from $\mc{H}_t$(a), Fact \ref{lem:maxZbound}, and the first result of $\mc{H}_{t+1}(a)$ proved above.

\textbf{(b)}  As in the proof of $\mc{H}_1 (b)$, we provide the main steps of the proof, referring the reader to \cite{steps24b} for details. Throughout we use generic $\phi_{k, \ell}(x, y, z)$ as the steps are identical for all $k \in \{1, 2, 3, 4\}$.    From Lemma \ref{lem:hb_cond} \eqref{eq:Ha_dist},
\begin{equation*}
\begin{split}
\phi_{k, \ell} &\left.\left( \sum_{u=0}^{t} a_u h^{u+1}_{\ell}, \sum_{v=0}^{t} b_v h^{v+1}_{\ell}, \beta_{0_{\ell}}\right) \right \lvert_{\mscrs_{t+1,t}} \\
&\overset{d}{=} \phi_{k, \ell}\left( \sum_{u=0}^{t-1} a'_u h^{u+1}_{\ell} + a_t \bar{\tau}_{t}^{\perp} Z_{t_{\ell}} + a_t[\Delta_{t+1,t}]_{\ell}, \right. \\
&\qquad \qquad \left. \sum_{v=0}^{t-1} b'_v h^{v+1}_{\ell} + b_t \bar{\tau}_{t}^{\perp} Z_{t_{\ell}} + b_t[\Delta_{t+1,t}]_{\ell}, \beta_{0_{\ell}}\right),
\end{split}
\end{equation*}
where $a'_u = a_u$ and $b'_v = b_v$ for $0 \leq u,v \leq t-2$ and $a'_{t-1} = a_{t-1} + a_t \frac{\bar{\tau}_t^2}{\bar{\tau}_{t-1}^2}$ and $b'_{t-1} = b_{t-1} + b_t \frac{\bar{\tau}_t^2}{\bar{\tau}_{t-1}^2}$.  By $\mc{H}_{t+1}$(a), $\max_{j \in sec(\ell)} \abs{[\Delta_{t+1, t}]_{j}} \overset{a.s.}{=} o(n^{-\delta'} \sqrt{\log M} )$ for each $\ell \in [L]$ and some $\delta' >0$.   In \cite{steps24b}, the first step of the proof uses this to show for each of the functions in \eqref{eq:phih_fns},\begin{equation*}
\begin{split}
&\frac{1}{L}  \sum_{\ell=1}^L  \left \lvert  \phi_{k, \ell}\left( \sum_{u=0}^{t-1} a'_u h^{u+1}_{\ell} + a_t \bar{\tau}_{t}^{\perp} Z_{t_{\ell}} + a_t[\Delta_{t+1,t}]_{\ell}, \right. \right. \\
& \qquad \qquad \qquad \left. \left. \sum_{v=0}^{t-1} b'_v h^{v+1}_{\ell} + b_t \bar{\tau}_{t}^{\perp}Z_{t_{\ell}} + b_t[\Delta_{t+1,t}]_{\ell}, \beta_{0_{\ell}}\right) \right. - \\
&\left.   \phi_{k, \ell}\left( \sum_{u=0}^{t-1} a'_u h^{u+1}_{\ell} + a_t \bar{\tau}_{t}^{\perp} Z_{t_{\ell}}, \sum_{v=0}^{t-1} b'_v h^{v+1}_{\ell} + b_t \bar{\tau}_{t}^{\perp} Z_{t_{\ell}}, \beta_{0_{\ell}}\right)\right \lvert
\end{split}
\end{equation*}
is almost surely $o({n^{- \delta'} \log M})$.  Choosing $\delta \in (0, \delta')$ ensures that we can drop the deviation terms $\Delta_{t+1, t}$.  

The second step of the proof appeals to Fact \ref{fact:slln} to show that the limit of the expression
\begin{small}
\begin{equation}
\begin{split}
& \frac{n^{\delta}}{L} \hspace{-2pt} \sum_{\ell=1}^L \hspace{-1pt} \Big[ \phi_{k, \ell} \Big ( \sum_{u=0}^{t-1} a'_u h^{u+1}_{\ell} + a_t \bar{\tau}_{t}^{\perp} Z_{t_{\ell}}, \sum_{v=0}^{t-1} b'_v h^{v+1}_{\ell} + b_t \bar{\tau}_{t}^{\perp}Z_{t_{\ell}}, \beta_{0_{\ell}} \Big )  \\
& -  \mathbb{E}_{Z_t} \phi_{k, \ell}\Big( \sum_{u=0}^{t-1} a'_u h^{u+1}_{\ell} + a_t \bar{\tau}_{t}^{\perp} Z_{t_{\ell}}, \sum_{v=0}^{t-1} b'_v h^{v+1}_{\ell} + b_t \bar{\tau}_{t}^{\perp}Z_{t_{\ell}}, \beta_{0_{\ell}} \Big) \Big ]
\label{eq:2b51}
\end{split}
\end{equation}
\end{small}
is almost surely $0$.  Let $\tilde{Z}_t$ be an independent copy of $Z_t$. Define the value $\textsf{diff}_{k, \ell}$ to be the following difference for each $\ell \in [L]$ and each function in \eqref{eq:phih_fns} with $k = 1,2 , 3, 4$,
\begin{equation*}
\begin{split}
&\textsf{diff}_{k, \ell} := \\
& \phi_{k, \ell}\Big( \sum_{u=0}^{t-1} a'_u h^{u+1}_{\ell} + a_t \bar{\tau}_{t}^{\perp} \tilde{Z}_{t_{\ell}}, \sum_{v=0}^{t-1} b'_v h^{v+1}_{\ell} + b_t \bar{\tau}_{t}^{\perp} \tilde{Z}_{t_{\ell}}, \beta_{0_{\ell}} \Big)  \\
& - \phi_{k, \ell}\Big( \sum_{u=0}^{t-1} a'_u h^{u+1}_{\ell} + a_t \bar{\tau}_{t}^{\perp} Z_{t_{\ell}}, \sum_{v=0}^{t-1} b'_v h^{v+1}_{\ell} + b_t \bar{\tau}_{t}^{\perp} Z_{t_{\ell}}, \beta_{0_{\ell}} \Big) 
\end{split}
\end{equation*}
In order to use Fact \ref{fact:slln} (conditionally on $\mscrs_{t+1, t}$) to get the above result we must prove that
\begin{equation}
\begin{split}
\label{eq:2b4a1}
&\frac{1}{L}\sum_{\ell=1}^L\mathbb{E}_{\tilde{Z}_t, Z_t} \left\lvert n^{\delta} \textsf{diff}_{k, \ell} \right\lvert^{2 + \kappa}\leq cL^{\kappa/2},
\end{split}
\end{equation}
for some constants $0 \leq \kappa \leq 1$ and  $c>0$.  The exact condition required by Fact \ref{fact:slln} follows from \eqref{eq:2b4a1} by an application of Jensen's inequality.   In \cite{steps24b} it is shown that for each function in \eqref{eq:phih_fns} and each $\ell \in [L]$,
\begin{equation}
\begin{split}
 \mathbb{E}_{\tilde{Z}_t, Z_t}& \left\lvert\textsf{diff}_{k , \ell} \right\lvert^{2 + \kappa} \stackrel{a.s.}{=} O((\log M)^{2+\kappa}).
\label{eq:2b4a2}
\end{split}
\end{equation}
Bound \eqref{eq:2b4a2} implies \eqref{eq:2b4a1} holds if $\delta$ is chosen such that $\delta(2 + \kappa) < \kappa/2$. Hence \eqref{eq:2b51} holds.

Considering result \eqref{eq:2b51}, define new functions $\phi^{NEW}_{k, \ell}$ for $k \in \{1, 2, 3, 4\}$ as
\begin{small}
\ben
\begin{split}
&\phi^{NEW}_{k, \ell}\left( \sum_{u=0}^{t-1} a'_u h^{u+1}_{\ell} ,  \sum_{v=0}^{t-1} b'_v h^{v+1}_{\ell},  \beta_{0_\ell} \right) := \\
&\mathbb{E}_{Z_t}   \phi_{k, \ell}\Big(\sum_{u=0}^{t-1} a'_u h^{u+1}_{\ell}  + a_{t}\bar{\tau}_{t}^{\perp} Z_{t_{\ell}},  \sum_{v=0}^{t-1} b'_v  h^{v+1}_{\ell} + b_{t} \bar{\tau}_{t}^{\perp} Z_{t_{\ell}}, \beta_{0_{\ell}}\Big).
\end{split}
\een
\end{small}
Using Jensen's inequality, it can be shown that the induction hypothesis $\mc{H}_{t}(b)$ holds for the function $\phi^{NEW}_{k, \ell}$ whenever $\mc{H}_t$(b) holds for the function $\phi_{k, \ell}$ inside the expectation.  This work can be found in \cite{steps24b}. Therefore, the limit of
\ben
\begin{split}
\frac{n^{\delta} }{L}\sum_{\ell=1}^L & \left\{  \mathbb{E}_{Z_t} \Big[ \phi_{k, \ell} \Big( \sum_{u=0}^{t-1} a'_u h^{u+1}_{\ell} + a_{t}\bar{\tau}_{t}^{\perp} Z_{t_{\ell}},  \right. \\
& \hspace{1in}  \sum_{v=0}^{t-1} b'_v  h^{v+1}_{\ell} + b_{t} \bar{\tau}_{t}^{\perp} Z_{t_{\ell}}, \beta_{0_{\ell}} \Big) \Big ] \\
& -  \mathbb{E}\mathbb{E}_{Z_t} \Big[ \phi_{k, \ell} \Big( \sum_{u=0}^{t-1} a'_u \bar{\tau}_u\breve{Z}_{u_{\ell}}  + a_{t}\bar{\tau}_{t}^{\perp} Z_{t_{\ell}}, \\
&  \hspace{1in}  \left. \sum_{v=0}^{t-1} b'_v  \bar{\tau}_v\breve{Z}_{v_{\ell}}+ b_{t} \bar{\tau}_{t}^{\perp} Z_{t_{\ell}}, \beta_{\ell}\Big) \Big ] \right\}
\end{split}
\een
is almost surely $0$.  To complete the proof we show that
\begin{small}
\ben
\begin{split}
&\mathbb{E}\mathbb{E}_{Z_t} \Big[ \phi_{k, \ell} \Big( \sum_{u=0}^{t-1} a'_u \bar{\tau}_u\breve{Z}_{u}  + a_{t}\bar{\tau}_{t}^{\perp} Z_{t}, \, \sum_{v=0}^{t-1} b'_v  \bar{\tau}_v\breve{Z}_{v}+ b_{t} \bar{\tau}_{t}^{\perp} Z_{t}, \, \beta_{\ell} \Big) \Big]  \\
&\qquad = \mathbb{E} \Big[ \phi_{k, \ell} \Big(\sum_{u=0}^{t} a_u \bar{\tau}_u\breve{Z}_{u}, \, \sum_{v=0}^{t} b_v  \bar{\tau}_v\breve{Z}_{v}, \, \beta_{\ell}
\Big) \Big].
\end{split}
\een
\end{small}
Recall $a'_{t-1}  = a_{t-1} + a_t (\bar{\tau}_t^2/\bar{\tau}_{t-1}^2)$ and $b'_{t-1} = b_{t-1} + b_t (\bar{\tau}_t^2/\bar{\tau}_{t-1}^2)$. Then to prove the above, we will show that $(\bar{\tau}_t^2/\bar{\tau}_{t-1}) \breve{Z}_{t-1} + \bar{\tau}_{t}^{\perp} Z_{t} \overset{d}{=} \bar{\tau}_t \breve{Z}_t$ where $\bar{\tau}_r \bar{\tau}_t \mathbb{E}[ \breve{Z}_r \breve{Z}_t] = \bar{\tau}_t^2$ for $0 \leq r \leq t-1$.  Indeed, $\left((\bar{\tau}_t^2/\bar{\tau}_{t-1}) \breve{Z}_{t-1}  + \bar{\tau}_{t}^{\perp} Z_{t}\right)$ is Gaussian with variance equal to $(\bar{\tau}_t^2/\bar{\tau}_{t-1})^2 + (\bar{\tau}_{t}^{\perp} )^2 = \bar{\tau}_t^2$, using the definition of $\bar{\tau}_{t}^{\perp}$ in \eqref{eq:sigperp_defs} and the independence of $\breve{Z}_{t-1} $ and $Z_{t}$.  Further, for $0 \leq r \leq t-1$
\ben
\begin{split}
&\mathbb{E}\left[ \bar{\tau}_r \breve{Z}_r \left((\bar{\tau}_t^2/\bar{\tau}_{t-1}) \breve{Z}_{t-1} + \bar{\tau}_{t}^{\perp} Z_{t}\right)\right] \\
& = (\bar{\tau}_t^2/\bar{\tau}_{t-1}^2)  \bar{\tau}_r \bar{\tau}_{t-1}\mathbb{E}[ \breve{Z}_r \breve{Z}_{t-1}] = \bar{\tau}_t^2.
\end{split}
\een
The existence of the limit of $\mathbb{E}\{ \phi_{k, \ell} (\sum_{u=0}^{t} a_u \bar{\tau}_u\breve{Z}_{u}, \, \sum_{v=0}^{t} b_v  \bar{\tau}_v\breve{Z}_{v}, \, \beta_{\ell})\}$  for $k=1$ follows from the law of large numbers; for $k=2,3,4$, the existence of the limit follows from Appendix \ref{app:qrqs_lim}.

\textbf{(c)}, \textbf{(d)}, \textbf{(e)} These are shown by invoking $\mc{H}_{t+1}(b)$, and are similar to the corresponding results for step $\mc{H}_1$.

\textbf{(f)} Using the fourth function in \eqref{eq:phih_fns} for any $0 \leq r,s \leq t$ by $\mc{H}_{t+1}$(b),
\ben
\lim \frac{(h^{s+1})^* q^{r+1}}{n} \overset{a.s.}{=} \lim \frac{1}{n} \sum_{\ell = 1}^L \mathbb{E}\{\bar{\tau}_s \breve{Z}_{s_{\ell}}^*[\eta^r_{\ell}(\beta - \bar{\tau}_r \breve{Z}_r) - \beta_{\ell}]\},
\label{eq:4d1}
\een
and the convergence is $o(n^{-\delta})$.  Using arguments very similar to those in $\mc{H}_1$(f) (iterated expectations and Stein's lemma), we obtain that
\be
\begin{split}
&\mathbb{E}\{\bar{\tau}_s \breve{Z}_{s_{\ell}}^*[\eta^r_{\ell}(\beta - \bar{\tau}_r \breve{Z}_r) - \beta_{\ell}]\} \\
& \quad =\frac{\bar{\tau}_s }{\bar{\tau}_r}  \expec[\breve{Z}_{s_1} \breve{Z}_{r_1}] \left( \mathbb{E}\norm{\eta^r_{\ell}(\beta - \bar{\tau}_r \breve{Z}_{r})}^2 - nP_{\ell}\right) \\
&\quad = \frac{\bar{\tau}_{\max(r,s)}^2}{\bar{\tau}_r^2} \left( \mathbb{E}\norm{\eta^r_{\ell}(\beta - \bar{\tau}_r \breve{Z}_{r})}^2 - nP_{\ell}\right), \quad \ell \in [L].
\label{eq:4d11}
\end{split}
\ee
Here $\breve{Z}_{s_1}, \breve{Z}_{r_1}$ refer to  the first entries of the vectors $\breve{Z}_s, \breve{Z}_r$, respectively. Using  \eqref{eq:4d11} along with the fact that $\Big(P-\tfrac{1}{n}\mathbb{E}\left\{\norm{\eta^{r}(\beta - \bar{\tau}_r Z_r)}^2\right\}\Big)  \to \bar{\sigma}^2_{r+1}$ (cf. Appendix \ref{app:qrqs_lim}), \eqref{eq:4d1} becomes
\ben
\lim \frac{(h^{s+1})^* q^{r+1}}{n} \overset{a.s.}{=} -\frac{\bar{\tau}_{\max(r,s)}^2  \bar{\sigma}^2_{r+1}}{\bar{\tau}_r^2}.
\een
Next, from \eqref{eq:lambda_t_def}, we observe that
\be
\begin{split}
&\lambda_{r+1} =  \frac{1}{\bar{\tau}^2_r}\left(  \frac{\norm{\beta^{r+1}}^2}{n} - P \right) \\
&\stackrel{a.s.}{\to} \lim \frac{1}{\bar{\tau}_r^2}\left( \frac{\mathbb{E}\norm{\eta^{r}(\beta - \bar{\tau}_r \breve{Z}_r)}^2}{n} - P\right)  = 
\frac{-\bar{\sigma}^2_{r+1}}{\bar{\tau}^2_r},
\label{eq:lams1_conn} 
\end{split}
\ee
where the convergence at rate $n^{-\delta}$ follows from $\mc{H}_{t+1}$(b) applied to the second function in \eqref{eq:phih_fns}.  The last equality in \eqref{eq:lams1_conn} is from Appendix \ref{app:qrqs_lim}. By $\mc{B}_t$(e) $ \lim \, (m^r)^*m^s/n \overset{a.s.}{=} \bar{\tau}_{\max(r,s)}^2$, which along with \eqref{eq:lams1_conn} completes the proof.

\textbf{(g)} Note that $\vec{\gamma}^{t+1} = \left( \frac{Q_{t+1}^* Q_{t+1}}{n} \right)^{-1} \frac{Q_{t+1}^* q^{t+1}}{n}$.  Similarly to the proof of step $\mc{B}_t$(g), the matrix $\frac{1}{n}Q_{t+1}^* Q_{t+1}$ can be shown to be invertible with a finite limit using $\mc{H}_{1}$(e) -- $\mc{H}_{t}$(e), $\mc{H}_1$(h) -- $\mc{H}_{t}$(h), Fact \ref{fact:eig_proj}, and Fact \ref{fact:eig_conv}.  Then use $\mc{H}_{1}$(e) -- $\mc{H}_{t}$(e) to find the value of the limit of $\vec{\gamma}^{t+1}$.

\textbf{(h)}   This result follows similarly to $\mc{B}_t$(h) but uses the convergence results $\mathcal{H}_{1}$(e) -- $\mathcal{H}_{t+1}$(e).

\appendix
\subsection{AMP Derivation} \label{app:amp_derive}

In \eqref{eq:z_update}, the dependence of $z^t_{a \to i}$ on $i$ is only due to the  term $A_{ai} \beta^t_{i \to a}$ being excluded from the sum. Similarly, in \eqref{eq:beta_update} the dependence of $\beta^t_{i \to a}$ on $a$ is  due to excluding the term $A_{ai} z^t_{a \to i}$ from the argument.
We begin by estimating the order of these excluded terms.  

Note that $A_{ai} = O(n^{-1/2})$, and $\beta^t_{i \to a} = O(\sqrt{\log n} )$. The latter is true since for $i$ in section $\ell$,
$\beta_i \leq \sqrt{nP_\ell}$, where $P_\ell=O(1/L)$, and $L = \Theta(n/ \log n)$. Therefore $ A_{ai} \beta^t_{i \to a} = O\left( \sqrt{\log n/n}  \right)$.
In \eqref{eq:beta_update}, the excluded term  $A_{bi} z^t_{b \to i}$ is  $O({n}^{-1/2})$ because $z^t_{b \to i} = O(1)$.
We set
\be
z^{t}_{a \to i} = z^t_a + \delta z^t_{a \to i}, \quad \text{ and } \quad \beta^{t+1}_{i \to a} = \beta^{t+1}_{i} + \delta\beta^{t+1}_{i \to a}.
\label{eq:delta_defs}
\ee
Comparing \eqref{eq:delta_defs} with \eqref{eq:z_update}, we can write
\begin{align}
z^t_a  = y_a - \sum_{j \in [N]} A_{aj} \beta^t_{j \to a}, \quad \delta z^t_{a \to i} = A_{ai} \beta^t_{i \to a}.
\label{eq:za_delza}
\end{align}

 For $i \in [N]$, let $\secl$ denote the set of indices in the section containing $i$. To determine $\delta\beta^t_{i \to a}$, we expand $\eta_i^t$ in \eqref{eq:beta_update} in a Taylor series around the argument
$ \left\{ \sum_{b \in [n] }  A_{bj} z^t_{b \to j}  \right\}_{j \in \secl}$,
which does not depend on $a$. We thus obtain
\be
\begin{split}
\beta^{t+1}_{i \to a} 
& \approx  \eta_i^t \Bigg( \Big\{ \sum_{b \in [n] }  A_{bj} z^t_{b \to j}  \Big\}_{j \in \secl } \Bigg) \\
& \quad - A_{ai} z^t_{a \to i}  \, \partial_i \eta^t_i \Bigg( \Big\{ \sum_{b \in [n] }  A_{bj} z^t_{b \to j}  \Big\}_{j \in \secl }  \Bigg),
\end{split}
\label{eq:tay_bt}
\ee
where $\partial_i \eta^t_i(.)$ is the partial derivative  of $\eta^t_i$ with respect to the component of the argument corresponding to index $i$. (Recall from \eqref{eq:eta_def} that the argument is a length $M$ vector.) From \eqref{eq:eta_def}, the partial derivative can be evaluated as
\be
\begin{split}
 \partial_i \eta^t_i(s) &= \eta^t_i(s) \; \partial_i \ln \eta^t_i(s) \\
 &=  \eta^t_i(s) \, \left(  \frac{\sqrt{nP_\ell}}{\tau_t^2}  -   \frac{\sqrt{nP_\ell}}{\tau_t^2}  \frac{e^{\frac{s_i \sqrt{n P_\ell}}{\tau^2_t} }} {\sum_{j \in \text{sec}(i)} \, e^{\frac{s_j \sqrt{n P_\ell}}{\tau^2_t} }}  \right) \\
 &=   \frac{\eta^t_i(s)}{\tau^2_t} \left(  \sqrt{nP_\ell}  -  \eta^t_i (s)   \right).
\end{split}
\label{eq:partial_calc}
\ee
Using \eqref{eq:partial_calc} in \eqref{eq:tay_bt} yields
\be
\begin{split}
\beta^{t+1}_{i \to a}   &=   \eta_i^t \Bigg( \Big\{ \sum_{b \in [n] }  A_{bj} z^t_{b \to j} \Big\}_{j \in \secl } \Bigg)  \\
&  \quad -  \frac{A_{ai} z^t_{a} }{\tau^2_t}
\, \eta_i^t \Bigg( \Big\{ \sum_{b \in [n] }  A_{bj} z^t_{b \to j} \Big\}_{j \in \secl } \Bigg) \\
 & \qquad \cdot \left[\sqrt{nP_\ell} -  \eta_i^t \Bigg( \Big\{ \sum_{b \in [n] }  A_{bj} z^t_{b \to j} \Big\}_{j \in \secl } \Bigg)   \right].
\end{split}
\label{eq:beta_tia}
\ee
Notice that we have replaced the stand-alone term $A_{ai} z^t_{a \to i}$ in \eqref{eq:tay_bt} with $A_{ai} z^t_{a}$ because the difference $A_{ai} \delta z^t_{a \to i}$ is $O(\sqrt{\log n}/n)$, which can be ignored --- we only keep terms as small as $O(n^{-1/2})$.

Since only the second term on  the right-hand side of \eqref{eq:beta_tia} depends on $a$, we can write
\be
\beta^{t+1}_{i} = \eta_i^t \left( \Big\{ \sum_{b \in [n] }  A_{bj} (z^t_b + \delta z^t_{b \to j}) \Big\}_{j \in \secl } \right),
\label{eq:beta_t}
\ee
and
\be
\begin{split}
&\delta\beta^{t+1}_{i \to a} =  -  \frac{A_{ai} z^t_{a} }{\tau^2_t}
\, \eta_i^t \left( \Big\{ \sum_{b \in [n] }  A_{bj} (z^t_b + \delta z^t_{b \to j}) \Big\}_{j \in \secl } \right) \\
& \ \cdot \left[\sqrt{nP_\ell}  -  \eta_i^t \left( \Big\{ \sum_{b \in [n]} A_{bj} (z^t_b + \delta z^t_{b \to j}) \Big\}_{j \in \secl } \right) \right].
\label{eq:delbet}
\end{split}
\ee
We observe that $\delta\beta^t_{i \to a} = O(\log n/\sqrt{n})$. Hence, in \eqref{eq:za_delza},
we can write
\be
\delta z^t_{a \to i} = A_{ai} \beta^t_{i}
\label{eq:delz}
\ee
because the difference  $A_{ai} \delta\beta^t_{i \to a} = O(\log n/ n)$. Substituting \eqref{eq:delz} in \eqref{eq:beta_t}, we see that
\be
\begin{split}
\beta^{t+1}_{i} &= \eta_i^t \Bigg( \Big\{ \sum_{b \in [n] } A_{bj} z^t_b +  A^2_{bj} \beta^t_{j} \Big\}_{j \in \secl } \Bigg) \\
&\stackrel{(a)}{=}
\eta_i^t \left( \Big\{ (A^* z^t +   \beta^t)_j \Big\}_{j \in \secl } \right),
\label{eq:bet_t1}
\end{split}
\ee
where  $(a)$ holds because $\sum_{b} A^2_{bj} \to 1$ as $n \to \infty$. Analogously, using \eqref{eq:delz} in \eqref{eq:delbet} gives
\be
\begin{split}
\delta\beta^{t+1}_{i \to a} &=  \frac{-A_{ai} z^t_{a} }{\tau^2_t}
\eta_i^t \Bigg( \Big\{ \sum_{b \in [n] } (A^* z^t +   \beta^t)_j \Big\}_{j \in \secl } \Bigg) \\
& \cdot \left[\sqrt{nP_\ell}  -  \eta_i^t \Big( \Big\{ \sum_{b \in [n]} (A^* z^t +   \beta^t)_j \Big\}_{j \in \secl } \Big) \right].
\label{eq:delbet_t1}
\end{split}
\ee

Finally, we use \eqref{eq:bet_t1} and \eqref{eq:delbet_t1} in \eqref{eq:za_delza} to obtain
\begin{align}
& z^t_a   = y_a - \sum_{k \in [N]} A_{ak} (\beta^t_k + \delta \beta^t_{k \to a}) \nonumber \\
& = y_a - \sum_{k \in [N]} A_{ak} \, \eta_k^{t-1} \left(  A^* z^{t-1} +   \beta^{t-1}  \right) \nonumber \\
& \quad + \frac{ A^2_{ak} z^{t-1}_{a} }{\tau^2_{t-1}} \, \eta_k^{t-1} \left(  A^* z^{t-1} +   \beta^{t-1}  \right) \nonumber \\
&\qquad \cdot
\left[ \sqrt{n P_{\text{sec}(k)}} - \eta_k^{t-1} \left(A^* z^{t-1} +   \beta^{t-1}  \right) \right] \nonumber \\
& \stackrel{(b)}{=} y_a - (A\beta^t)_a + \frac{z^{t-1}_{a} }{n\tau^2_{t-1}} \, (nP - \norm{\beta^t}^2),
\label{eq:zta1}
\end{align}
where $(b)$ is obtained as follows.   First,  we use $A^2_{ak} \approx \tfrac{1}{n}$. Next,  \eqref{eq:eta_def} implies that for all $s$,
\[ \sum_{k \in [N]}\sqrt{n P_{\text{sec}(k)}} \ \eta^t_k (s) =  \sum_{\ell=1}^L n P_\ell = nP. \]
Finally, note  from \eqref{eq:bet_t1} that $\sum_{k} (\eta_k^{t-1} \left(  A^* z^{t-1} + \beta^{t-1}  \right))^2 = \sum_k (\beta_k^t)^2  = \norm{\beta^t}^2$.
The AMP update equations are thus given by \eqref{eq:zta1} and \eqref{eq:bet_t1}. 

\subsection{Proof of Lemma \ref{lem:conv_expec}} \label{app:conv_exp}

From \eqref{eq:xt_tau_def}, $x(\tau)$ can be written as
\be
x(\tau) := \sum_{\ell=1}^{L} \frac{P_\ell}{P} \, \mc{E}_\ell(\tau),
\ee
where 
\be 
\begin{split}
&\mc{E}_\ell (\tau) =
\expec \left[
\frac{e^{\frac{\sqrt{n P_\ell}}{\tau} \, U^{\ell}_1}}
{ e^{ \frac{\sqrt{n P_\ell}}{\tau} \, U^{\ell}_1}  +  e^{-\frac{n P_\ell}{\tau^2}} \sum_{j=2}^M e^{\frac{\sqrt{n P_\ell}}{\tau}U^{\ell}_j } } \right].
\label{eq:Eell_def}
\end{split}
\ee
The result needs to be proved only for $\xi^* >0$. (For brevity, we supress the dependence of $\xi^*$ on $\tau$.) Since $P_\ell$ is non-increasing with $\ell$,  it is enough\footnote{We can also prove that $\lim \mc{E}_{\lfloor \xi^* L \rfloor} = \tfrac{1}{2}$, but we do not need this for the exponentially decaying power allocation since it will only affect a vanishing fraction of sections as $L$ increases. Since $\mc{E}_\ell \in [0,1]$, these sections do not affect the value of $\lim x(\tau)$ in \eqref{eq:Eell_def}.} to prove that for  $\xi \in  ( 0, 1]$,  
\be
\lim \mc{E}_{\lfloor \xi L \rfloor} (\tau) = \left\{
\begin{array}{ll}
1, & \text{ if } \xi <  \xi^*, \\
0, & \text{ if } \xi > \xi^*.
\end{array}
\right.
\label{eq:Eell_def0}
\ee
Using the relation $nR =  {L \ln M}/{\ln 2}$,
we can write
\ben
\frac{n P_{\lfloor \xi L \rfloor} }{\tau^2} = \nu_{\lfloor \xi L \rfloor} \ln M, \quad \text{ where } \quad \nu_{\lfloor \xi L \rfloor} = \frac{ L P_{\lfloor \xi L \rfloor} }{R \tau^2 \ln 2 }.
\een
From the definition of $\xi^*$ in the lemma statement and the non-increasing power-allocation, we see that $\lim  \nu_{\lfloor \xi L \rfloor} >2$ for $\xi <\xi^*$, and $\lim  \nu_{\lfloor \xi L \rfloor} < 2$ for $\xi > \xi^*$.

For brevity, in what follows we drop the superscripts on  $U_j^{\lfloor \xi L \rfloor}$, and denote it by $U_j$ for $j \in [M]$. From  \eqref{eq:Eell_def},  $\mc{E}_{\lfloor \xi L \rfloor}(\tau) $  can be written as
\begin{small}
\begin{align}
&\mc{E}_{\lfloor \xi L \rfloor}(\tau) \nonumber \\
& = \expec \left[  \frac{e^{\sqrt{ \nu_{\lfloor \xi L \rfloor} \ln M} \, U_1}}
{e^{\sqrt{ \nu_{\lfloor \xi L \rfloor} \ln M} \, U_1 }  +   M^{- \nu_{\lfloor \xi L \rfloor}} \sum_{j=2}^M e^{\sqrt{ \nu_{\lfloor \xi L \rfloor} \ln M} \, U_j} }  \right] \nonumber \\
& = \expec \, \expec \left[  \frac{e^{\sqrt{ \nu_{\lfloor \xi L \rfloor} \ln M} \, U_1 }}
{e^{\sqrt{ \nu_{\lfloor \xi L \rfloor} \ln M} \, U_1}  +   M^{- \nu_{\lfloor \xi L \rfloor}} \sum_{j=2}^M e^{\sqrt{ \nu_{\lfloor \xi L \rfloor} \ln M} \, U_j } }   \Big{|} U_1\right].
\label{eq:Eell_iter}
\end{align}
\end{small}
The inner expectation in \eqref{eq:Eell_iter} is of the form
\be
\begin{split}
&\expec \left[  \frac{e^{\sqrt{ \nu_{\lfloor \xi L \rfloor} \ln M} \, U_1}}
{e^{\sqrt{ \nu_{\lfloor \xi L \rfloor} \ln M} \, U_1 }  +   M^{- \nu_{\lfloor \xi L \rfloor}} \sum_{j=2}^M e^{\sqrt{ \nu_{\lfloor \xi L \rfloor} \ln M} \, U_j} }   \Big{|} U_1\right]  \\
& \qquad = \expec_X \left[ \frac{c}{c + X} \right], 
\label{eq:inner_exp0} 
\end{split}
\ee
where $c = \exp\left(\sqrt{ \nu_{\lfloor \xi L \rfloor} \ln M} \, U_1 \right)$ is treated as a positive constant, and the expectation is with respect to the random variable
\be
X := M^{- \nu_{\lfloor \xi L \rfloor}} \sum_{j=2}^M \exp\left(\sqrt{ \nu_{\lfloor \xi L \rfloor} \ln M} \, U_j \right).
\label{eq:Xrv_def}
\ee

\textbf{Case $1$: $\xi < \xi^*$}. Here we have $ \lim \nu_{\lfloor \xi L \rfloor} >2$.  Since $\frac{c}{c +X}$ is a convex function of $X$, applying Jensen's inequality we get
$\expec_X [  \frac{c}{c + X} ] \geq \frac{c}{c + \expec X}$.
The expectation of $X$ is
\ben
\begin{split}
\expec X &= M^{- \nu_{\lfloor \xi L \rfloor}} \sum_{j=2}^M \expec \left[e^{\sqrt{ \nu_{\lfloor \xi L \rfloor} \ln M} \, U_j } \right] \\
& \stackrel{(a)}{=} M^{- \nu_{\lfloor \xi L \rfloor}} (M-1)  M^{ \nu_{\lfloor \xi L \rfloor} /2}  \leq M^{1 -  \nu_{\lfloor \xi L \rfloor} /2},
\end{split}
\een
with $(a)$ is obtained from the moment generating function of a Gaussian random variable. Therefore,
\be
\begin{split}
1 \geq \expec_X \left[ \frac{c}{c + X} \right] \geq  \frac{c}{c + \expec X} &\geq \frac{c}{c + M^{1-  \nu_{\lfloor \xi L \rfloor} /2}} \\
&= \frac{1}{1 + c^{-1} \, M^{1-  \nu_{\lfloor \xi L \rfloor} /2}}.
\label{eq:jensen_chain}
\end{split}
\ee
Recalling that $c = \exp\left(\sqrt{ \nu_{\lfloor \xi L \rfloor} \ln M} \, U_1 \right)$, \eqref{eq:jensen_chain} implies that 
\be
\begin{split}
&\expec_X \left[ \frac{e^{\sqrt{ \nu_{\lfloor \xi L \rfloor} \ln M} \, U_1}}{e^{\sqrt{ \nu_{\lfloor \xi L \rfloor} \ln M} \, U_1 } + X}  \ \Big{|} \ U_1 \right]  \\
& \qquad \geq \frac{1}{ 1 +   M^{1-  \nu_{\lfloor \xi L \rfloor} /2} \, e^{- \sqrt{ \nu_{\lfloor \xi L \rfloor} \ln M} \, U_1 }}.
\label{eq:UU_lb}
\end{split}
\ee
When $\{ U_1 > - (\ln M)^{1/4} \}$, the RHS of \eqref{eq:UU_lb} is at least $[1 +  M^{1-  \nu_{\lfloor \xi L \rfloor} /2} \, \exp\left( (\ln M)^{3/4} \sqrt{ \nu_{\lfloor \xi L \rfloor}} \right)]^{-1}$. Using this in  \eqref{eq:Eell_iter}, we obtain that
\be
\begin{split}
& 1 \geq \mc{E}_{\lfloor \xi L \rfloor} (\tau) \\
&\geq  \frac{P(U_1 > - (\ln M)^{1/4} )}{1 +  M^{1-  \nu_{\lfloor \xi L \rfloor} /2} \, e^{ (\ln M)^{3/4} \sqrt{ \nu_{\lfloor \xi L \rfloor}} }} \stackrel{ M \to \infty}{\longrightarrow} \  1,
\end{split}
\ee
since  $\lim  \nu_{\lfloor \xi L \rfloor} > 2$.  Hence $\mc{E}_{\lfloor \xi L \rfloor} \to 1$ when $\lim  \nu_{\lfloor \xi L \rfloor} > 2$.

\textbf{Case $2$: $\xi > \xi^*$.} Here we have $ \lim \nu_{\lfloor \xi L \rfloor} < 2$. The random variable $X$ in \eqref{eq:Xrv_def} can be bounded from below as follows.
\be
\begin{split}
X &\geq  M^{- \nu_{\lfloor \xi L \rfloor}} \max_{j \in \{2, \ldots, M\}} e^{\sqrt{ \nu_{\lfloor \xi L \rfloor} \ln M} \, U_j } \\
&= M^{- \nu_{\lfloor \xi L \rfloor}} e^{ \left[ \max_{j \in \{2, \ldots, M\}} U_j \right] \sqrt{ \nu_{\lfloor \xi L \rfloor} \ln M} }.
\label{eq:X_lb0}
\end{split}
\ee
Using standard bounds for the standard normal distribution, it can be shown that
\be
P\left( \max_{j \in \{2, \ldots, M\}} U_j \ < \sqrt{2 \ln M}(1-\e)\right) \leq e^{-M^{\e(1-\e)}},
\label{eq:pmax_gauss}
\ee
for
$ \e = \omega\left( \frac{\ln \ln M}{\ln M} \right)$.\footnote{Recall that $f(n) = \omega(g(n))$ if for each $k >0$, $\abs{f(n)} / \abs{g(n)} \geq k$ for all sufficiently large $n$. } \label{eq:eps_order}
Combining \eqref{eq:pmax_gauss} and \eqref{eq:X_lb0}, we obtain that
\ben
\begin{split}
&\exp(-M^{\e(1-\e)}) \geq P\left( \max_{j \in \{2, \ldots, M\}} U_j \ < \sqrt{2 \ln M}(1-\e)\right)  \\ 
 & \quad \geq  P\left( X <  M^{- \nu_{\lfloor \xi L \rfloor}} e^{ \sqrt{2 \ln M}(1-\e) \sqrt{ \nu_{\lfloor \xi L \rfloor} \ln M}} \right) \\
 & \quad = P\left( X  <  M^{\sqrt{2  \nu_{\lfloor \xi L \rfloor}} (1-\e) - \nu_{\lfloor \xi L \rfloor}}  \right).
\end{split}
\een
Since $\lim  \nu_{\lfloor \xi L \rfloor} < 2$ and $\e >0$ can be an arbitrarily small constant, there exists a strictly positive constant $\delta$ such that 
$\delta <  \sqrt{2  \nu_{\lfloor \xi L \rfloor}}(1-\e) - \nu_{\lfloor \xi L \rfloor}$ for all sufficiently large $L$. Therefore, for sufficiently large $M$, the expectation in \eqref{eq:inner_exp0} can be bounded as
\be
\begin{split}
\expec_X \left[ \frac{c}{c +X} \right] &\leq P(X < M^\delta) \cdot 1 + P (X \geq M^{\delta})\cdot \frac{c}{c + M^{\delta}} \\
& \leq e^{-M^{\e(1-\e)}} + 1 \cdot \frac{c}{c + M^{\delta}} \leq \frac{2}{1 + c^{-1} M^{\delta}}.
\end{split}
\label{eq:ccx_bound}
\ee
Recalling that $c = \exp\left(\sqrt{ \nu_{\lfloor \xi L \rfloor} \ln M} \, U_1  \right)$, and  using the bound  of \eqref{eq:ccx_bound} in  \eqref{eq:Eell_iter}, we obtain
\be
\begin{split}
&\mc{E}_{\lfloor \xi L \rfloor} (\tau)  \leq \expec \left[  \frac{2}
{1 +   M^{\delta} e^{-\sqrt{ \nu_{\lfloor \xi L \rfloor} \ln M} \, U_1 }}  \right] \\
& \leq P(U_1 > (\ln M)^{1/4}) \cdot 2 + \frac{2P(U_1 \leq (\ln M)^{1/4})}{1 + M^{\delta} e^{-\sqrt{ \nu_{\lfloor \xi L \rfloor}} \, (\ln M)^{3/4}} } \\
&  \stackrel{(a)}{\leq} 2e^{-\tfrac{1}{2} (\ln M)^{1/2}} + \, 1 \cdot \frac{2}{1 + e^{ \delta \ln M -\sqrt{ \nu_{\lfloor \xi L \rfloor}} \, (\ln M)^{3/4}}} \\
& \stackrel{(b)}{\longrightarrow} 0 \text{ as } M \to \infty.
\end{split}
\label{eq:Eell_final}
\ee
In \eqref{eq:Eell_final}, $(a)$ is obtained using the bound  $\Phi(x) < \exp(-x^2/2)$ for $x \geq 0$, where $\Phi(\cdot)$ is the Gaussian cdf; $(b)$ holds since $\delta$ and $\lim  \nu_{\lfloor \xi L \rfloor}$ are both positive constants.

This proves that $\mc{E}_{\lfloor \xi L \rfloor}(\tau) \to 0$ when $\lim \nu_{\lfloor \xi L \rfloor} < 2$.  The proof of the lemma is complete since we have proved both statements in \eqref{eq:Eell_def0}.

\subsection{Proof of Lemma \ref{lem:lim_xt_taut}} \label{app:lim_xt_taut}

 For brevity,  let  $\xi_t:= \xi^*(\bar{\tau}_t)$ for $t \geq 0$, where $\xi^*(\cdot)$ is defined in Lemma \ref{lem:conv_expec}.  For $t=0$,  $\bar{\tau}_0^2= \sigma^2+P$.  Then, from Lemma \ref{lem:conv_expec}  we obtain
\ben
\bar{x}_1 =   \lim_{L\to \infty} \sum_{\ell=1}^{\lfloor \xi_0 L \rfloor} \frac{P_\ell}{P},
\een
where $\xi_0$ is the supremum of all $\xi \in (0,1]$ that satisfy
\be \label{eq:x10_exp}  \lim_{ L \to \infty} \, L P_{\lfloor \xi L \rfloor} = \sigma^2 (1+\snr)^{1-\xi} \ln(1+\snr)  > 2R (\sigma^2+P) \ln 2. \ee
The first equality in  \eqref{eq:x10_exp} is due to \eqref{eq:cell}. Simplifying \eqref{eq:x10_exp} yields the condition $\xi < \frac{1}{2 \mc{C}} \log (\mc{C}/R)$, from which it follows that the supremum is $\xi_0 = \tfrac{\log(\mc{C}/R)}{2\mc{C}}$.

Using the geometric series formula
$\sum_{\ell=1}^k P_{\ell} = (P+\sigma^2)(1- 2^{-2\mc{C}k/L})$, \eqref{eq:x10_exp} becomes
\ben
\begin{split}
\bar{x}_1 = \lim_{L \to \infty} \sum_{\ell =1}^{\lfloor \xi_0 L \rfloor} \frac{P_\ell}{P} &= \frac{P+\sigma^2}{P}(1 - 2^{-2\mc{C} \xi_0}) \\
&=
\frac{ (1+ \snr) - (1+ \snr)^{1- \xi_0}}{\snr}.
\end{split}
\een
The expression for $\bar{\tau}^2_1$ is a straightforward simplification of $\sigma^2 + P(1-\bar{x}_1)$. 

Assume towards induction that \eqref{eq:limxt1} and \eqref{eq:limtaut1} hold for $\bar{x}_t, \bar{\tau}^2_t$. For step $(t+1)$,  from Lemma  \ref{lem:conv_expec},
\ben
\bar{x}_{t+1} = \lim_{L\to \infty} \sum_{\ell=1}^{\lfloor \xi_t L \rfloor} \frac{P_\ell}{P},
 \een
where $\xi_t$ is the supremum of all $\xi \in (0,1]$ that satisfy
\be \label{eq:x1t_exp}  \lim_{ L \to \infty} \, L P_{\lfloor \xi L \rfloor} = \sigma^2 (1+\snr)^{1-\xi} \ln(1+\snr)  > 2R \bar{\tau}_t^2 \ln 2. \ee
Using  the expression in \eqref{eq:limtaut1} for $\bar{\tau}_t^2$ (due to the induction hypothesis) and simplifying \eqref{eq:x1t_exp} yields the condition 
\[\xi < \xi_{t-1} + \frac{1}{2\mc{C}} \log_2 \frac{\mc{C}}{R}.  \]
Hence the supremum is $\xi_t = \xi_{t-1} + \frac{1}{2\mc{C}} \log_2 ( \mc{C}/{R})$. It follows that 
\be
\begin{split}
\bar{x}_{t+1} =   \lim_{L \to \infty} \sum_{\ell =1}^{\lfloor \xi_t L \rfloor} \frac{P_\ell}{P} &= \frac{P+\sigma^2}{P}(1 - 2^{-2\mc{C} \xi_t}) \\
&=
\frac{ (1+ \snr) - (1+ \snr)^{1- \xi_t}}{\snr}.
\label{eq:xt1_final}
\end{split}
\ee
The proof is concluded by using \eqref{eq:xt1_final} to compute $\bar{\tau}^2_{t+1} = P + \sigma^2(1 - \bar{x}_{t+1})$.

\subsection{The limit of  $\frac{1}{n} \mathbb{E}\{[\eta^r(\beta - \bar{\tau}_r \breve{Z}_r) - \beta]^*  [\eta^s(\beta - \bar{\tau}_s \breve{Z}_s) - \beta] \}$ equals $\bar{\sigma}^2_{s+1}$ for $-1 \leq r \leq s \leq t$.} \label{app:qrqs_lim}

Noting that $\norm{\beta}^2=nP$, we prove that the desired limit
\be
\begin{split}
&  \lim\left[ \frac{1}{n} \expec\{ [\eta^r(\beta - \bar{\tau}_r \breve{Z}_r)]^* [\eta^s(\beta - \bar{\tau}_s \breve{Z}_s)] \} \right.\\
& \left. -   \frac{1}{n} \expec\{   \beta^{*}\eta^r(\beta - \bar{\tau}_r \breve{Z}_r) \} -  \frac{1}{n} \expec\{   \beta^{*}\eta^s(\beta - \bar{\tau}_s \breve{Z}_s) \}  + P \right]
\end{split}
\label{eq:expec_expand} 
\ee
equals $\bar{\sigma}^2_{s+1}=\sigma^2\left( (1+\snr)^{1-\xi_s} -1 \right)$.  For the case $r= s= -1$ the result holds since $\bar{\sigma}_0^2 = P$, so assume $s > -1$. To obtain  \eqref{eq:expec_expand}, we show the following: for $0 \leq r \leq t$,
\be
\lim \frac{1}{n} \expec\{   \beta^{*}\eta^r(\beta - \bar{\tau}_r \breve{Z}_r) \} =  \bar{\tau}_0^2 - \bar{\tau}_{r+1}^2,  \label{eq:lim_beta_eta}
\ee
and for $0 \leq r \leq s \leq t$,
\be
 \lim \frac{1}{n} \expec\{ [\eta^r(\beta - \bar{\tau}_r \breve{Z}_r)]^* [\eta^s(\beta - \bar{\tau}_s \breve{Z}_s)] \}  = \bar{\tau}_0^2 - \bar{\tau}_{r+1}^2.
 \label{eq:lim_etar_etas}
\ee
The above results are all trivially true if $r = -1$.

We first show \eqref{eq:lim_beta_eta}.  Since $\beta$ is distributed uniformly over the set $\mcb$, the expectation in \eqref{eq:lim_beta_eta} can be computed by assuming that  $\beta$ has a non-zero in the first entry of each section. Thus
\begin{align}
& \lim \frac{1}{n} \expec\{   \beta^{*}\eta^r(\beta - \bar{\tau}_r \breve{Z}_r) \} \nonumber \\
 &=  \lim \sum_{\ell=1}^L P_\ell  \,  \expec \left[ \frac{ e^{ \frac{n P_\ell}{\bar{\tau}_r^2}}  e^{\frac{\sqrt{n P_\ell}}{\bar{\tau}_r} \, U_1  }}
{ e^{ \frac{n P_\ell}{\bar{\tau}_r^2} }   e^{\frac{\sqrt{n P_\ell}}{\bar{\tau}_r} \, U_1}  + \sum_{j=2}^M e^{\frac{\sqrt{n P_\ell}}{\bar{\tau}_r}U_j}} \right] \nonumber \\
& \stackrel{(a)}{=} \lim \sum_{\ell =1}^{\lfloor \xi_r L \rfloor} P_\ell  = \sigma^2\left( (1 + \snr)  - (1 + \snr)^{1-\xi_r} \right)  \nonumber \\
& \stackrel{(b)}{=}  \bar{\tau}_0^2 - \bar{\tau}_{r+1}^2. \label{eq:lim_beta_etar}
\end{align}
 In \eqref{eq:lim_beta_etar}, $\{U^\ell_j \}$ with $\ell \in [L], j \in [M]$ is a relabeled version of $-\breve{Z}_r$, and is thus i.i.d.\ $\mc{N}(0,1)$.  Equalities $(a)$ and $(b)$ are obtained from Lemmas \ref{lem:conv_expec} and \ref{lem:lim_xt_taut} (cf.\ \eqref{eq:xt_tau_def}, \eqref{eq:xt_tau_def1}, and \eqref{eq:limtaut1}).

Consider result \eqref{eq:lim_etar_etas}.  From the proof of Proposition \ref{prop:se_cons}, (noting that $\beta^{r+1} = \eta^r(\beta - \bar{\tau}_r \breve{Z}_r)$ and cf.\ \eqref{eq:betat_beta_sq} and \eqref{eq:betat_sq}), it follows that 
\be
\frac{1}{n} \expec\norm{\eta^r_\ell(\beta - \bar{\tau}_r \breve{Z}_r)}^2  = \frac{1}{n} \expec\{   \beta_\ell^{*}\eta^r_\ell(\beta - \bar{\tau}_r \breve{Z}_r) \}, \quad \ell \in [L],
 \label{eq:lim_etar_sq}
\ee
which proves the result if $r=s$.  For  $r < s$, we obtain the result by showing that 
\begin{align}
\lim \frac{1}{n} \expec\{  [\eta^r(\beta - \bar{\tau}_r \breve{Z}_r)]^*  [\eta^s(\beta - \bar{\tau}_s \breve{Z}_s)] \} &\leq\lim \sum_{\ell =1}^{\lfloor \xi_r L \rfloor} P_\ell , \label{eq:etars_ub} \\
\lim \frac{1}{n} \expec\{  [\eta^r(\beta - \bar{\tau}_r \breve{Z}_r)]^*  [\eta^s(\beta - \bar{\tau}_s \breve{Z}_s)] \} &\geq \lim \sum_{\ell =1}^{\lfloor \xi_r L \rfloor} P_\ell . \label{eq:etars_lb}
\end{align}
We then we get the desired result by observing that the limit on the RHS above  equals $\bar{\tau}_0^2 - \bar{\tau}_{r+1}^2$, as in \eqref{eq:lim_beta_etar}. From the Cauchy-Schwarz inequality, we have
\be
\begin{split}
& \lim  \frac{1}{n} \expec\{  [\eta^r(\beta - \bar{\tau}_r \breve{Z}_r)]^*  [\eta^s(\beta - \bar{\tau}_s \breve{Z}_s)] \} \\
&=  \lim \frac{1}{n} \sum_{\ell=1}^L  \expec\{  [\eta^r_\ell(\beta - \bar{\tau}_r \breve{Z}_{r})]^*  [\eta^s_\ell( \beta - \bar{\tau}_s \breve{Z}_{s} )] \} \\ 
  & \stackrel{(a)}{\leq}  \lim \frac{1}{n}\sum_\ell (\mathbb{E}\norm{\eta^r_\ell(\beta - \bar{\tau}_r \breve{Z}_{r} )}^2)^{1/2} 
  (\mathbb{E} \norm{\eta^s_\ell(\beta - \bar{\tau}_s \breve{Z}_{s} )}^2)^{1/2}  \\
& \stackrel{(b)}{=}  \lim \sum_{\ell} P_\ell \,  \mc{E}_\ell(\bar{\tau}_r^2) \,  \mc{E}_\ell(\bar{\tau}_s^2)
\stackrel{(c)}{=}  \, \lim \sum_{\ell =1}^{\lfloor \xi_r L \rfloor} P_\ell, 
\end{split}
\label{eq:etar_etas_ub}
\ee
where $(a)$ is obtained using the Cauchy-Schwarz inequality; $(b)$ follows from \eqref{eq:lim_etar_sq}, \eqref{eq:lim_beta_etar}, and the definition of $\mc{E}_\ell(\cdot)$ in \eqref{eq:Eell_def}; $(c)$ is obtained as follows. Consider $ \mc{E}_{\lfloor \xi L \rfloor}(\bar{\tau}_r^2)$ and $\mc{E}_{\lfloor \xi L \rfloor}(\bar{\tau}_s^2)$  for some $\xi \in (0,1]$. It follows from the proofs of Lemmas \ref{lem:conv_expec} and \ref{lem:lim_xt_taut} that, 
\ben
\lim \mc{E}_{\lfloor \xi L \rfloor} (\bar{\tau}_r) = \left\{
\begin{array}{ll}
1, & \text{ for } \xi <  \xi_r, \\
0, & \text{ for } \xi > \xi_r,
\end{array}
\right. 
\een
and
\ben
\lim \mc{E}_{\lfloor \xi L \rfloor} (\bar{\tau}_s) = \left\{
\begin{array}{ll}
1, & \text{ for } \xi <  \xi_s, \\
0, & \text{ for } \xi > \xi_s,
\end{array}
\right.
\een
where $\xi_r, \xi_s$ are as defined in Lemma \ref{lem:lim_xt_taut}. Since $r < s$, we have $ \xi_r \leq  \xi_s$, which yields $(c)$ in \eqref{eq:etar_etas_ub}.

For the lower bound \eqref{eq:etars_lb}, since $\beta$ is distributed uniformly over the set $\mcb$, the expectation in \eqref{eq:lim_etar_etas} can be computed by assuming that $\beta$ has a non-zero in the first entry of each section:
\be
\begin{split}
&\frac{1}{n} \expec\{  [\eta^r(\beta - \bar{\tau}_r \breve{Z}_r)]^*  [\eta^s(\beta - \bar{\tau}_s \breve{Z}_s)] \}  \\
&= \frac{1}{n}  \sum_{\ell} \mathbb{E}\{ [\eta^r_{\ell}(\beta - \bar{\tau}_r \breve{Z}_{r}) ]^*  [\eta^s_{\ell}(\beta - \bar{\tau}_s \breve{Z}_{s})]  \}
= \sum_{\ell} P_\ell \, \mc{E}_{rs,\ell}
 \label{eq:Ers_def}
 \end{split}
\ee
where
\be
\begin{split}
 & \mc{E}_{rs,\ell} = \\
 & \expec 
 \frac{ e^{\mathsf{b}^2_{r_\ell} + \mathsf{b}_{r_\ell} U^\ell_{r1}} \, e^{\mathsf{b}^2_{s_\ell} + \mathsf{b}_{s_\ell} U^\ell_{s1}} + \sum_{i=2}^M e^{\mathsf{b}_{r_\ell} U^\ell_{ri}} \,e^{\mathsf{b}_{s_\ell} U^\ell_{si} }}
{ (e^{\mathsf{b}^2_{r_\ell} + \mathsf{b}_{r_\ell} U^\ell_{r1} } +   \sum_{j=2}^M e^{\mathsf{b}_{r_\ell} U^\ell_{rj}} ) (e^{\mathsf{b}^2_{s_\ell}+ \mathsf{b}_{s_\ell} U^\ell_{s1} } +   \sum_{j=2}^M e^{\mathsf{b}_{s_\ell} U^\ell_{sj}})} 
\end{split}
\label{eq:Ersell_def}
\ee
with $\mathsf{b}^2_{r_\ell} := n P_\ell/\bar{\tau}_r^2$ and $\mathsf{b}^2_{s_\ell} := n P_\ell/\bar{\tau}_s^2$.  In \eqref{eq:Ersell_def}, the pairs of random variables $\{(U^\ell_{rj},  U^\ell_{sj} ) \}, \, j\in [M]$ are i.i.d.\ across index $j$, and for each $j$, $U^\ell_{rj}$ and $U^\ell_{sj}$ are jointly Gaussian with $\mc{N}(0,1)$ marginals and covariance $\bar{\tau}_s/\bar{\tau}_r$.

Consider the expectation using just the first term in the numerator on the right-hand side of \eqref{eq:Ersell_def}.  This can be written as
\be
\begin{split}
& \expec\Bigg[ \expec\Bigg[\Bigg( \frac{ e^{\mathsf{b}_{r_\ell} U^\ell_{r1}}  }
{ e^{\mathsf{b}_{r_\ell} U^\ell_{r1}}  +  \sum_{j =2}^M  e^{\mathsf{b}_{r_\ell} U^\ell_{rj}  -\mathsf{b}^2_{r_\ell}}} \Bigg) \\
&\qquad \cdot  \Bigg( \frac{e^{\mathsf{b}_{s_\ell} U^\ell_{s1}}  }
{ e^{\mathsf{b}_{s_\ell} U^\ell_{s1}}  + \sum_{j =2}^M  e^{\mathsf{b}_{s_\ell} U^\ell_{sj} -\mathsf{b}^2_{s_\ell}}}  \Bigg) \Big{|} \ U^\ell_{r1}, U^\ell_{s1} \Bigg] \Bigg] \\ 
& \stackrel{(a)}{\geq}  \expec\left[ \Bigg( \frac{ e^{\mathsf{b}_{r_\ell} U^\ell_{r1}}  } { e^{ \mathsf{b}_{r_\ell} U^\ell_{r1}}  +  M e^{-\frac{1}{2}\mathsf{b}_{r_\ell}^2} } \Bigg)
 \Bigg( \frac{ e^{\mathsf{b}_{s_\ell} U^\ell_{s1}}  } { e^{ \mathsf{b}_{s_\ell} U^\ell_{s1}}  +  M e^{-\frac{1}{2}\mathsf{b}_{s_\ell}^2} } \Bigg) \right]\\
& = \expec\left[ \left( 1 +  M e^{ -\frac{\mathsf{b}^2_{r_\ell}}{2} -\mathsf{b}_{r_\ell} U^\ell_{r1}} \right)^{-1}
 \left( 1 +  M e^{ -\frac{\mathsf{b}^2_{s_\ell}}{2} -\mathsf{b}_{s_\ell} U^\ell_{s1}} \right)^{-1} \right]\\
 & \geq P\left(U^\ell_{r1} > -\mathsf{b}_{r_\ell}^{1/2}, \  U^\ell_{s1} > -\mathsf{b}_{s_\ell}^{1/2} \right)  \left( 1 +  M e^{ -\frac{\mathsf{b}^2_{r_\ell}}{2} + \mathsf{b}_{r_\ell}^{3/2}}  \right)^{-1} \\
& \qquad \cdot \left( 1 +  M e^{ -\frac{\mathsf{b}^2_{s_\ell}}{2} + \mathsf{b}_{s_\ell}^{3/2}} \right)^{-1} 
 \\
& \stackrel{(b)}{\longrightarrow}1 \ \text{ as } \ M \to \infty \ \text{ for } 1\leq \ell <  \lfloor \xi_r L \rfloor.
\end{split}
\label{eq:1st_term}
\ee
In \eqref{eq:1st_term}, $(a)$ is obtained as follows. The inner expectation on the first line of the form $\expec_{X,Y}[f(X,Y)]$ with
$f(X,Y) = \tfrac{\kappa_1}{\kappa_1 + X} \cdot  \tfrac{\kappa_2}{\kappa_2 + Y}$, 
where $\kappa_1, \kappa_2$ are positive constants. Since $f$ is a convex function of $(X,Y)$, Jensen's inequality implies
$\expec[f(X,Y)] \geq f(\expec X, \expec Y)$, with 
$\expec [ \exp(\mathsf{b}_{r_\ell} U^\ell_{rj}) ] = \exp(\frac{1}{2} \mathsf{b}^2_{r_\ell})$.  

To obtain the convergence in step $(b)$ of \eqref{eq:1st_term},   note that for $\ell < \lfloor \xi_r L \rfloor$,  
\be
\begin{split}  
 \lim  \frac{\mathsf{b}^2_{r_\ell} }{2 \ln M}  =    \lim \frac{n P_\ell}{2\bar{\tau}_r^2 \ln M} & >  \lim \frac{n P_ {\lfloor \xi_r L \rfloor } }{2\bar{\tau}_r^2 \ln M} \\
&=  
 \lim \frac{L P_ {\lfloor \xi_r L \rfloor } }{2R \bar{\tau}_r^2 \ln 2} = 1, \label{eq:brl_lim} 
 \end{split}
 \ee
where we have used  $nR = L \log M$ and the fact that $\xi_r$ is the supremum of $\xi \in (0,1]$ for which $L P_ {\lfloor \xi_r L \rfloor} > 2R \bar{\tau}_r^2 \ln 2$ (see proof of Lemma \ref{lem:lim_xt_taut}).

Since $\mc{E}_{rs,\ell}$ in \eqref{eq:Ersell_def} lies in $[0,1]$ for all $\ell$, \eqref{eq:1st_term} implies that 
$ \lim \mc{E}_{rs,\ell} =1$ for $1 \leq \ell <  \lfloor \xi_r L \rfloor$. 
Using this in \eqref{eq:Ers_def} gives the lower bound \eqref{eq:etars_lb}.  Together with the upper bound in \eqref{eq:etars_ub}, this proves \eqref{eq:lim_etar_etas}, and hence completes the proof.

\subsection*{Acknowledgement}
The authors thank A. Barron and S. Cho for several insightful discussions, and the anonymous reviewers and the associate editor for their helpful comments. This work was supported in part by a  Marie Curie Career Integration Grant (Grant Agreement No. 631489). A. Greig was supported by an EPSRC Doctoral Training Award.

\IEEEtriggeratref{9}


\end{document}